\documentclass[
  acmsmall, screen,
  appendix, authorversion, ]{acmart}

\RequirePackage[T1]{fontenc}
\RequirePackage[utf8]{inputenc}

\usepackage{appendix}

\RequirePackage{./style}

\setcopyright{cc}
\setcctype{by}
\acmDOI{10.1145/3776648}
\acmYear{2026}
\acmJournal{PACMPL}
\acmVolume{10}
\acmNumber{POPL}
\acmArticle{6}
\acmMonth{1}
\received{2025-07-10}
\received[accepted]{2025-11-06}

\setcitestyle{nosort}

\title{RapunSL: Untangling Quantum Computing with Separation, Linear Combination and Mixing}

\author{Yusuke Matsushita}
\authornote{The first two authors contributed equally to this work.}
\orcid{0000-0002-5208-3106}
\affiliation{\institution{Kyoto University}
  \city{Kyoto}
  \country{Japan}
}
\email{ymat@fos.kuis.kyoto-u.ac.jp}

\author{Kengo Hirata}
\authornotemark[1]
\orcid{0009-0005-4416-2655}
\affiliation{\institution{University of Edinburgh}
  \city{Edinburgh}
  \country{United Kingdom}
}
\affiliation{\institution{Kyoto University}
  \city{Kyoto}
  \country{Japan}
}
\email{k.hirata@sms.ed.ac.uk}

\author{Ryo Wakizaka}
\orcid{0000-0001-8762-9335}
\affiliation{\institution{Kyoto University}
  \city{Kyoto}
  \country{Japan}
}
\email{wakizaka@fos.kuis.kyoto-u.ac.jp}

\author{Emanuele D'Osualdo}
\orcid{0000-0002-9179-5827}
\affiliation{\institution{University of Konstanz}
  \city{Konstanz}
  \country{Germany}
}
\email{emanuele.dosualdo@uni-konstanz.de}

\ccsdesc[500]{Theory of computation~Separation logic}
\ccsdesc[300]{Theory of computation~Quantum computation theory}

\keywords{quantum separation logic, entanglement, quantum program verification}
 
\begin{document}

\begin{abstract}
Quantum Separation Logic (QSL) has been proposed as an effective tool
to improve the scalability of deductive reasoning for quantum programs.
In QSL, separation is interpreted as \emph{disentanglement},
and the frame rule brings a notion of \emph{entanglement-local} specification
(one that only talks about the qubits entangled with those acted upon by the program).
In this paper, we identify two notions of locality unique to the quantum domain,
and we construct a novel quantum separation logic, \ourlogic,
which is able to soundly reduce reasoning about superposition states
to reasoning about pure states (\emph{basis-locality}),
and reasoning about mixed states arising from measurement
to reasoning about pure states (\emph{outcome-locality}).
To do so, we introduce two connectives, linear combination and mixing,
which together with separation provide a dramatic improvement
in the scalability of reasoning,
as we demonstrate on a series of challenging case studies.
\end{abstract}

 \maketitle

\section{Introduction}
\label{sect:intro}

\begin{mathfig}
  \begin{proofrules}
    \infer*[lab = hoare-frame]{
      \hoare{\Prop}{\Cmd}{\PropB}
    }{
      \hoare{\Prop * \PropC}{\Cmd}{\PropB * \PropC}
    }
    \label{rule:hoare-frame}

    \infer*[lab = hoare-sum]{
      \hoare{\Prop}{\Cmd}{\PropB}
      \\
      \hoare{\Prop'}{\Cmd}{\PropB'}
    }{
      \hoare{\Prop + \Prop'}{\Cmd}{\PropB + \PropB'}
    }

    \infer*[lab = hoare-mix]{
      \hoare{\Prop}{\Cmd}{\PropB}
      \\
      \hoare{\Prop'}{\Cmd}{\PropB'}
    }{
      \hoare{\Prop \mix{\hvar} \Prop'}{\Cmd}{\PropB \mix{\hvar} \PropB'}
    }
    \label{rule:hoare-mix}
  \end{proofrules}
  \caption{The three rules of \ourlogic{} embodying the three locality principles.}
  \label{fig:locality-rules}
  \Description{}
\end{mathfig}

Separation logic~\cite{OHearnP99-bi} represents a foundational breakthrough in the scalability of deductive verification methods.
By introducing a resource-based form of \emph{local} reasoning,
separation logic allows program specifications to describe behaviour
solely in terms of the resources a program fragment accesses,
without reference to the rest of the program state~\cite{IshtiaqO01-bi, OHearnRY01-sl, Reynolds02-sl}.
This is formalized through the \emph{frame rule}
(\ref{rule:hoare-frame} in \cref{fig:locality-rules}),
which guarantees that verified properties of a program fragment remain valid when executed in a larger context,
provided the context's resources are \emph{separate}.
By giving different interpretations to the notion of ``separable resources'',
separation logic has been shown to support a wide array of different computational phenomena, from the heap to concurrency~\cite{OHearn07-csl, Brookes04-csl, BrookesO16-csl-summary}, to probabilistic non-determinism~\cite{BartheHL19-probsep, BaoDHS21-bunched-probsep, LiAH23-probsep, LiAJPAH24-nominal-probsl, BaoDF25-bluebell}.

Recently, the separation logic paradigm has been applied
to \emph{quantum computation}~\cite{ZhouBHYY21-qsl, LeLSS22-qsl, DengWX24-qsl, SuZFY24-qsl}.
In the quantum domain, separation has been interpreted
as \emph{disentanglement}:
roughly speaking,
$P * Q$ describes a configuration where the qubits described by $P$
are not entangled with the ones in $Q$.
Since a quantum program acting on some qubits $\bar{\qit}$ would not affect
any disentangled qubits~$\bar{\qitB}$,
the frame rule is valid, improving the modularity of reasoning.
For example, in proving
$
  \hoare
    {{\qit \mapsto \qket * \qitB \mapsto \qketB}}
    {\Unitary[x]}
    {{\qit \mapsto \Unitary \qket * \qitB \mapsto \qketB}}
$
one can apply \ref{rule:hoare-frame} framing $\qitB \mapsto \qketB$
and reduce reasoning to the simpler
$
  \hoare
    {{\qit \mapsto \qket}}
    {\Unitary[x]}
    {{\qit \mapsto \Unitary \qket}}
$
as one would do in an informal proof.
This allows one to focus only on the relevant state for a proof of a component,
and still reuse the specification when used in larger contexts.
We dub this \emph{entanglement-locality}, as it allows one to focus only on the entangled part.

The starting observation of this paper is that
quantum computation adds two unique ways to add context to some computation,
\emph{superposition} and \emph{measurement},
and each of them grants a fundamentally new notion of locality which
cannot be captured by separation alone.

\emph{Basis-locality.}
A qubit's state is, in general,
a superposition $\alpha \ket{0} + \beta \ket{1}$
of the two classical states~$\ket{0}$ and $\ket{1}$ (the basis vectors).
Quantum gates, the basic building blocks of quantum computation,
act on qubits as unitary operators~$\Unitary$:
their effect on a state $\alpha \ket{0} + \beta \ket{1}$ is entirely determined
by their effect on $\ket{0}$ and $\ket{1}$,
\ie $
  \Unitary(\alpha \ket{0} + \beta \ket{1}) =
    \alpha(\Unitary \ket{0}) + \beta(\Unitary \ket{1}).
$
This suggests, in addition to disentanglement, a second notion of locality,
which we call \emph{basis-locality},
which would allow us to focus on the effect of a program on the basis states,
and extrapolate its effect to a superposition.

\emph{Outcome-locality.}
In the quantum world, the act of measuring is a delicate affair:
measuring a qubit has the very global effect of making its superposition state
probabilistically collapse to a classical state.
The state after a measurement can thus be described using a
so-called \emph{mixed state}, \ie a probabilistic ensemble of pure states.
The effect of a program continuing after a measurement, however,
would be entirely determined by its effect on each of the possible outcomes.
This suggests a third notion of locality,
which we call \emph{outcome-locality},
where a specification on (potentially pure) states can be lifted
to a specification on mixed states.

Unfortunately, none of the quantum logics in the literature has achieved all three locality principles within one logic.
The quantum separation logics of \citet{ZhouBHYY21-qsl} and \citet{LeLSS22-qsl} support only entanglement-locality.
Recent work by \citet{DengWX24-qsl} additionally supports outcome-locality, but not basis-locality.
This is not for lack of imagination:
these logics are built on models that are fundamentally incompatible
with basis-locality or outcome-locality.

What we set out to find in this paper is a way to soundly enable all these
three locality principles in a program logic for quantum computation,
and to articulate the patterns of reasoning they unlock.
Due to the extremely subtle interaction of all these three mechanisms---(dis)entanglement, superposition, and measurement---the natural models of assertions all fail.
Our solution starts from an analysis of such failures.
We identify two sources of incompatibility with the three locality principles
in the models used in the literature.
The first pertains to the level of abstraction of assertions on pure states.
The logics of \citet{ZhouBHYY21-qsl} and \citet{Ying11-qhl}
use global-phase-insensitive assertions,
which is natural considering that measurements
are insensitive to global-phase changes.
This choice, however, impedes basis-locality:
the sum of vectors needed to form a superposition
is only meaningful when the global phase is tracked.

The second, and more serious, incompatibility
involves the treatment of measurements.
The issue is that measurements do not fulfil, strictly speaking,
the basis-locality principle, as their effect is \emph{not}
fully determined by their effect on basis states.
On a classical state, a measurement is a no-op.
However, when applied to a state in a superposition, the measurement
collapses it, probabilistically, to a classical state.
On the face of it, this seems to suggest that there is a fundamental
incompatibility between measurements and basis-locality.
In fact, this is indeed the case in the
(global-phase-sensitive) model of \citet{LeLSS22-qsl}.

The main contribution of this paper
is to show that this apparent conflict can be resolved
by introducing a new quantum separation logic called \ourlogic.
We identify the handling of mixed states as the source of the conflict:
In the logic of \citet{LeLSS22-qsl},
assertions talk, just like in any traditional separation logic,
about \emph{single} pure states, so specifications can only assert facts that apply to \emph{every outcome}.
In \ourlogic, assertions are global-phase-sensitive and predicate over
the whole mixed state at once.
More precisely, \ourlogic{} includes three logical connectives:
\emph{separation}~($P * Q$) representing disentanglement,
\emph{sum}~($P + Q$) representing superposition, and
\emph{mixing}~($P \mix{\hvar} Q$) joining two outcomes into a mixed state.
The conflict is resolved because in our model,
a measurement on a classical state still gives us a mixed state
where one outcome is a degenerate zero-probability state,
and is therefore distinguishable from a no-op.
The three locality principles are then embodied by the three rules
of~\cref{fig:locality-rules}:
\ref{rule:hoare-frame} for entanglement-locality,
\ref{rule:hoare-sum} for basis-locality, and
\ref{rule:hoare-mix} for outcome-locality.

\paragraph{Contributions}
The main contributions of this paper are:
\begin{enumerate}
  \item A new \emph{assertion language} to describe mixed states using
        separation, sum and mixing, and its supporting model.
  \item A thorough study of the rich interactions between the new connectives.
        In particular, we study the distributivity and interchange properties
        that hold in our model.
  \item A sound program logic, \ourlogic,
        supporting all three locality principles.
  \item Case studies evaluating how effective local reasoning is in \ourlogic,
        and illustrating the new proof patterns available in it.
\end{enumerate}

\paragraph{Outline}
We start with an informal overview~\cref{sect:overview} explaining the key ideas.
Then, after setting things up in \cref{sect:lang}, we formalize our logic and prove its soundness in \cref{sect:logic}.
We present case studies in \cref{sect:cases} and discuss key topics in \cref{sect:discussion}.
Finally, \cref{sect:related} reviews related work and \cref{sect:concl} concludes with future work.
All omitted details and proofs can be found in \appendixref{}\ifnoappendix{} in \cite{MatsushitaHirataWD25-RapunSL-arXiv}\fi.
 \section{Overview of \ourlogic}
\label{sect:overview}

In this section, we highlight the key ideas
unlocking the development of \ourlogic.

\subsection{Handling Superposition, Compositionally}
\label{sect:overview:superposition}

To motivate the problem with superposition,
imagine we are given the task of designing a quantum circuit~$\Cmd$
that should implement some Boolean 1-to-1 function
$f \from \{0, 1\}^n \to \{0, 1\}^n$
of $n$ qubits~$\bar{\qit}$,
  using basic quantum gates,
\eg Toffoli gates.
Sometimes, some auxiliary qubits are necessary to do so.
A well-understood technique for managing auxiliary state
is to employ so-called \emph{dirty qubits}~\cite{HanerRS17-factoring-using, Gidney18-halving-cost, NieZS24-quantum-circuit}:
a qubit~\qitTmp{} that serves as auxiliary workspace,
and is not assumed to be in any particular state (hence ``dirty'').
As long as \qitTmp{} is returned to its original state once the computation
is done, even if it was in a superimposed state with other
qubits in the context, the only state change would be in $\bar{\qit}$.

Informally, one would reason about $\Cmd$ as follows.
Since $\Cmd$ is implemented with linear operators, it is itself linear.
Therefore, it is sufficient to analyse its behaviour on \emph{classical} states
for $\bar{\qit}$ and $\qitTmp$, prove that $\bar{\qit}$ is correctly transformed
as dictated by $f$, and that $\qitTmp$ is restored to its input state.
Then, by linearity, any state in a superposition would simply see
the operator~$U_f$ (the unitary lifting of $f$)
applied to $\bar{\qit}$, and the rest of the qubits left untouched.

More formally, the argument would start by proving the correctness
of $\Cmd$ on a classical state, obtaining the following triple:
\begin{equation}
    \all{\bar{b} \in \{0, 1\}^n}
    \all{t \in \{0, 1\}}
     \hoare
       {\bar{\qit} \mapsto \ket{\smashbar{b}} * \qitTmp \mapsto \ket{t}}
       {\Cmd}
       {\bar{\qit} \mapsto \ket{f(\smashbar{b})} * \qitTmp \mapsto \ket{t}}
  \label{spec:dirty-classical}
\end{equation}
The challenge would then be to reuse the triple in a context where
there are some other qubits~$\bar{\qitB}$,
and all the qubits are in some arbitrary superimposed state~$\qket$,
\ie to deduce from~\cref{spec:dirty-classical} the triple:
\begin{equation}
  \hoare
    {(\bar{\qit}, \qitTmp, \bar{\qitB}) \mapsto \qketB}
    {\Cmd}
    {(\bar{\qit}, \qitTmp, \bar{\qitB}) \mapsto
     (U_{f} \otimes \gid_{\qitTmp, \bar{\qitB}})
     \qketB}
  \label{spec:dirty-superimp}
\end{equation}
To perform such a deduction,
\ourlogic{} provides two rules encoding the linearity argument:
\begin{proofrules}
  \infer*[lab = hoare-sum]{
    \hoare{\Prop}{\Cmd}{\PropB}
    \\
    \hoare{\Prop'}{\Cmd}{\PropB'}
  }{
    \hoare{\Prop + \Prop'}{\Cmd}{\PropB + \PropB'}
  }
  \label{rule:hoare-sum}

  \infer*[lab = hoare-scale]{
    \hoare
      {\Prop}
      {\Cmd}
      {\PropB}
  }{
    \hoare
      {\scl{\alpha}{\Prop}}
      {\Cmd}
      {\scl{\alpha}{\PropB}}
  }
  \label{rule:hoare-scale}
\end{proofrules}
The rules use the logical connectives~$ \Prop + \Prop' $
and~$ \scl{\alpha}{\Prop} $
representing the sum and scaling of quantum states respectively,
\ie \ourlogic{} defines them so that
$
  (\qit \mapsto \qket) + (\qit \mapsto \qketB)
    \dashvdash
  \qit \mapsto (\qket + \qketB)
$
and
$
  \scl{\alpha}{(\qit \mapsto \qket)}
    \dashvdash
  \qit \mapsto \alpha \qket
$ are valid.\footnote{We use $\vdash$ for entailment, and $\dashvdash$ for bidirectional entailment, \ie logical equivalence.
}\footnote{Interestingly, we can derive \ref{rule:hoare-scale} from the frame rule \ref{rule:hoare-frame},
  because $\scl{\alpha}{\Prop}$
  can be represented as $(() \mapsto \alpha) \, *\, \Prop$,
  where $() \mapsto \alpha$ is a zero-qubit state of a one-dimensional vector (\ie a scalar) $\alpha$.
}

Using such rules, it is possible to prove \cref{spec:dirty-superimp}
by first seeing $\qket$ as an explicit superposition of classical states:
$ \qket = \sum_{\bar{b}, t, \bar{c}} \alpha_{\bar{b}, t, \bar{c}} \ket{\smashbar{b}\pk1 t\pk1 \bar{c}} $.
Then by \ref{rule:hoare-sum} and \ref{rule:hoare-scale}, we reduce the goal
to a triple on a classical state; the state of $\bar{\qitB}$ can now simply be framed
to reduce the problem to our original \cref{spec:dirty-classical}.
Note how the framing of $\bar{\qitB}$ is only possible \emph{after} having applied
\ref{rule:hoare-sum}: in the original state, $\bar{\qit}$, $\qitTmp$ and $\bar{\qitB}$ are
arbitrarily entangled.

Perhaps surprisingly, no logic in the literature can perform the
intuitive steps above.
Prior logics, in fact, cannot soundly admit \ref{rule:hoare-sum}.
The logics of \citet{ZhouBHYY21-qsl} and \citet{SuZFY24-qsl}
are incompatible with it because they adopt a \emph{global-phase-insensitive}
semantics for their assertions
(\eg by representing states as a density matrix).
This choice seems justified as no measurement can distinguish between
two states that differ only in the global phase.
This is, however, fundamentally incompatible with \ref{rule:hoare-sum}.
For example, in a global-phase-insensitive logic,
the two states $\qit \mapsto \ket{1}$ and $\qit \mapsto - \ket{1}$ are indistinguishable (\ie $\qit \mapsto \ket{1} \dashvdash \qit \mapsto - \ket{1} $).
If sum were to preserve entailment here,
we could derive
$
  {\qit \mapsto 0}
  \vdash
  ({\qit \mapsto \onehalf \ket{1}}
    +
    {\qit \mapsto - \onehalf \ket{1}})
  \vdash
  ({\qit \mapsto \onehalf \ket{1}}
    +
    {\qit \mapsto \onehalf \ket{1}})
  \vdash
  {\qit \mapsto \ket{1}}
$---a meaningful state out of an impossible probability-zero state.

Retaining the global-phase information in \ourlogic{}
allows for the compositional treatment
of some deductions, like the ones using linearity.
The logic of \citet{LeLSS22-qsl} is the only other global-phase-sensitive
quantum separation logic we are aware of.
This logic is also fundamentally incompatible with \ref{rule:hoare-sum}
because of its model of \emph{measurements}.
The main contribution of \ourlogic{} is a new model which can support
unrestricted applications of \ref{rule:hoare-sum}
\emph{and} measurements in the same logic.

\subsection{The Main Challenge: Handling Measurements Soundly}
\label{sect:overview:measurement}

The fundamental issue with measurement is that it is not a unitary operator.
Specifically, it collapses a superposition state
$ \sum_i \alpha_i \ket{s_i} $
into the classical state $\ket{s_i}$ with probability $\abs{\alpha_i}^2$.
As a consequence,
one would think the following innocent-looking triples should be valid:
\begin{align}
  &\hoare{\qit \mapsto \ket{0}}{\gMZ[\qit]}{\qit \mapsto \ket{0}}
  &
  &\hoare{\qit \mapsto \ket{1}}{\gMZ[\qit]}{\qit \mapsto \ket{1}}
  \label{spec:meas-classical-wrong}
\end{align}
The triples say that an already classical state cannot be ``collapsed''
further by measuring it.
Remarkably, however,
any logic that would admit \cref{spec:meas-classical-wrong}
is fundamentally incompatible with \ref{rule:hoare-sum}.
This is because its application to \cref{spec:meas-classical-wrong}
would yield the invalid triple
\begin{equation}
  \hoare
    {\qit \mapsto (\alpha \ket{0} + \beta \ket{1})}
    {\gMZ[\qit]}
    {\qit \mapsto (\alpha \ket{0} + \beta \ket{1})}
\end{equation}
where no collapse of the superposition happens at all.
In fact, the logic of \citet{LeLSS22-qsl} would admit
\cref{spec:meas-classical-wrong}, and thus cannot support
the local reasoning afforded by \ref{rule:hoare-sum}.

At first sight, this might seem an insurmountable obstacle:
the \ref{rule:hoare-sum} rule seems to imply every program is linear,
but measurement is not.
Is our objective even achievable?

The key insight behind the solution we provide with \ourlogic
is that the real culprit is the handling of \emph{mixed states},
\ie probabilistic mixtures of pure quantum states (the outcomes of measurement).
In a logic like \citet{LeLSS22-qsl}'s,
just as in standard separation logic,
assertions talk about one possible outcome at a time.
For example, the most accurate representation of the mixture of two outcomes
$\Prop$ and $\PropB$ is $\Prop \lor \PropB$,
an assertion that does not fully describe the mixed state,
but only an over-approximation of the possible outcomes.

As a first step towards a solution, in \ourlogic{},
we move to assertions that can predicate over the whole mixed state.
To handle mixed states compositionally,
we introduce a new connective~$\Prop\mix{\hvar} \PropB$,
called \emph{tagged mixing},
which represents the mixed state resulting from running a measurement
tagged with $\hvar$ and with two outcomes~$0$ and $1$.\footnote{
  \citet{DengWX24-qsl}'s logic has a similar connective $\oplus$,
  but their model is fundamentally different from ours and incompatible with basis-locality.
  This will be discussed in more detail later in \cref{sect:related}.
}
Under this interpretation, the triples~\cref{spec:meas-classical-wrong}
become invalid: in \ourlogic{} they would assert that the state resulting
from a measurement is equivalent to one where no measurement was taken.
A valid rule for measurement in \ourlogic{} is:\footnote{
  The superscript $\hvar$ of the precondition means that the ownership of the variable $\hvar$ is required.
  See \cref{sect:logic:program} for the details.
}
\begin{proofrule}
  \infer*[lab = hoare-m${}_{\textsc{z}}$]{}{
    \hoare
      { \qit \mapsto (\alpha \ket{0} + \beta \ket{1}) }[^\hvar]
      { \gMZ^\hvar[\qit] }
      { (\qit \mapsto \alpha \ket{0}) \mix{\hvar} (\qit \mapsto \beta \ket{1}) }
  }
  \label{rule:hoare-mz}
\end{proofrule}
The rule states that from a state where \qit{} is in a superposition of states
$\ket{0}$ and $\ket{1}$ with coefficients $\alpha$ and $\beta$, respectively,
measuring \qit{} gives us a mixed state with two outcomes:
one where the state of \qit{} collapsed to the classical
state $\ket{0}$ and one where it collapsed to $\ket{1}$.
More precisely, the single outcomes in the postcondition retain their
coefficient (\ie we do not normalize the states) so that we can read off
the probabilities of each outcome (as the squared norm of the coefficient).

The crucial change induced by moving to assertions over mixed states is
that now the $\Prop + \PropB$ connective is not just given the meaning of
``linear combination'' but of ``outcome-wise linear combination''.
This is essentially what makes \ref{rule:hoare-sum} sound
in the presence of measurement: it no longer states that the program
is linear, but linear on each outcome
(\ie once all the measurement's effects have been factored out).
This is summarized in the following interchange rules,
which are validated by \ourlogic's model:
\begin{proofrules}
  \infer*[lab = mix-sum]{}{
    (\Prop_0\mix{\hvar} \Prop_1) + (\PropB_0\mix{\hvar} \PropB_1)
    \,\dashvdash\,
    (\Prop_0 + \PropB_0) \mix{\hvar} (\Prop_1 + \PropB_1)
  }
  \label{rule:mix-sum}

  \infer*[lab = mix-scale]{}{
    \scl{\alpha}{(\Prop_0\mix{\hvar} \Prop_1)}
    \,\dashvdash\,
    (\scl{\alpha}{\Prop_0}) \mix{\hvar} (\scl{\alpha}{\Prop_1})
  }
  \label{rule:mix-scale}
\end{proofrules}

Now we can resolve the apparent conflict between measurement and sum
we started with:
we can derive the specification
\ref{rule:hoare-mz}
of the behaviour of a measurement
on a superposition state,
from a specification of its behaviour on classical states.
In \ourlogic, the valid rules for measuring a classical state are:
\begin{align}
  &\hoare
    {\qit \mapsto \ket{0}}[^\hvar]
    {\gMZ^\hvar[\qit]}
    {
      (\qit \mapsto \ket{0})
      \mix{\hvar}
      \mghost{(\qit \mapsto \ket{1})}
      {(\qit \mapsto 0)}
    }
  \label{spec:meas-mix-zero}
  \\
  &\hoare
    {\qit \mapsto \ket{1}}[^\hvar]
    {\gMZ^\hvar[\qit]}
    {
      \mghost{(\qit \mapsto \ket{0})}
      {(\qit \mapsto 0)}
      \mix{\hvar}
      (\qit \mapsto \ket{1})
    }
  \label{spec:meas-mix-one}
\end{align}
Crucially, for these to be valid in our model,
they have to include in the postcondition
an explicit tagged mix connective,
with only one meaningful branch and one
probability-zero outcome (the measurement that cannot materialize)
represented as the zero-norm state $\qit \mapsto 0 $.
In fact, we can now \emph{derive} \ref{rule:hoare-mz} from
\cref{spec:meas-mix-zero} and \cref{spec:meas-mix-one}:
\begin{derivation}
\infer*[Right = {\labelstep{meas-deriv:swap}}]{
\infer*[Right = {\labelstep{meas-deriv:sum}}]{
\infer*[Right = {\labelstep{meas-deriv:scale-0}}]{
  \hoare
      {\qit \mapsto \ket{0}}[^\hvar]
      {\gMZ^\hvar[\qit]}
      {
        \qit \mapsto \ket{0}
        \pk1\mix{\hvar}\pk1
        \qit \mapsto 0
      }
}{
  \hoare
      {\qit \mapsto \alpha \ket{0}}[^\hvar]
      {\gMZ^\hvar[\qit]}
      {
        \qit \mapsto \alpha \ket{0}
        \pk1\mix{\hvar}\pk1
        \qit \mapsto 0
      }
}
\\
\infer*[Right = {\labelstep{meas-deriv:scale-1}}]{\hoare
      {\qit \mapsto \ket{1}}[^\hvar]
      {\gMZ^\hvar[\qit]}
      {
        \qit \mapsto 0
        \pk1\mix{\hvar}\pk1
        \qit \mapsto \alpha \ket{1}
      }
}{
  \hoare
      {\qit \mapsto \beta \ket{1}}[^\hvar]
      {\gMZ^\hvar[\qit]}
      {
        \qit \mapsto 0
        \pk1\mix{\hvar}\pk1
        \qit \mapsto \beta \ket{1}
      }
}}{
  \hoare
      {\qit \mapsto (\alpha \ket{0} + \beta \ket{1})}[^\hvar]
      {\gMZ^\hvar[\qit]}
      {
        \paren{\pk1
          \qit \mapsto \alpha \ket{0}
          \pk1\mix{\hvar}\pk1
          \qit \mapsto 0
        \pk1}
        +
        \paren{\pk1
          \qit \mapsto 0
          \pk1\mix{\hvar}\pk1
          \qit \mapsto \beta \ket{1}
        \pk1}
      }
}}{
  \hoare
      {\qit \mapsto (\alpha \ket{0} + \beta \ket{1})}[^\hvar]
      {\gMZ^\hvar[\qit]}
      {
        \qit \mapsto \alpha \ket{0}
        \pk1\mix{\hvar}\pk1
        \qit \mapsto \beta \ket{1}
      }
}
\end{derivation}
We start by applying \ref{rule:hoare-scale} and \ref{rule:mix-scale}
to \cref{spec:meas-mix-zero,spec:meas-mix-one} in
\cref{meas-deriv:scale-0,meas-deriv:scale-1},
to introduce the coefficients for each (classical) state.
Then, in \cref{meas-deriv:sum},
we apply \ref{rule:hoare-sum} to combine the states in a superposition;
this gives a postcondition which is the sum of two mixed states.
In the final \cref{meas-deriv:swap},
we use \ref{rule:mix-sum} to obtain a mixed state of sums, as required;
note how summing a state with the impossible outcome $\qit \mapsto 0$
leaves the state unchanged.

\subsection{The Three Layers of Locality}
\label{sect:overview:locality}

As we described, \ourlogic{} supports three locality principles
at the same time, displayed in \cref{fig:locality-rules}.
That is, in a mixed state, we can reason on a per-outcome basis
(\ref{rule:hoare-mix});
in a superposition state, we can reason on the classical basis states
(\ref{rule:hoare-sum});
and in a state where some qubits are disentangled, we can focus
on the relevant ones and ignore the others
(\ref{rule:hoare-frame}).
Given that a proof in \ourlogic{} will inevitably involve a combination
of all three layers, it is crucial to study the interaction between the three
constructs of mixing, sum and separation.
Our model is carefully constructed to enjoy a number of
distributivity/interchange rules that allow for flexible combinations of the connectives,
which we review next.

A first illustration of the flexibility of \ourlogic{}
was already presented in \ref{rule:mix-sum},
which shows an interchange law between mixing and sum.
This allows for handling superposition and mixing in any order
without losing information.
A similar interchange law holds for mixings arising from two consecutive measurements on different qubits.
\begin{proofrule}
  \infer*[lab = mix-mix]{}{
    (\Prop_{00} \mix{\hvar} \Prop_{01})
      \mix{\hvarB}
    (\Prop_{10} \mix{\hvar} \Prop_{11})
    \,\dashvdash\,
    (\Prop_{00} \mix{\hvarB} \Prop_{10})
      \mix{\hvar}
    (\Prop_{01} \mix{\hvarB} \Prop_{11})
  }
  \label{rule:mix-mix}
\end{proofrule}
The \ref{rule:mix-mix} rule says that
\ourlogic's model is insensitive to the order
in which measurements are taken,
showing that the $\mix{\hvar}$ connective is not simply
a representation of the program's measurements in the assertions,
but a genuine compositional abstraction over them.

Finally, separation is also well-behaved concerning superposition and mixing:
\begin{proofrules}
  \infer*[lab = sum-frame]{}{
    (\Prop + \PropB)*\PropC
    \,\vdash\,
    (\Prop*\PropC) + (\PropB*\PropC)
  }
  \label{rule:sum-frame}

  \infer*[lab = mix-frame]{}{
    (\Prop \mix{\hvar} \PropB)*\PropC
    \,\vdash\,
    (\Prop*\PropC) \mix{\hvar} (\PropB*\PropC)
  }
  \label{rule:mix-frame}
\end{proofrules}
Thanks to \ref{rule:mix-frame},
adding a disentangled frame to a mixed state
does not disallow per-outcome reasoning.
As we explain next, the reverse direction of \ref{rule:mix-frame}
is crucial for the scalability of reasoning in \ourlogic.

\subsection{Abstraction}
\label{sect:overview:abstraction}

As we argued, we can only hope to have a sound logic if
every measurement introduces a mixing connective in the postcondition.
Without care, this can induce an exponential growth in the
outcomes to be considered.
The key tool provided by \ourlogic{} to control this complexity is the
reverse direction of the \ref{rule:mix-frame} rule:
\begin{proofrule}
  \infer*[lab = mix-unframe]{
    \PropC \col \precise
  }{
    (\Prop*\PropC) \mix{\hvar} (\PropB*\PropC)
    \,\vdash\,
    (\Prop \mix{\hvar} \PropB)*\PropC
  }
  \label{rule:mix-unframe}
\end{proofrule}
The \ref{rule:mix-unframe} rule
says that it is possible (under a technical condition on $R$)
to factor out the portions of state that are common to multiple outcomes
into a single frame.
This allows any reasoning that only depends on $R$ to be done once,
and for the remaining information to be framed around the reasoning.

To show more concretely the positive effect of this rule on the abstraction
capabilities of \ourlogic, let us consider a specific example
which we treat in full detail in \cref{sect:cases:meas-cnot}.
In certain quantum architectures designed
to support fault tolerance~\cite{HorsmanFDM12-lattice-surgery},
the implementation of basic \mbox{2-qubit} gates, such as \gCNOT{} (a.k.a. \gCX{}), can be improved by
implementing their functionality through a combination of
\mbox{1-qubit} unitary gates and \mbox{2-qubit} measurements~\cite{FowlerG19-lattice-surgery}.
In \cref{fig:mcnot}, we show the schema of such a circuit
called \mCNOT,
encoding a \gCNOT{} gate.
The details are explained in \cref{sect:cases:meas-cnot},
but for our purposes, what is important is that the \mCNOT{}
encodes a \gCNOT{} between \qit{} and \qitC{},
using an auxiliary (ancilla) qubit \qitB{} initialized with state~$\ket{0}$, and
that the circuit contains several measurements.

Although the \mCNOT{} circuit is designed to be morally equivalent to the \gCNOT{},
formally, the two have important differences: \mCNOT{} requires an extra qubit,
and the measurements produce a mixed state.
Without care, using \mCNOT{} in compositional reasoning instead of the \gCNOT{}
may cause incorrect results.
However, if the circuit using \mCNOT{} is not making use of these differences,
reasoning about the overall correctness should proceed essentially
as if we used \gCNOT{} instead of \mCNOT{}.
In \ourlogic, we can replicate this rough argument fully formally.

The idea is that \mCNOT{} can be proven to satisfy a specification of the form:\footnote{
  To be precise, the precondition has the superscript of the classical variables used for storing the results of measurements.
}
\[
  \all{\qket}
  \hoare
    {(\qit, \qitC) \mapsto \qket * \qitB \mapsto \ket{0}}
    {\mCNOT[\qit, \qitB, \qitC]}
    {
      (\qit, \qitC) \mapsto \gCX \qket
      *
      \Prop_{\oplus}
    }
\]
where $\Prop_{\oplus}$ contains mixing connectives $\mix{}$ introduced by the measurements
and (per-outcome) information about the state of \qitB{} and the global phase.
The rest of the postcondition asserts that the effect on the \qit{} and \qitC{}
qubits is the same as the effect of a \gCNOT{} gate.
The ability to group the measurement ``side effects'' into $\Prop_{\oplus}$
is given by \ref{rule:mix-unframe}.
This grouping has two nice effects.
First, it allows any user of the specification to
lift reasoning that holds for \gCNOT{} to reasoning that holds for \mCNOT{} by
framing $\Prop_{\oplus}$.
Second,
when lifting reasoning in this way,
framing $\Prop_{\oplus}$ safeguards against
possible unsoundness,
\ie if a step applies to \gCNOT{} but not to \mCNOT{},
the frame $\Prop_{\oplus}$ prevents lifting the result to \mCNOT{}.

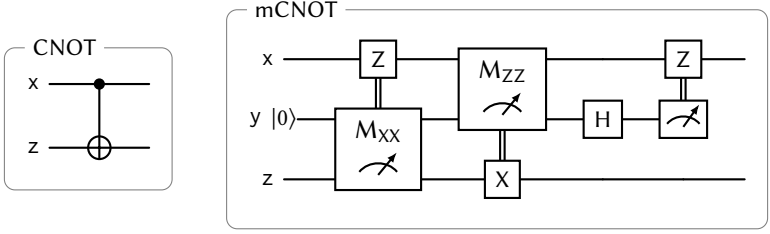
\begin{figure}
  \adjustfigure[\small]
  \begin{quantikz}[with border, pic label = \gCNOT{}, row sep = 0.3cm]
    \lstick{\qit} & \ctrl{2} & \\
    \setwiretype{n} & \\
    \lstick{\qitC} & \targ{} &
  \end{quantikz}
  \qquad
  \begin{quantikz}[with border, pic label = \mCNOT, row sep = 0.3cm]
    \lstick{\qit} &
    \gate{\gZ} \wire[d]{c} &
    \gate[2]{\smash{\raisebox{-.5\height}{\usebox{\meterinsideZ}}}} &
    &
    \gate{\gZ} \wire[d]{c} &
  \\
    \lstick{\qitB} \ket{0} &
    \gate[2]{\smash{\raisebox{-.5\height}{\usebox{\meterinsideX}}}} &
    \wire[d]{c} &
    \gate{\gH} &
    \meter{}
  \\
    \lstick{\qitC} &
    &
    \gate{\gX} & & &
  \end{quantikz}
  \caption{A \gCNOT{} gate (left) and its encoding \mCNOT{} using measurements (right).}
  \label{fig:mcnot}
  \Description{}
\end{figure}

We also note that, although we adopt a global-phase-sensitive logic,
we do not preclude users from employing density matrices
when a more compact representation is desired.
Density matrices can be encoded within our logic,
thus users can switch to use them to have smaller state descriptions;
however, this comes at the cost of reduced modularity.
We will see the details in \cref{sect:discussion}.

\subsection{Summary}
\label{sect:overview:summary}

In summary, we constructed a new logic, \ourlogic, that is able to achieve
modular reasoning across three ways of combining quantum programs:
by adding a disentangled state via \emph{separation},
by superposition via \emph{sum}, and
by adding measurements via \emph{mixing}.
In the remainder of the paper, we provide the model and rules of
\ourlogic, and evaluate it through a series of challenging case studies.
The omitted details and full proofs of soundness can be found in \appendixref{app:sect:logic}.
 \section{Preliminaries and Program Language}
\label{sect:lang}

\subsection{A Primer on Quantum Computing}
\label{sect:lang:primer}

The basic unit of data in quantum computing is the \emph{qubit},
which can take not only the \emph{classical states} $\ket{0}$ and $\ket{1}$,
but also a \emph{superposition} of them,
that is, a linear combination of the two states
$\qket = \alpha \ket{0} + \beta \ket{1} \in \mathbb{C}^2$.
Here, the squared norm $\norm{\qket}^2 = \abs{\alpha}^2 + \abs{\beta}^2$ represents the probability of that state.
More generally, the \emph{pure state} of a system of $n$ qubits
is a vector in a $2^n$-dimensional Hilbert space.
For example, the pure states of two qubits
are described by vectors of the form
$\alpha \ket{00} + \beta \ket{01} + \gamma \ket{10} + \delta \ket{11}$.
The states $\ket{b}$ with $b \in \mathbb{B}^n$ are the \emph{basis} states.
If two spaces $\Hspace_0$ and $\Hspace_1$
represent the states of quantum data $\qit$ and $\qitB$ respectively,
a state of the composite system $(\qit, \qitB)$ is called \emph{separated}
or \emph{disentangled} if
it can be represented as $\qket[_0] \otimes \qket[_1]$,
\ie there is no correlation between the two data.
If not, then the state is called \emph{entangled}.
An example of an entangled state is the Bell state
$\ketBell = (\ket{00} + \ket{11})/\sqrt{2}$.

Quantum states that differ only in a \emph{global phase} $\alpha \in \CC$,
\ie $\qket$ and $\alpha \qket$,
are physically indistinguishable: no measurement can detect the difference.
Only \emph{relative} phase---the difference in phase between coefficients of the
basis states---can be measured.

Quantum computing is performed by \emph{quantum gates}.
A quantum gate on $n$ qubits is a \emph{unitary operator} (or matrix) on the vector space $(\CC^2)^{\otimes n}$.
In particular, quantum gates are linear, and thus their effect is
determined by their effect on basis states:
$\Unitary (\sum_i \alpha_i \ket{b_i}) = \sum_i \alpha_i \Unitary \ket{b_i}$.
For example,
the $\gX$ gate (on a single qubit) flips
$\ket{0}$ to $\ket{1}$ and vice versa;
the $\gZ$ gate flips the sign of $\ket{1}$;
the $\gCX = \gCNOT$ gate (on two qubits) maps $\ket{a b}$ to $\ket{a, a \xor b}$
($\xor$ stands for xor).
Any classical bijection $f \col \{0, 1\}^n \to \{0, 1\}^n$ can be lifted to a unitary operator $\Unitary_f$ on $(\CC^2)^{\otimes n}$ that maps $\ket{x}$ to $\ket{f(x)}$.

Performing a measurement on a qubit in state
$\alpha \ket{0} + \beta \ket{1}$
makes the state collapse to $\ket{0}$ with probability $\abs{\alpha}^2$
and to $\ket{1}$ with probability $\abs{\beta}^2$.
More generally, a measurement $\Measure$ is defined as a set of linear operators $\{ \Measure^{(i)} \}_{i = 0}^{k}$ satisfying $\sum_{i = 0}^k \Measure^{(i)\dagger} \Measure^{(i)} = \gid$.
For example, the usual single-qubit $\gZ$-basis measurement $\gMZ$ is defined as $\curly{\gMZ^{(i)}}_{i = 0}^1$ such that $\gMZ^{(0)} \pk1\defeq\pk1 \ketbra{0}{0}$ and $\gMZ^{(1)} \pk1\defeq\pk1 \ketbra{1}{1}$.
The process of a measurement $\Measure$ is described as follows:
\[
  \qket \rightarrow \Measure^{(i)} \qket / \sqrt{\prob_i}
  \qquad \prob_i = \norm{\Measure^{(i)} \qket}^2.
\]
That is, the quantum state is projected with probability $\prob_i$ onto the subspace corresponding to each $\Measure^{(i)}$.
The probabilistic mixture of pure states resulting from a measurement
is called a \emph{mixed state}.

\subsection{Program Language}
\label{sect:lang:lang}

We define a simple, imperative language for quantum computing
with minimal features.
We equip it with a straightforward small-step operational semantics and
extract a denotational collecting semantics
which we use as a foundation for the model of our logic.

\paragraph{Syntax}

\begin{mathfig}
  \begin{gather*}
    \text{Qubit} \quad
    \qit, \qitB \, \in\, \Qubit
  \hspace{3em}
    \text{Variable} \quad
    \var, \varB, \varC, \hvar, \hvarB, \hvarC \, \in\, \dVar
  \\[.3em]
    \text{Pure expression} \quad
    \expr \sdefeq n \sor \var \sor \op(\bar{\expr})
  \hspace{3em}
    \text{Value} \quad
    \Val \, \ni\, \val, \valB \sdefeq n
  \\[.3em]
    \begin{aligned}
      \text{Command} \quad \dCmd \pk2\ni\pk2 \Cmd
      & \sdefeq \cskip
      \sor \Cmd;\pk1 \Cmd'
      \sor \ifelse{\expr}{\Cmd}{\Cmd'}
      \sor \whilex{\expr}{\Cmd}
    \\[-.2em] & \sskip
      \sor \var \store \expr
      \sor \Unitary[\bar{\qit}]
      \sor \var \store \Measure[\bar{\qit}]
    \end{aligned}
  \\[.4em]
    \ifonly{\expr}{\Cmd} \pk6\defeq\pk6
      \ifelse{\expr}{\Cmd}{\cskip}
  \hspace{2.5em}
    \Measure^\var[\bar{\qit}] \pk6\defeq\pk6
      \var \store \Measure[\bar{\qit}]
  \end{gather*}
  \caption{Syntax of the program language.}
  \label{fig:lang-syntax}
  \Description{}
\end{mathfig}

The syntax of our program language is summarized in \cref{fig:lang-syntax}.
We have pure expressions $\expr$ and commands $\Cmd$.
For simplicity, the command for measurement has the form $\var \store \Measure[\bar{\qit}]$, requiring that the result of the measurement be stored in some (classical) variable $\var$.
The notation $\Measure^\var[\bar{\qit}]$ used in \cref{sect:overview} is just shorthand for this command.
We represent Boolean values using $1$ for true and $0$ for false.
The guard of if statements and while loops is an expression.
To branch on a measurement result,
one can first store the outcome of the measurement into a variable
and then use it as a guard.

\paragraph{States}

\begin{mathfig}
  \[
    \textstyle
    \text{Quantum state} \hspace{.7em}
    \qket \pk3\in\pk3
    \Qstate \pk6\defeq\pk6
      \sum_{\pk1 \Qit \subseteq \Qubit}
      \bigotimes_{\qit \in \Qit} \CC^2
  \hspace{2em}
    \text{Store} \hspace{.7em}
    \Store \pk3\in\pk3 \dStore \pk6\defeq\pk6
      \dVar \pto \Val
  \]
  \caption{Domains for states in the language.}
  \label{fig:lang-states}
  \Description{}
\end{mathfig}

The domains for states in the language are summarized in \cref{fig:lang-states}.
The domain $\Qstate$ assigns a state vector $\qket$ to some set of qubits $\Qit$.
We denote its elements as $\Qit \mapsto \qket$, or just $\qket$ when there is no confusion.
A store $\Store \in \dStore$ assigns to some set of variables a value,
that is, classical data of any type.

\paragraph{Semantics of pure expressions}

The semantics $\sem{\expr}_\Store \pk2\in\pk2 \Val$ of a pure expression $\expr$ under the store $\Store$ is naturally defined as follows:
\[
  \sem{n}_\Store \pk5\defeq\pk5 n
\hspace{3em}
  \sem{\var}_\Store \pk5\defeq\pk5 \Store[\var]
\hspace{3em}
  \sem{\op(\bar{\expr})}_\Store \pk5\defeq\pk5
    \op(\overline{\sem{\expr}_\Store})
\]
Note that it is undefined when $\expr$ contains a variable $\var$ that is not in the domain of $\Store$.

\paragraph{Operational semantics}

\begin{mathfig}
  \begin{gather*}
    \text{Configuration} \quad
    \config \sdefeq (\bar{\Cmd}, \pk2 \qket, \pk2 \Store)
  \\[.4em]
    (\cskip, \, \bar{\Cmd}, \, \qket, \, \Store) \kredto
    (\bar{\Cmd}, \, \qket, \, \Store)
  \hspace{3em}
    (\Cmd_0; \Cmd_1, \, \bar{\Cmd'}\, \qket, \, \Store) \kredto
    (\Cmd_0, \, \Cmd_1, \, \bar{\Cmd'}, \, \qket, \, \Store)
  \\[.2em]
    (\ifelse{\expr}{\Cmd_1}{\Cmd_0}, \, \bar{\Cmd'}, \, \qket, \, \Store) \kredto
    (\Cmd_{\sem{\expr}_\Store}, \, \bar{\Cmd'}, \, \qket, \, \Store)
  \\[.2em]
    (\whilex{\expr}{\Cmd}, \, \bar{\Cmd'}, \, \qket, \, \Store) \kredto
    (\ifonly{\expr}{(\Cmd;\pk3 \whilex{\expr}{\Cmd})}, \,
      \bar{\Cmd'}, \, \qket, \, \Store)
  \\[.2em]
    (\var \store \expr, \, \bar{\Cmd}, \, \qket, \,
    \Store \{\var \store \val\})
    \kredto
    (\bar{\Cmd}, \, \qket, \,
    \Store \{\var \store \sem{\expr}_\Store\})
  \hspace{3em}
    (\Unitary[\bar{\qit}], \, \bar{\Cmd}, \, \qket, \, \Store) \kredto
    (\bar{\Cmd}, \, \Unitary_{\bar{\qit}} \qket, \, \Store)
  \\[.0em]
    (\var \store \Measure[\bar{\qit}], \, \bar{\Cmd}, \, \qket, \,
    \Store \{\var \store \val\})
    \kredto[i]
    (\bar{\Cmd}, \, \Measure^{(i)}_{\bar{\qit}} \qket, \,
    \Store \{\var \store i\})
  \end{gather*}
  \caption{Operational semantics of the program language.}
  \label{fig:opsem}
  \Description{}
\end{mathfig}

Next, we present the operational semantics of our language.
It is summarized in \cref{fig:opsem}.
The configuration $\config$ has the syntax shown at the top of \cref{fig:opsem}.
Here, $\qket \in \Qstate$ and $\Store \in \dStore$.
The small-step reduction relation $\config \redto \config'$, $\config \redto[i] \config'$ is defined by the rules in \cref{fig:opsem}.
Importantly, we put a label $i \in \NN$ on the reduction to indicate the result of a measurement $\Measure[\bar{\qit}]$.

\paragraph{Denotational semantics}

Now we present the denotational semantics of our language.
Our denotational semantics collects all the possible branches as a whole.
It is summarized in \cref{fig:denosem}.

\begin{mathfig}
  \begin{gather*}
    \text{Behaviour tree} \quad
    \BTree \, \ni\, \tree
      \pk2\sdefeq\pk2 \Branch(\bar{\tree})
      \pk2\sor\pk2 \Leaf(\qket, \Store)
      \pk2\sor\pk2 \Nil
    \quad\text{(coinductively)}
  \\[.5em]
    \sem{\empseq}(\qket, \Store) \pk6\defeq\pk6 \Leaf(\qket, \Store)
  \hspace{3em}
    \sem{\pk1 \cskip, \, \bar{\Cmd} \pk1}(\qket, \pk1 \Store) \pk6\defeq\pk6
      \sem{\pk1 \bar{\Cmd} \pk1}(\qket, \pk1 \Store)
  \\[.2em]
    \sem{\pk1 \Cmd_0; \Cmd_1, \, \bar{\Cmd'} \pk1}(\qket, \pk1 \Store) \pk6\defeq\pk6
      \sem{\pk1 \Cmd_0, \pk1 \Cmd_1, \pk1 \bar{\Cmd'} \pk1}(\qket, \pk1 \Store)
  \\[.2em]
    \sem{\pk1 \ifelse{\expr}{\Cmd_1}{\Cmd_0}, \, \bar{\Cmd'} \pk1}(\qket, \pk1 \Store) \pk6\defeq\pk6
      \sem{\Cmd_{\sem{\expr}_\Store}, \pk1 \bar{\Cmd'}}(\qket, \pk1 \Store)
  \\[.2em]
    \sem{\pk1 \whilex{\expr}{\Cmd}, \, \bar{\Cmd'} \pk1}(\qket, \pk1 \Store) \pk6\defeq\pk6
      \sem{\pk1 \ifonly{\expr}{(\Cmd;\pk3 \whilex{\expr}{\Cmd})}, \, \bar{\Cmd'} \pk1}(\qket, \pk1 \Store)
  \\[.2em]
    \sem{\pk1 \var \store \expr, \, \bar{\Cmd}}
      (\qket, \pk1 \Store\{\var \store \val\})
    \pk6\defeq\pk6
    \sem{\pk1 \bar{\Cmd} \pk1}
      (\qket, \pk1 \Store \{\var \store \sem{\expr}_\Store\})
  \hspace{3em}
    \sem{\pk1 \Unitary[\bar{\qit}], \, \bar{\Cmd} \pk1}(\qket, \pk1 \Store) \pk6\defeq\pk6
      \sem{\pk1 \bar{\Cmd} \pk1}(\Unitary_{\bar{\qit}} \qket, \pk1 \Store)
  \\[.0em]
    \sem{\pk1 \var \store \Measure[\bar{\qit}], \, \bar{\Cmd} \pk1}
    (\qket, \pk1 \Store \{ \var \store \val \})
    \pk6\defeq\pk6
      \Branch \pk2\paren[\big]{\pk3
        \overline{\sem{\pk1 \bar{\Cmd} \pk1}
          (\Measure^{(i)}_{\bar{\qit}} \qket, \pk1 \Store \{\var \store i\}) \nk2}^{\pk4 i}
      \pk2}
  \end{gather*}
  \caption{Denotational semantics of the program language.}
  \label{fig:denosem}
  \Description{}
\end{mathfig}

First, we define the domain of \emph{behaviour trees} $\tree \in \BTree$ as possibly infinite trees \emph{coinductively} by the syntax at the top of \cref{fig:denosem}.
The branch $\Branch(\bar{\tree})$ represents branching by a measurement.
The leaf $\Leaf(\qket, \Store)$ represents a branching-free terminating execution with the final state $(\qket, \Store)$.
The nil tree $\Nil$ represents a branching-free non-terminating execution.

We also define the partial order $\tree \le \tree'$ over behaviour trees coinductively by the rules $\Nil \pk2\le\pk2 \tree$, $\Leaf(\qket, \Store) \pk2\le\pk2 \Leaf(\qket, \Store)$, and ``if $\tree_i \le \tree'_i$ for each $i$, then $\Branch(\bar{\tree}) \pk2\le\pk2 \Branch(\bar{\tree'})$''.
In other words, $\tree \le \tree'$ means that $\tree'$ can be obtained from $\tree$ by replacing each occurrence of $\Nil$ with some tree.
Also, we define the child access $\tree.i$ as $\Branch(\bar{\tree'}).i \pk1\defeq\pk1 \tree'_i$ and undefined otherwise.

The denotational semantics $\sem{\pk1 \bar{\Cmd} \pk1}(\qket, \Store) \pk1\in\pk1 \BTree$ is defined \emph{inductively} by the equations in \cref{fig:denosem}.
The semantics is defined for a general sequence of commands $\bar{\Cmd}$ ($\empseq$ is the empty sequence).
Technically, the \emph{least} fixed point for the equations in \cref{fig:denosem} is constructed as the limit $\lim_{n \to \infty} \sem{-}^n(-, -)$ of $n$-th approximations $\sem{-}^n(-, -) \col \SemArg \pto \BTree$, leaving parts that have not terminated in $n$ steps as $\Nil$.\footnote{
  First, let $\Sem(f)(-, -, -) \pk2\col\pk2 \SemArg \pto \BTree$ for $f \col \SemArg \pto \BTree$ (where $\SemArg \pk1\defeq\pk1 \dCmd^* \pk1\times\pk1 \Qstate \pk1\times\pk1 \dStore$) be the map obtained by replacing all the self-references to $\semop$ with $f$ in the definition of $\sem{\pk1 \bar{\Cmd} \pk1}(\qket, \Store)$ in \cref{fig:denosem}.
  Then the $n$-th approximation is inductively defined by
  $\sem{-}^0(-, -) \pk2\defeq\pk2 \lam{\_, \_, \_} \Nil$
  and $\sem{-}^{n + 1}(-, -) \pk2\defeq\pk2 \Sem(\sem{-}^n(-, -))$.
  Finally, we set $\sem{-}(-, -) \pk2\defeq\pk2 \lim_{n \to \infty} \sem{-}^n(-, -)$.
  Here, we take the limit of an $\omega$-chain (\ie increasing sequence) $\sem{-}^0(-, -) \le \sem{-}^1(-, -) \le \sem{-}^2(-, -) \le \cdots$, which is defined because $\BTree$ is $\omega$-complete.
  Note that partial maps $\SemArg \pto \BTree$ are given the pointwise order, where undefined parts are regarded as the maximum element.
}
Notably, we can enjoy equational reasoning using the denotational semantics.

The following theorem formalizes the relation between
operational and denotational semantics.

\begin{theorem}[Equivalence of the operational and denotational semantics]
\label{thm:opsem-denosem}
  Take any configuration $\config \pk1=\pk1 (\bar{\Cmd}, \qket, \Store)$.
  The configuration never gets stuck in any branches (\ie any $\config'$ reachable from $\config$ is reducible) if and only if $\pk3\tree \pk1\defeq\pk1 \sem{\pk1 \bar{\Cmd} \pk1}(\qket, \Store)$ is defined.
  Moreover, assuming that $\tree$ is defined, $\config \redsto[\bar{i}] (\qketB, \Store')$ holds if and only if $\pk2\tree\overline{.i} = \Leaf(\qketB, \Store')$.
\end{theorem}
\begin{proof}
  For the first statement, the backward implication is straightforward.
  The forward implication follows from the definedness of all the approximations $\semop^n$, proved by induction over $n$.
  For the second statement, the forward implication is obtained by unfolding $\semop$.
  The backward implication can be proved by the fact that there exists some $n$ such that the $n$-th approximation has that leaf, \ie $\sem{\pk1 \bar{\Cmd} \pk1}^n(\qket, \Store) \overline{.i} = \Leaf(\qketB, \Store')$.
\end{proof}
 \section{The \ourlogic{} Logic}
\label{sect:logic}

In this section, we formally introduce our logic, \ourlogic.
We treat assertions and judgments semantically,
with entailment rules being just valid lemmas about the semantic model.

\begin{mathfig}
  \begin{gather*}
    \res, \resB \pk3\in\pk3 \Res \pk6\defeq\pk6
      \Qstate \pk1\times\pk1 \dStore
  \hspace{3em}
    \Prop, \PropB \pk3\in\pk3 \SLProp \pk6\defeq\pk6 \Pow(\Mult(\Res))
  \\[.1em]
    \Prop \vdash \PropB \pk6\defeq\pk6 \Prop \subseteq \PropB
  \hspace{3em}
    \Prop\vdash \Prop
  \hspace{3em}
    \frac{
      \Prop \vdash \PropB
    \hspace{1.5em}
      \PropB\vdash \PropC
    }{
      \Prop \vdash \PropC
    }
  \hspace{3em}
  \end{gather*}
  \caption{Model of \ourlogic's SL propositions.}
  \label{fig:model-slprop}
  \Description{}
\end{mathfig}

\subsection{Resources, Propositions and Entailment}
\label{sect:logic:basics}

We start by fixing notation for multisets,
which we use to build our model of mixed state resources.

\begin{definition}[Multisets]
\label{def:multiset}
  For a set $A$, the set of multisets $\Mult(A)$ is defined as the quotient of the set\footnote{
    Here, we fix some universe of sets $\Set$.
  } $\bigcup_{I \in \Set} I \to A$ over the following equivalence relation:
  $f \col I \to A$ and $g \col J \to A$ are equivalent if there is a bijection $\prob \col I \to J$ such that $f = g \circ \prob$.
  In other words, multisets are maps with the key forgotten.
  We write $\floor{f} \in \Mult(A)$ for the multiset of the equivalence class of $f \col I \to A$.

  We use the extensional notation $\mset{a_0, \dots, a_{n - 1}}$, regarding the sequence as a map from $\{0, \dots, n - 1\}$.
  For example, the multiset $\mset{0, 1, 1, 2}$ is the same as $\mset{2, 1, 0, 1}$ but different from $\mset{0, 1, 2}$.
  We also use the comprehension notation for multisets.
  For example, we can write $\mset{f\, i \mid i \in I}$ for $\floor{f \col I \to A}$.
  Also, we can take an element from multisets in the comprehension notation, using multiset membership $\mIn$, taking the multiplicity into account.
  For example, $\mset{ n \mid n \mIn \mset{0, 0, 1, 5}, \pk1 n < 2 }$ means the multiset $\mset{0, 0, 1}$, not $\mset{0, 1}$.

  Also, multiset inclusion $\mult \msubseteq \mult'$ over $\mult, \mult' \in \Mult(A)$ is defined as follows:
  $\floor{f \col I \to A} \msubseteq \floor{g \col J \to A}$ holds if and only if there exists an injection $\prob \col I \to J$ such that $f = g \circ \prob$.
  For example, $\mset{0, 1, 2} \msubseteq \mset{0, 1, 2, 2}$ holds, but $\mset{0, 1, 2, 2} \msubseteq \mset{0, 1, 2}$ does not.
\qed
\end{definition}

Now we can model \ourlogic's SL propositions as presented in \cref{fig:model-slprop}.
A resource $\res \in \Res$ consists of a quantum state $\Qit \mapsto \qket \in \Qstate$ and a store $\Store \in \dStore$.
A mixed state is represented as a \emph{multiset} of resources,
collecting all the outcomes of branching due to measurement.
An SL proposition $\Prop \in \SLProp$ is modelled as the set of multisets of resources that satisfy it.
The entailment relation $\Prop \vdash \PropB$ is defined simply by set inclusion, which is clearly reflexive and transitive.
We can define the standard connectives, such as the universal and existential quantifiers $\all{x \in A}{\Prop_x}$, $\ex{x \in A}{\Prop_x}$, as usual, with standard proof rules;
please refer to \appendixref{app:sect:logic:standard} for the details.

\subsection{Bare Mixing}
\label{sect:logic:bmix}

\begin{mathfig}
  \begin{gather*}
    \nb \pk6\defeq\pk6
      \{ \mset{ } \}
  \hspace{2.5em}
    \Prop \oplus \PropB \pk6\defeq\pk6
      \set{ \mult \uplus \mult' }{ \mult \in \Prop, \, \mult' \in \PropB }
  \\[.1em] \textstyle
    \bigoplus_{x \in I} \Prop_x \pk6\defeq\pk6
      \set{ \biguplus_{x \in I} \mult_x }{
        \all{x \in I} \mult_x \in \Prop_{x}
      }
  \end{gather*} \vspace{-.3em}
  \begin{proofrules}
    \infer*[lab = bmix-mono]{
      \Prop \vdash \Prop'
    \hspace{1.5em}
      \PropB \vdash \PropB'
    }{
      \Prop \oplus \PropB \, \vdash\, \Prop' \oplus \PropB'
    }
    \label{lrule:bmix-mono}

    \infer*[lab = nb-bmix]{}{
      \nb \oplus \Prop \, \dashvdash\, \Prop
    }
    \label{lrule:nb-bmix}

    \infer*[lab = bmix-comm]{}{
      \Prop \oplus \PropB \, \dashvdash\,
        \PropB \oplus \Prop
    }
    \label{lrule:bmix-comm}

    \infer*[lab = bmix-assoc]{}{
      (\Prop \oplus \PropB) \oplus \PropC \, \dashvdash\,
        \Prop \oplus (\PropB \oplus \PropC)
    }
    \label{lrule:bmix-assoc}
  \end{proofrules} \vspace{.0em}
  \begin{proofrules}
    \infer*[lab = bigbmix-mono]{
      \all{x \in I}\, \paren{\Prop_x \vdash \PropB_x}
    }{ \textstyle
      \bigoplus_{x \in I} \Prop_x \, \vdash\,
        \bigoplus_{x \in I} \PropB_x
    }
    \label{lrule:bigbmix-mono}

    \infer*[lab = bigbmix-comm]{
      \text{$f \col I \to J$ is a bijection}
    }{ \textstyle
      \bigoplus_{x \in I} \Prop_{\pk1 f\pk2 x} \, \dashvdash\,
        \bigoplus_{y \in J} \Prop_y
    }
    \label{lrule:bigbmix-comm}
  \end{proofrules} \vspace{.2em}
  \begin{proofrules}
    \infer*[lab = bigbmix-assoc]{}{ \textstyle
      \bigoplus_{x \in I} \bigoplus_{y \in J_x} \Prop_{x, y} \, \dashvdash\,
        \bigoplus_{(x, y) \in \bigsqcup_{x \in I} J_x} \Prop_{x, y}
    }
    \label{lrule:bigbmix-assoc}

    \infer*[lab = nb-bigbmix]{}{ \textstyle
      \nb \, \dashvdash\, \bigoplus_{\_ \in \empset}
    }
    \label{lrule:nb-bigbmix}

    \infer*[lab = bmix-bigbmix]{}{ \textstyle
      \Prop_0 \oplus \Prop_1 \, \dashvdash\,
        \bigoplus_{i \in \{0, 1\}} \Prop_i
    }
    \label{lrule:bmix-bigbmix}
  \end{proofrules}
  \caption{Bare mixing and its proof rules in \ourlogic.}
  \label{fig:bmix}
  \Description{}
\end{mathfig}

Tagged mixing $\mix{\hvar}$ we introduced in \cref{sect:overview}
is derived, in \ourlogic, from a more fundamental connective~$\oplus$
called \emph{bare mixing}, defined in \cref{fig:bmix}.
A proposition $P \oplus Q$ represents all collections of outcomes
in the mixed states described by $P$ and by $Q$.
It is therefore naturally expressed via (pointwise) multiset sum.

\begin{definition}[Sum of multisets]
\label{def:mutiset-uplus}
  The sum of multisets $\mult \uplus \mult'$ is defined as follows:
  $\floor{f \col I \to A} \uplus \floor{g \col J \to A} \pk3\defeq\pk3 \floor{h \col I + J \to A}$ where $h\, (\inl i) \pk1\defeq\pk1 f\, i$ and $h\, (\inr j) \pk1\defeq\pk1 g\, j$.
  For example, $\mset{0, 1} \uplus \mset{1, 2} = \mset{0, 1, 1, 2}$.
  Moreover, the indexed sum $\biguplus_{x \in I} \mult_x$ over an indexed family of multisets $(\mult_x)_{x \in I}$ is defined as
  $\biguplus_{x \in I} \floor{f_x \col J_x \to A} \pk3\defeq\pk3
    \floor{g \col \bigsqcup_{x \in I} J_x \to A}$
  where $g\, (x, j) \pk1\defeq\pk1 f_x\, j$.
  We can also use the comprehension notation $\mset{\pk1 a \mid x \in I, \pk1 a \mIn \mult_x \pk1}$ for $\biguplus_{x \in I} \mult_x$.
\qed
\end{definition}

\Cref{fig:bmix} present rules that reflect basic properties of
the multiset sum semantics.
Notably, binary mixing $\oplus$ has the unit $\nb$ (\ref{lrule:nb-bmix}) and is commutative (\ref{lrule:bmix-comm}) and associative (\ref{lrule:bmix-assoc}).
Here, we introduce the no-behaviour assertion $\nb$, modelled as the empty multiset.
As we allow finite and infinite multisets,
we can also introduce the indexed mixing $\bigoplus_{x \in I} \Prop_x$ over an arbitrary (possibly infinite) set $I$.
Indexed mixing is commutative and associative (\ref{lrule:bigbmix-comm}, \ref{lrule:bigbmix-assoc}).
Note that $\nb$ corresponds to mixing indexed over the empty set $\empset$ (\ref{lrule:nb-bigbmix}; here the body of $\bigoplus$ is an empty family of propositions) and binary $\oplus$ corresponds to mixing indexed over $\{0, 1\}$ (\ref{lrule:bmix-bigbmix}).

\subsection{Separating Conjunction}
\label{sect:logic:sep}

\begin{mathfig}
  \begin{gather*}
    \unit_\Qstate \pk3\defeq\pk3 \empset \mapsto 1
  \hspace{3em}
    (\Qit \mapsto \qket) \pk2\cdot_\Qstate\pk2 (\QitB \mapsto \qketB) \pk8\defeq\pk8
      \Qit \cup \QitB \pk2 \mapsto \pk2 \qket \otimes \qketB
      \hspace{.9em} \text{if $\Qit \pk1\cap\pk1 \QitB \pk2=\pk2 \empset$}
  \\[.1em]
    \unit_\dStore \pk3\defeq\pk3 \empset
  \hspace{3em}
    \Store \pk2\cdot_\dStore\pk2 \Store' \pk8\defeq\pk8
      \Store \cup \Store'
      \hspace{.9em} \text{if $\dom \Store \pk1\cap\pk1 \dom \Store' \pk2=\pk2 \empset$}
  \\[.1em]
    \unit_\Res \pk3\defeq\pk3 (\unit_\Qstate, \unit_\dStore)
  \hspace{3em}
    (\qket, \Store) \pk2\cdot_\Res\pk2 (\qketB, \Store') \pk6\defeq\pk6
      (\qket \cdot_\Qstate \qketB, \pk3 \Store \cdot_\dStore \Store)
  \end{gather*}
  \caption{PCM structure over $\Qstate$, $\dStore$, and $\Res$.}
  \label{fig:res-pcm}
  \Description{}
\end{mathfig}

Our next question is how to define separating conjunction $*$.
Following the usual approach, we use PCMs (partial commutative monoids).

\begin{definition}[PCM]
\label{def:pcm}
  A partial commutative monoid (PCM) $(A, \pk2 \unit_A \in A, \pk2 \cdot_A \col A \times A \pto A)$ is a set $A$ equipped with the unit $\unit_A$ and the partial product $\cdot_A$ (we may omit the subscript) such that $\unit_A \cdot_A a = a$, $a \cdot_A b = b \cdot_A a$, and $(a \cdot_A b) \cdot_A c = a \cdot_A (b \cdot_A c)$.
  Here, we use the Kleene equality, where the left-hand side is defined if and only if the right-hand side is defined.
\qed
\end{definition}

The PCM structure over $\Qstate$, $\dStore$ and $\Res$, presented in \cref{fig:res-pcm}, is standard.
For stores $\dStore$, the product is defined only if the domains are disjoint.
For quantum states $\Qstate$, the product is defined only if the qubit sets are disjoint and uses the tensor product $\qket \otimes \qketB$ for the qubit vector.
For resources $\Res$, we simply define operations component-wise.

\begin{mathfig}
  \begin{gather*}
    \unit_{\pk1\Mult(\Res)} \pk3\defeq\pk3 \mset{ \unit_\Res }
  \hspace{3em}
    \mult \pk2\cdot_{\Mult(\Res)} \pk2 \mult' \pk8\defeq\pk8
      \mset{\pk2 \res \cdot_\Res \resB \mid \res \mIn \mult, \pk2 \resB \mIn \mult' \pk2}
  \\[.2em]
    \emp \pk6\defeq\pk6
      \{\pk4 \unit_{\pk1\Mult(\Res)} \pk4\}
  \hspace{2.5em}
    \Prop \pk1*\pk1 \PropB \pk6\defeq\pk6
      \set{\pk2 \mult \cdot_{\Mult(\Res)} \mult' }{ \mult \in \Prop, \, \mult' \in \PropB \pk2}
  \end{gather*} \vspace{-.3em}
  \begin{proofrules}
    \infer*[lab = sep-mono]{
      \Prop \vdash \Prop'
    \hspace{1.5em}
      \PropB \vdash \PropB'
    }{
      \Prop \pk1*\pk1\PropB \, \vdash\, \Prop' \pk1*\pk1 \PropB'
    }
    \label{lrule:sep-mono}

    \infer*[lab = emp-sep]{}{
      \emp \pk1*\pk1 \Prop \, \dashvdash\, \Prop
    }
    \label{lrule:emp-sep}

    \infer*[lab = sep-comm]{}{
      \Prop \pk1*\pk1 \PropB \, \dashvdash\,
        \PropB \pk1*\pk1 \Prop
    }
    \label{lrule:sep-comm}

    \infer*[lab = sep-assoc]{}{
      (\Prop \pk1*\pk1 \PropB) \pk1*\pk1 \PropC \, \dashvdash\,
        \Prop \pk1*\pk1 (\PropB \pk1*\pk1 \PropC)
    }
    \label{lrule:sep-assoc}
  \end{proofrules} \vspace{-.3em}
  \begin{proofrules}
\infer*[lab = bigbmix-frame]{}{ \textstyle
      (\bigoplus_{x \in I} \Prop_x) \pk2*\pk2 \PropB \, \vdash\,
        \bigoplus_{x \in I} \pk2 (\Prop_x * \PropB)
    }
    \label{lrule:bigbmix-frame}

\infer*[lab = bigbmix-unframe]{
      \PropB \col \precise
    }{ \textstyle
      \bigoplus_{x \in I} \pk2 (\Prop_x * \PropB) \, \vdash\,
        (\bigoplus_{x \in I} \Prop_x) \pk2*\pk2 \PropB
    }
    \label{lrule:bigbmix-unframe}
  \end{proofrules}
  \caption{PCM structure over $\Mult(\Res)$, separating conjunction, and its proof rules in \ourlogic.}
  \label{fig:pcm-sep}
  \Description{}
\end{mathfig}

\begin{mathfig}
  \begin{gather*}
    \Prop \col \precise \pk9\defeq\pk9
      \all{\mult, \mult' \in \Prop} \pk3 \mult = \mult'
\\[.2em]
    \frac{
      \Prop \col \precise
    \hspace{1em}
      \PropB \vdash \Prop
    }{
      \PropB \col \precise
    }
  \hspace{2.5em}
    \nb, \pk2 \emp \col \precise
  \hspace{2.5em}
    \frac{
      \all{x} \pk2 \paren{\Prop_x \col \precise}
    }{
      \bigoplus_{x \in I} \Prop_x \col \precise
    }
  \hspace{2.5em}
    \frac{
      \Prop, \PropB \col \precise
    }{
      \Prop * \PropB \col \precise
    }
  \end{gather*}
  \caption{Precision of SL propositions.}
  \label{fig:precise}
  \Description{}
\end{mathfig}

Now, our SL propositions $\SLProp$ are a set of multisets $\Mult(\Res)$.
We want to give a PCM structure to them to define separating conjunction $*$, the core of separation logic.
The top of \cref{fig:pcm-sep} shows how we can do that.
The unit is just the singleton multiset consisting of the unit resource.
The product distributes over each argument multiset.
Note that membership $\res \mIn \mult$, $\resB \mIn \mult'$ in the comprehension notation takes the multiplicity into account.
More explicitly, $\cdot_{\Mult(\Res)}$ is defined as follows:
$\floor{f \col I \to \Res} \cdot_{\Mult(\Res)} \floor{g \col J \to \Res} \pk4\defeq\pk4
  \floor{\lam{(i, j) \in I \times J} \pk2 f\, i \cdot_\Res g\, j}$.

Directly using this PCM structure, we can define empty ownership $\emp$ and separating conjunction $\Prop * \PropB$.
We enjoy the rules \ref{lrule:sep-mono}, \ref{lrule:emp-sep}, \ref{lrule:sep-comm}, \ref{lrule:sep-assoc} directly from the definition and the PCM structure.
Also, bare mixing $\oplus$ distributes over separating conjunction $*$ by rules \ref{lrule:bigbmix-frame} and \ref{lrule:bigbmix-unframe}.
This comes from the fact that multiset sum $\uplus$ distributes over $\cdot_{\Mult(\Res)}$, \ie
$(\mult_0 \uplus \mult_1) \cdot_{\Mult(\Res)} \mult' \pk1=\pk1 (\mult_0 \cdot_{\Mult(\Res)} \mult') \pk2\uplus\pk2 (\mult_1 \cdot_{\Mult(\Res)} \mult')$ (with the Kleene equality).
Note that framing out an assertion $\PropB$ \ref{lrule:bigbmix-unframe} requires that the assertion $\PropB$ is precise, meaning that it represents up to one multiset of resources.
Very roughly speaking, assertions `without disjunction' are precise.
For example, in the rule \ref{lrule:bigbmix-unframe}, if we could set $\PropB \pk1=\pk1 (\var \mapsto 0 \lor \var \mapsto 1)$ (which is \emph{not} precise) and $\Prop_0 = \Prop_1 = \emp$ with $I = \{0, 1\}$, then this is unsound, because the state $\var \mapsto 0 \pk1\oplus\pk1 \var \mapsto 0$ is included in the left-hand side but not in the right-hand side.
\Cref{fig:precise} shows the definition and basic rules of the precision
judgment $\Prop \col \precise$, which are straightforward.

\subsection{Reasoning about Quantum Programs}
\label{sect:logic:program}

Now we present the features of \ourlogic{} for reasoning about quantum programs.
The sum $\Prop + \PropB$ of SL assertions, the heart of \ourlogic,
is discussed later in \cref{sect:logic:sum}.

\paragraph{Tokens for quantum and classical states}

\begin{mathfig}
  \begin{gather*}
    \bar{\qit} \mapsto \qket \pk6\defeq\pk6
      \set{
        \mset{ (\curly{\bar{\qit}} \mapsto \qket, \unit) }
      }{
        \distinct(\bar{\qit})
      }
  \hspace{3em}
    \lift{\alpha} \pk6\defeq\pk6
      () \mapsto \alpha
  \hspace{3em}
    \scl{\alpha}{\Prop} \pk6\defeq\pk6
      \lift{\alpha} * \Prop
  \\[.1em] \textstyle
    \var \mapsto \val \pk6\defeq\pk6
      \curly[\big]{
        \mset{
          (\unit, \curly{ (\var, \val) })
        }
      }
  \hspace{3em}
    \vtok{\bar{\var}} \pk6\defeq\pk6
      \bigsep_i\pk1 \ex{\val} \var_i \mapsto \val
  \hspace{3em}
    \rtok{\val} \pk6\defeq\pk6
      \ret \mapsto \val
  \end{gather*} \vspace{-.3em}
  \begin{proofrules}
    \infer*{}{
      \bar{\qit} \mapsto \qket, \pk2 \var \mapsto \val \col \precise
    }

    \infer*[lab = bigbmix-scale]{}{ \textstyle
      \scl{\alpha}{(\bigoplus_{x \in I} \Prop_x)} \, \dashvdash\,
        \bigoplus_{x \in I} \pk2 \scl{\alpha}{\Prop_x}
    }
    \label{lrule:bigbmix-scale}
  \end{proofrules} \vspace{-.5em}
  \begin{proofrules}
    \infer*[lab = qpoints-sep]{}{
      \bar{\qit} \mapsto \qket \pk3*\pk3
      \bar{\qitB} \mapsto \qketB \pk5\dashvdash\pk5
        (\bar{\qit}, \bar{\qitB}) \mapsto (\qket \otimes \qketB)
    }
    \label{lrule:qpoints-sep}

    \infer*[lab = scale-qpoints]{}{
      \scl{\alpha}{\paren{\bar{\qit} \mapsto \qket}} \pk4\dashvdash\pk4
        \bar{\qit} \mapsto \alpha \qket
    }
    \label{lrule:scale-qpoints}
\end{proofrules}
  \caption{Tokens for quantum and classical states and their proof rules.}
  \label{fig:tokens}
  \Description{}
\end{mathfig}

Now we introduce the tokens for quantum and classical states as shown in \cref{fig:tokens}.
The quantum points-to token $\bar{\qit} \mapsto \qket$ is defined simply as the quantum state with the domain $\curly{\bar{\qit}}$.
Here, $\distinct(\bar{\qit})$ means that the qubit names $\bar{\qit}$ are mutually distinct, \ie $\qit_i \ne \qit_j$ if $i \ne j$.
As a special case, we have $ () \mapsto \alpha$ for the state of a one-dimensional vector of zero qubit, which is just a complex-number coefficient $\alpha$ working as a global phase of the (global) state.
We abbreviate such an assertion with $\lift{\alpha}$ (or sometimes just $\alpha$)
and represent scaling $\scl{\alpha}{\Prop}$ as $\lift{\alpha} * \Prop$.
The classical points-to token $\var \mapsto \val$ is modelled as standard.
For utility, we introduce the classical variable token $\vtok{\bar{\var}}$, which represents ownership of the variables $\var$, while ignoring their values.
We also introduce the return-value token $\rtok{\val}$ for representing the return value of a pure expression (see \cref{fig:hoare} shown later).
For simplicity, it is derived as a classical points-to token for a special variable $\ret \in \dVar$ that clients cannot use in programs.

The tokens satisfy natural proof rules, as shown in \cref{fig:tokens}.
First, the quantum and classical points-to tokens are precise.
Thanks to this, in particular, the phase assertion $\lift{\alpha}$ is precise, and combining this with \ref{lrule:bigbmix-unframe} along with \ref{lrule:bigbmix-frame}, we can derive the property that mixing distributes over scaling (\ref{lrule:bigbmix-scale}).
Also, the quantum points-to tokens can be merged and split according to the tensor product $\otimes$ (\ref{lrule:qpoints-sep}).
So separating conjunction $*$ represents the disentanglement of quantum states along with the ownership disjointness.
Note that from it we can derive the scaling rule \ref{lrule:scale-qpoints}.
We also have a proof rule for permuting the qubits of a quantum points-to token;
see \appendixref{app:sect:logic:rules} for the details.

\paragraph{Hoare triples}

\begin{mathfig}
  \begin{gather*}
    \msem{\pk1 \expr \pk1}(\mult) \pk6\defeq\pk6
      \mset[\big]{\pk2 (\qket, \Store \uplus \curly{ (\ret, \sem{\expr}_\Store) } ) \pk5\big\vert\pk5 (\qket, \Store) \mIn \mult \pk2}
  \\[.1em]
    \WP{ \expr }{ \Prop } \pk8\defeq\pk8
      \set{ \mult }{ \msem{\pk1 \expr \pk1}(\mult) \in \Prop }
  \hspace{3em}
    \hoare{ \Prop }{ \expr }{ \PropB } \pk8\defeq\pk8
      \Prop \pk3\vdash\pk3 \WP{ \expr }{ \PropB }
  \\[.0em]
    \hoare{\pk1 \emp \pk1}{ \val }{\pk1 \rtok{\val} \pk1}
  \hspace{3em}
    \hoare{\pk2 \var \mapsto \val \pk2}{ \var }{\pk2 \rtok{\val} \pk2*\pk2 \var \mapsto \val \pk2}
  \hspace{3em}
    \frac{
      \all{i} \pk2 \hoare{\pk2 \Prop \pk2}{ \expr_i }{\pk2 \rtok{\val_i} \pk2*\pk2 \Prop \pk2}
    }{
      \hoare{\pk2 \Prop \pk2}{ \op(\bar{\expr}) }{\pk2 \rtok{\op(\bar{\val})} \pk2*\pk2 \Prop \pk2}
    }
  \\[.05em]
    \frac{
      \hoare{ \Prop }{ \expr }{ \PropB }
    }{
      \hoare{\pk2 \Prop * \PropC \pk2}{ \expr }{\pk2 \PropB * \PropC \pk2}
    }
  \hspace{3em}
    \frac{
      \all{x \in A} \hoare{ \Prop_x }{ \expr }{ \PropB_x }
    }{
      \hoare{\pk2 \bigoplus_{x \in I} \Prop_x \pk2}{ \expr }{\pk2 \bigoplus_{x \in I} \PropB_x \pk2}
    }
  \hspace{3em}
    \frac{
      \hoare{ \Prop }{ \expr }{ \PropB }
    \hspace{1em}
      \PropB \vdash \PropB'
    }{
      \hoare{ \Prop }{ \expr }{ \PropB' }
    }
  \\[.8em]
    \Leaves(\tree) \pk6\defeq\pk6
      \mset[\big]{\pk2 (\qket, \Store) \pk5\big\vert\pk5
        \tree\overline{.i} = \Leaf(\qket, \Store) \pk2}
  \\[.1em]
    \msem{\pk1 \Cmd \pk1}(\mult) \pk8\defeq\pk8
      \mset[\big]{\pk2 (\qketB, \Store') \pk5\big\vert\pk5
        (\qket, \Store) \mIn \mult, \pk2
        (\qketB, \Store') \mIn \Leaves(\sem{\Cmd}(\qket, \Store)) \pk2}
  \\[.1em]
    \WP{ \Cmd }{ \Prop } \pk8\defeq\pk8
      \set{ \mult }{ \msem{\pk1 \Cmd \pk1}(\mult) \in \Prop }
  \hspace{3em}
    \hoare{ \Prop }{ \Cmd }{ \PropB } \pk8\defeq\pk8
      \Prop \pk3\vdash\pk3 \WP{ \Cmd }{ \PropB }
  \\[-.05em]
    \hoare{ \Prop }[^{\bar{\var}}]{ \expr }{ \PropB }[^{\bar{\varB}}] \pk8\defeq\pk8
      \hoare{ \Prop * \vtok{\bar{\var}} }{ \expr }{ \PropB * \vtok{\bar{\varB}} }
  \hspace{2.5em}
    \hoare{ \Prop }[^{\bar{\var}}]{ \Cmd }{ \PropB }[^{\bar{\varB}}] \pk8\defeq\pk8
      \hoare{ \Prop * \vtok{\bar{\var}} }{ \Cmd }{ \PropB * \vtok{\bar{\varB}} }
  \end{gather*} \vspace{-.4em}
  \begin{proofrules}
    \infer*[lab = hoare-unitary]{}{
      \hoare{\pk2 \bar{\qit} \mapsto \qket \pk2}
        { \Unitary[\bar{\qit}] }
        {\pk2 \bar{\qit} \mapsto \Unitary \qket \pk2}
    }
    \label{lrule:hoare-unitary}

    \infer*[lab = hoare-store]{
      \hoare{ \Prop }[^{\var}]{ \expr }
        { \rtok{\val} \pk1*\pk1 \PropB }^{\var}
    }{
      \hoare{\pk2 \Prop \pk2}[^{\var}]
        { \var \store \expr }
        {\pk2 \var \mapsto \val \, *\, \PropB \pk2}
    }
    \label{lrule:hoare-store}
  \end{proofrules} \vspace{-.5em}
  \begin{proofrules}
    \infer*[lab = hoare-measure]{}{
      \hoare{\pk2 \bar{\qit} \mapsto \qket \pk2}[^\var]
        { \var \store \Measure[\bar{\qit}] }
        {\pk2 \textstyle \bigoplus_i \pk2\paren{\pk2
          \var \mapsto i \pk3*\pk3
          \bar{\qit} \mapsto \Measure^{(i)} \nk2 \qket
          \pk2} \pk2}
    }
    \label{lrule:hoare-measure}
  \end{proofrules} \vspace{-.5em}
  \begin{proofrules}
    \infer*[lab = hoare-seq]{
      \hoare{ \Prop }{ \Cmd }{ \PropB }
    \hspace{1em}
      \hoare{ \PropB }{ \Cmd' }{ \PropC }
    }{
      \hoare{ \Prop }{ \Cmd; \Cmd' }{ \PropC }
    }
    \label{lrule:hoare-seq}

    \infer*[lab = hoare-if]{
      \hoare{ \Prop }{ \expr }
        {\pk1 \rtok{i} \pk1*\pk1 \PropB \pk1}
    \hspace{1.5em}
      \hoare{ \PropB }{ \Cmd_i }{ \PropC }
    }{
      \hoare{ \Prop }{ \ifelse{\expr}{\Cmd_1}{\Cmd_0} }{ \PropC }
    }
    \label{lrule:hoare-if}
  \end{proofrules} \vspace{.05em}
  \begin{proofrules}
    \infer*[lab = hoare-while]{
      \all{n} \pk1 \hoare{\pk2 \Prop_n \pk2}{ \expr }
        {\pk2 (\rtok{0} \pk1*\pk1 \PropB_n) \pk1\oplus\pk1
          (\rtok{1} \pk1*\pk1 \PropC_n) \pk2}
    \hspace{1.5em}
      \all{n} \pk1 \hoare{ \PropC_n }{ \Cmd }{ \Prop_{n + 1} }
    }{ \textstyle
      \hoare{\pk2 \Prop_0 \pk2}{ \whilex{\expr}{\Cmd} }
        {\pk2 \bigoplus_{n \in \NN} \PropB_n \pk2}
    }
    \label{lrule:hoare-while}

    \infer*[lab = hoare-frame]{
      \hoare{ \Prop }{ \Cmd }{ \PropB }
    }{
      \hoare{\pk2 \Prop * \PropC \pk2}{ \Cmd }{\pk2 \PropB * \PropC \pk2}
    }
    \label{lrule:hoare-frame}
  \end{proofrules} \vspace{.05em}
  \begin{proofrules}
    \infer*[lab = hoare-bigbmix]{
      \all{x \in A} \pk2
      \hoare{ \Prop_x }{ \Cmd }{ \PropB_x }
    }{ \textstyle
      \hoare{\pk2 \bigoplus_{x \in I} \Prop_x \pk2}{ \Cmd }{\pk2 \bigoplus_{x \in I} \PropB_x \pk2}
    }
    \label{lrule:hoare-bigbmix}

    \infer*[lab = hoare-post]{
      \hoare{ \Prop }{ \Cmd }{ \PropB }
    \hspace{1em}
      \PropB \vdash \PropB'
    }{
      \hoare{ \Prop }{ \Cmd }{ \PropB' }
    }
    \label{lrule:hoare-post}

    \infer*[lab = hoare-scale]{
      \hoare{ \Prop }{ \Cmd }{ \PropB }
    }{
      \hoare{\pk1 \scl{\alpha}{\Prop} \pk1}{ \Cmd }{\pk1 \scl{\alpha}{\PropB} \pk1}
    }
    \label{lrule:hoare-scale}
  \end{proofrules}
  \caption{Hoare triples and their proof rules.}
  \label{fig:hoare}
  \Description{}
\end{mathfig}

Now we are ready to introduce Hoare triples, as summarized in \cref{fig:hoare}.
We first define Hoare triples over pure expressions $\expr$.
For that, we introduce the \emph{logical semantics} $\msem{\pk1 \expr \pk1}(\mult) \pk1\in\pk1 \Mult(\Res)$ for $\mult \in \Mult(\Res)$
that adds the value of the expression
to the store at the special variable $\ret$ in every outcome.
This induces a notion of weakest precondition $\WP{\expr}{\Prop} \in \SLProp$, and we can derive the Hoare triple $\hoare{\Prop}{\expr}{\PropB}$ as an entailment towards the weakest precondition.
The definition satisfies the expected rules for triples.

We now turn to Hoare triples over \emph{commands} $\Cmd$.
We introduce an auxiliary function $\Leaves(\tree) \pk1\in\pk1 \Mult(\Res)$ for $\tree \in \BTree$, which collects all the leaves of the tree $\tree$ as a multiset of resources.
Then we define the \emph{logical semantics} $\msem{\pk1 \Cmd \pk1}(\mult) \pk1\in\pk1 \Mult(\Res)$ from the denotational semantics $\sem{\pk1 \Cmd \pk1}(\qket, \Store)$ by simply collecting all the leaves starting from each outcome in the input multiset.
Then the weakest precondition $\WP{\Cmd}{\Prop} \in \SLProp$ and the Hoare triple $\hoare{\Prop}{\Cmd}{\PropB}$ are derived as expected.
Notably, this Hoare triple provides a probabilistic version of total correctness, collecting the results of terminating branches only.
For utility, we put classical variables $\bar{\var}$ as the superscript of pre- or postconditions of Hoare triples to describe ownership of them $\vtok{\bar{\var}}$.

The Hoare triples satisfy the expected proof rules.
The rules for quantum operations \ref{lrule:hoare-unitary}, \ref{lrule:hoare-measure} and storing \ref{lrule:hoare-store} follow immediately from the model.
The sequential execution also satisfies the natural chain rule \ref{lrule:hoare-seq}.
Notably, the Hoare triple of \ourlogic{} naturally satisfies the frame rule \ref{lrule:hoare-frame}, the core source of modularity of separation logic, without any side conditions.
The scale rule \ref{lrule:hoare-scale} immediately follows from the frame rule.
Also, executions can be combined using bare mixing \ref{lrule:hoare-bigbmix}.
For example, to handle a conditional where both booleans may occur,
we can first prove
$\hoare{\Prop_i}{\expr}{\rtok{i} * \PropB_i}$
and $\hoare{\PropB_i}{\Cmd_i}{\PropC_i}$ for both $i \in \{0, 1\}$,
apply \ref{lrule:hoare-if},
and mix them with \ref{lrule:hoare-bigbmix} to obtain
$\hoare{\Prop_0 \oplus \Prop_1}
       {\ifelse{\expr}{\Cmd_1}{\Cmd_0}}
       {\PropC_0 \oplus \PropC_1}$.

\paragraph{Verifying loops generally}
\label{sect:logic:program:loop}

The key technical idea of \ourlogic{} is to represent, in the model, mixed states as \emph{multisets of outcomes}.
The adoption of multisets allows \ourlogic{} to support a very general
and expressive handling of \emph{unbounded loops} by the rule \ref{lrule:hoare-while}.
Here, we benefit from supporting the \emph{infinitary} mixing $\bigoplus_{n \in \NN} \Prop_n$, modelled as the infinite sum of multisets described by $\Prop_n$,
that is, the limit of all finite mixing $\bigoplus_{n \le N} \Prop_n$.
The rule uses two assertions indexed by the current iteration $n$:
$\Prop_n$ describes the state at the beginning of the $(n + 1)$-th iteration
(and consequently, at the end of the $n$-th);
$\PropB_n$ describes the states the system can be in at the exit of the loop.
Then the rule asserts that the postcondition of the whole loop is simply the infinite mixing, or the limit of all the finite mixing, of $\PropB_n$ assertions.
Although \ourlogic{} provides natural proof principles for reasoning about such infinite mixing within the logic, which works for useful examples,
more thoroughly exploring the elimination principles of $\Mix_{n \in \NN}$ is left for future work.

Notably, this achieves reasoning about \emph{probabilistic total correctness}.
Since all the outcomes need to be accounted for when predicating over
such resources, assertions can insist on the global probability mass
being some definite quantity by looking at the norms of all the outcomes.
This allows us to specify \emph{almost-sure termination},
or even specify the \emph{exact probability} of termination when non-termination can happen with non-zero probability.

\subsection{The Sum, the Heart of \ourlogic}
\label{sect:logic:sum}

\begin{mathfig}
  \begin{gather*}
    (\Qit \mapsto \qket) \pk2+_\Qstate\pk2 (\QitB \mapsto \qketB) \pk8\defeq\pk8
      \Qit \mapsto (\qket + \qketB)
      \hspace{.9em} \text{if\, $\Qit = \QitB$}
  \\[.3em]
    (\qket, \Store) \pk2+_\Res\pk2 (\qketB, \Store') \pk8\defeq\pk8
      (\qket +_\Qstate \qketB, \pk1 \Store)
      \hspace{.9em} \text{if\, $\Store = \Store'$}
  \end{gather*}
  \caption{Resource ring structure over $\Qstate$ and $\Res$.}
  \label{fig:resource-ring-res}
  \Description{}
\end{mathfig}

Now we also want the sum $\Prop + \PropB$ over SL assertions, the heart of \ourlogic.
Extending the PCM structure for separating conjunction $*$, we introduce what we dub the \emph{resource ring}, a new algebra equipped with the \emph{partial} sum $+$.

\begin{definition}[Resource ring]
\label{def:resource-ring}
  A resource ring $(A, \pk2 \unit_A \in A, \pk2 \cdot_A \col A \times A \pto A, \pk2 +_A \col A \times A \pto A)$ is a PCM $(A, \unit_A, \cdot_A)$ equipped with the partial sum $+_A$ (we may omit the subscript) such that $a +_A b \pk1=\pk1 b +_A a$, $(a +_A b) +_A c \pk1=\pk1 a +_A (b +_A c)$, and $(a +_A b) \cdot_A c \pk1=\pk1 a \cdot_A c +_A b \cdot_A c$.\footnote{We use `$=$' for the Kleene equality, where the left-hand side is defined if and only if the right-hand side is.}
\qed
\end{definition}

We can define the resource ring structure over $\Qstate$ and $\Res$ as presented in \cref{fig:resource-ring-res}.
As expected, the vector sum $\qket + \qketB$ is used for defining $+_\Qstate$.
The point is that the sum is defined only when the two arguments agree on the same classical information, namely, the domain of qubits and the store, which makes the sum partial.

Now, how can we define the sum $\Prop + \PropB$ over SL assertions?
Again, recall that SL assertions are a set of \emph{multisets} of resources.
Naively, we want something like the `parallel pointwise sum' of multisets $\mult$ and $\mult'$ as the sum for the resource ring.
For example, something like $\mset{\res, \res'} + \mset{\resB, \resB'} \pk1=\pk1 \mset{\res + \resB, \pk1 \res' + \resB'}$.
However, this makes the sum ill-defined, because the multiset forgets the `order' of the elements.
For the example above, the sum can also be $\mset{\res + \resB', \pk1 \res' + \resB}$.
In other words, to take the sum, the elements of the two multisets should be matched \emph{in parallel}, and there can be multiple ways to match the elements of two multisets.
This is quite different from separating conjunction $*$, where the product is just distributive over the multiset elements.

To formalize the matching over multisets, we introduce the notion of the \emph{multiset bijection}.

\begin{definition}[Multiset bijection]
\label{def:multiset-biject}
  A multiset bijection $\mbij \col \mult \mbiject \mult'$ between two multisets $\mult \in \Mult(A)$ and $\mult' \in \Mult(B)$ is a multiset $\mbij \in \Mult(A \times B)$ such that $\mult = \mset{ a \mid (a, b) \mIn \mbij }$ and $\mult' = \mset{ b \mid (a, b) \mIn \mbij }$.
  In other words, a multiset bijection between $\floor{f \col I \to A}$ and $\floor{g \col I \to B}$ is $\floor{h \col I \to A \times B}$ such that $(h\, i).0 = f\, i$ and $(h\, i).1 = g\, i$ (and there is no multiset bijection between multisets of different sizes).
\qed
\end{definition}

\begin{mathfig}
  \begin{gather*}
    \mult +_\mbij \mult' \pk8\defeq\pk8
      \mset{ \res + \resB \mid (\res, \resB) \mIn \mult }
    \hspace{.9em} \text{where $\mbij \col \mult \mbiject \mult'$}
  \\[.0em]
    \Prop + \PropB \pk8\defeq\pk8
      \set{ \mult +_\mbij \mult' }{
        \mult \in \Prop, \, \mult' \in \PropB, \,
        \mbij \col \mult \mbiject \mult' }
  \end{gather*} \vspace{-.3em}
  \begin{proofrules}
    \infer*[lab = sum-mono]{
      \Prop \vdash \Prop'
    \hspace{1.5em}
      \PropB \vdash \PropB'
    }{
      \Prop + \PropB \, \vdash\, \Prop' + \PropB'
    }
    \label{lrule:sum-mono}

    \infer*[lab = sum-comm]{}{
      \Prop + \PropB \, \dashvdash\, \PropB + \Prop
    }
    \label{lrule:sum-comm}

    \infer*[lab = sum-assoc]{}{
      (\Prop + \PropB) + \PropC \, \dashvdash\,
        \Prop + (\PropB + \PropC)
    }
    \label{lrule:sum-assoc}
  \end{proofrules} \vspace{-.2em}
  \begin{proofrules}
    \infer*[lab = qpoints-sum]{}{
      \bar{\qit} \mapsto \qket \pk3+\pk3
      \bar{\qit} \mapsto \qketB \pk5\dashvdash\pk5
        \bar{\qit} \mapsto (\qket + \qketB)
    }
    \label{lrule:qpoints-sum}
  \end{proofrules} \vspace{-.6em}
  \begin{proofrules}
    \infer*[lab = hoare-sum]{
      \hoare{ \Prop }{ \Cmd }{ \PropB }
    \hspace{1em}
      \hoare{ \Prop' }{ \Cmd }{ \PropB' }
    }{
      \hoare{ \Prop + \Prop' }{ \Cmd }{ \PropB + \PropB' }
    }
    \label{lrule:hoare-sum}

    \infer*[lab = hoare-frame-untangle]{
      \all{\qket} \pk2
      \hoare{\pk1 \bar{\qit} \mapsto \qket \pk1}{ \Cmd }
        {\pk1 \bar{\qit} \mapsto \Unitary \qket \pk1}
    }{
      \hoare{\pk1 (\bar{\qit}, \bar{\qitB}) \mapsto \qketB \pk1}{ \Cmd }
        {\pk1 (\bar{\qit}, \bar{\qitB}) \mapsto \Unitary_{\bar{\qit}} \qketB \pk1}
    }
    \label{lrule:hoare-frame-untangle}
  \end{proofrules} \vspace{.1em}
  \begin{proofrules}
\infer*[lab = sum-bigbmix]{}{ \textstyle
      \bigoplus_{x \in I} (\Prop_x + \PropB_x) \, \vdash\,
      (\bigoplus_{x \in I} \Prop_x) + (\bigoplus_{x \in I} \PropB_x)
    }
    \label{lrule:sum-bigbmix}

    \infer*[lab = sum-frame]{}{
      (\Prop + \PropB) * \PropC \, \vdash\,
        (\Prop * \PropC) + (\PropB * \PropC)
    }
    \label{lrule:sum-frame}
  \end{proofrules}
  \caption{Sum over $\Mult(\Res)$, the sum connective and its proof rules.}
  \label{fig:sum}
  \Description{}
\end{mathfig}

Using multiset bijections, we can define the sum of SL assertions, as presented in \cref{fig:sum}.
First, the partial sum $\mult +_\mbij \mult'$ over multisets can be defined once we fix the multiset bijection $\mbij \col \mult \mbiject \mult'$.
Note that $\Mult(\Res)$ itself does not form a resource ring, due to the extra parameter $\mbij$.
Now we define the sum $\Prop + \PropB$ of SL assertions by taking the sum of multisets using an arbitrary multiset bijection.
The sum satisfies natural rules \ref{lrule:sum-mono}, \ref{lrule:sum-comm} and \ref{lrule:sum-assoc}.
Quantum points-to tokens sum up according to the vector sum (\ref{lrule:qpoints-sum}), as the model suggests.

The core rule of \ourlogic{} is \ref{lrule:hoare-sum}, summing Hoare triples according to $+$.
Intuitively, this holds because quantum programs cannot change what to execute depending on quantum state vectors.
The precondition $\Prop + \Prop'$ takes the sum of two initial states, taken respectively from $\Prop$ and $\Prop'$, that agree on the classical states and can only differ in the quantum states.
The two executions from these initial states give behaviour trees of the same form.
The behaviour tree for the sum initial state can simply be obtained by summing the corresponding leaves of the two behaviour trees.
The reasoning principle \ref{lrule:hoare-sum} is remarkable.
For example, from this rule, \ref{lrule:hoare-scale} and \ref{lrule:hoare-frame}, we can derive the rule \ref{lrule:hoare-frame-untangle} that can reason about a (possibly) entangled state $\qketB$ from the behaviours of the program $\Cmd$ on each input $\qket$, by decomposing $\qketB$ to a \emph{linear combination} of \emph{disentangled states}.\footnote{
  Here, the universal quantification $\forall \qket$ can range just over some basis of a finite size.
}
This generalizes the reasoning showcased in \cref{sect:overview:superposition}.

Also, bare mixing over sum entails sum of bare mixing \ref{lrule:sum-bigbmix} and sum enjoys the frame rule \ref{lrule:sum-frame}.
However, the converses of these two rules do not hold in general and require careful side conditions.
Technically, this comes from the freedom of the multiset bijection $\mbij$ or the way to match the elements in the model of sum $\Prop + \PropB$.
Later, we will discuss the side conditions for the two converses.

\begin{remark}[Right adjoints of $\land$, $*$, $\oplus$ and $+$]
\label{rem:right-adjoint}
  As usual, usual and separating conjunctions $\land$, $*$ have the right adjoints $\limp$ and $\wand$ (magic wand).
  Interestingly, in \ourlogic{}, mixing $\oplus$ and sum $+$ also have the right adjoints $\pine$ and $\cross$.
  See \appendixref{app:sect:logic:right-adjoint} for the details.
\end{remark}

\paragraph{Incompatibility and unambiguity}

\begin{mathfig}
  \begin{gather*}
    (\qket, \Store) \pk1\hash\pk1 (\qketB, \Store') \pk9\defeq\pk9
      \ex{\var \pk1\in\pk1 \dom \Store \pk1\cap\pk1 \dom \Store'} \pk2
      \Store[\var] \pk1\ne\pk1 \Store'[\var]
  \\[.3em]
    \Prop \hash \PropB \pk9\defeq\pk9
      \all{\mult \in \Prop} \pk2 \all{\mult' \in \PropB} \pk2
      \all{\res \mIn \mult} \pk2 \all{\resB \mIn \mult'} \pk2
      \res \hash \resB
  \\[.3em]
    \frac{
      \val \neq \valB
    }{
      \var \mapsto \val \ \hash\ \var \mapsto \valB
    }
  \hspace{2.5em}
    \frac{
      \Prop \hash \PropB
    }{
      \PropB \hash \Prop
    }
  \hspace{2.5em}
    \frac{
      \Prop \hash\, \PropB
    \hspace{1em}
      \PropB' \vdash \PropB
    }{
      \Prop \hash \PropB'
    }
  \hspace{2.5em}
    \frac{
      \Prop \hash \PropB
    }{
      \Prop * \PropC \, \hash\, \PropB
    }
  \hspace{2.5em}
    \frac{
      \all{x} (\Prop_x \hash \PropB)
    }{
      \bigoplus_{x \in I} \Prop_x \, \hash\, \PropB
    }
  \\[.8em]
    \Prop \col \unambig \pk9\defeq\pk9
      \all{\mult \in \Prop} \pk2
      \all{\mset{\res, \resB} \msubseteq \mult} \pk4
      \res \hash \resB
  \hspace{3em}
    \Prop \col \nonnb \pk9\defeq\pk9
      \all{\mult \in \Prop} \pk2 \mult \neq \mset{}
  \\[.4em]
    \emp,\, \bar{\qit} \mapsto \qket,\,
    \var \mapsto \val
    \col \unambig
  \hspace{3em}
    \frac{
      \Prop \col \unambig
    \hspace{1.5em}
      \PropB \vdash \Prop
    }{
      \PropB \col \unambig
    }
  \hspace{3em}
    \frac{
      \Prop, \PropB \col \unambig
    }{
      \Prop * \PropB,\pk2
      \Prop + \PropB
      \col \unambig
    }
  \end{gather*} \vspace{.0em}
  \begin{proofrules}
    \infer*[lab = bigbmix-unambig]{
      \all{x \in I} \pk2
      \paren{\Prop_x \col \unambig}
    \hspace{1.5em}
      \all{x, y \in I \st x \ne y} \pk2
      \Prop_x \hash \Prop_y
    }{ \textstyle
      \bigoplus_{x \in I} \Prop_x \col \unambig
    }
    \label{lrule:bigbmix-unambig}

    \infer*[lab = sum-precise]{
      \Prop, \PropB \colon \precise
    \hspace{1.5em}
      \Prop \colon \unambig
    }{
      \Prop + \PropB \colon \precise
    }
    \label{lrule:sum-precise}
  \end{proofrules}
  \begin{proofrules}

    \infer*[lab = bigbmix-sum]{
      \all{x, y \in I \st x \ne y} \pk2
      \Prop_x \hash \PropB_y
    }{ \textstyle
      (\bigoplus_{x \in I} \Prop_x) \pk2+\pk2
      (\bigoplus_{x \in I} \PropB_x) \pk5\vdash\pk5
        \bigoplus_{x \in I} \pk2 (\Prop_x + \PropB_x)
    }
    \label{lrule:bigbmix-sum}

    \infer*[lab = sum-unframe]{
      \Prop, \PropC \col \unambig
    \hspace{1.5em}
      \PropC \col \precise, \nonnb
    }{
      (\Prop * \PropC) + (\PropB * \PropC) \, \vdash\,
        (\Prop + \PropB) * \PropC
    }
    \label{lrule:sum-unframe}
  \end{proofrules}
  \caption{Incompatibility and unambiguity.}
  \label{fig:incomp-unambig}
  \Description{}
\end{mathfig}

In order to ensure the uniqueness of the multiset bijection in the assertion sum $\Prop + \PropB$, we introduce a notion of \emph{incompatibility} and a derived notion of \emph{unambiguity}, as in \cref{fig:incomp-unambig}.
First, the incompatibility $\res \hash \resB$ over resources $\res, \resB \in \Res$ is defined as having different values for some variable.
Then we naturally lift that to the incompatibility over propositions $\Prop \hash \PropB$.
We have natural proof rules for the incompatibility.
If any pair of the assertions at different indices are incompatible, sum of mixing can be taken in parallel (\ref{lrule:bigbmix-sum}), which is the converse of \ref{lrule:sum-bigbmix}.

Using the incompatibility between resources, we also define the \emph{unambiguity} $\Prop \col \unambig$ of an SL assertion $\Prop$, meaning that for any multiset of the assertion, any two distinct elements of the multiset are incompatible with each other.
Intuitively, this means all the branches have different stored values.
Here, we introduce an auxiliary predicate $\Prop \col \nonnb$, meaning that any multiset in $\Prop$ is not empty (\ie contains some behaviours).
The predicate satisfies natural proof rules (see \appendixref{app:sect:logic:rules} for the details).
The rule \ref{lrule:sum-unframe} has three side conditions:
the precision of $\PropC$, the behaviour non-emptiness ($\nonnb$) of $\PropC$, and the unambiguity of the assertions $\Prop$ and $\PropC$.
The precision of $\PropC$ is needed for the same reason as \ref{lrule:bigbmix-unframe}.
The behaviour non-emptiness excludes a subtle corner case.
For example, if $\PropC = \nb$ (violating $\PropC \col \nonnb$), $\Prop = \var \mapsto 0$ and $\PropB = \var \mapsto 1$, then the left-hand side is equivalent to $\nb$, while the right-hand side entails the falsehood $\bot \defeq \empset$ because $\Prop + \PropB$ entails $\bot$.
The key side condition is the unambiguity of $\Prop$ and $\PropC$.
For example, if $\PropC = \alpha \oplus \beta$ (violating $\PropC \col \unambig$) and
$\Prop = \PropB = \emp$, the left-hand side contains $(\alpha + \beta) \pk1\oplus\pk1 (\beta + \alpha)$, while the right-hand side does not.
We similarly have unsoundness if both $\Prop$ and $\PropB$ are ambiguous.\footnote{
  For example, let us set $\Prop \pk1=\pk1 \PropB \pk2=\pk2 \lift{\alpha} \pk2\oplus\pk2 \lift{\beta}$ (neither $\Prop \col \unambig$ nor $\PropB \col \unambig$ holds) and $\PropC \pk1=\pk1 \hvar \mapsto 0 \oplus \hvar \mapsto 1$.
  Then the left-hand side of \ref{lrule:sum-unframe} contains
  $\hvar \mapsto 0 \pk2*\pk2 \paren{\lift{\alpha + \beta} \oplus \lift{\beta + \alpha}} \pk3\oplus\pk3
    \hvar \mapsto 1 \pk2*\pk2 \paren{\lift{\alpha + \alpha} \oplus \lift{\beta + \beta}}$
  while the right-hand side does not.
}
The rule \ref{lrule:sum-precise} also requires the unambiguity for a similar reason.
We carefully designed the side conditions to achieve both soundness and flexibility.

\subsection{Tagged Mixing}
\label{sect:logic:mix}

\begin{mathfig}
  \[
    \Prop \mix{\hvar} \PropB \pk6\defeq\pk6
      \paren{\hvar \mapsto 0 \pk2*\pk2 \Prop} \pk3\oplus\pk3
      \paren{\hvar \mapsto 1 \pk2*\pk2 \PropB}
  \hspace{3em} \textstyle
    \bigoplus^\hvar_{x \in I} \Prop_x \pk6\defeq\pk6
      \bigoplus_{x \in I} \pk2 \paren{\hvar \mapsto x \pk2*\pk2 \Prop_x}
  \]\vspace{-.2em}
  \begin{proofrules}
    \infer*[lab = bigmix-sum]{}{ \textstyle
      (\bigoplus^\hvar_{x \in I} \Prop_x) \pk2+\pk2
      (\bigoplus^\hvar_{x \in I} \PropB_x) \pk5\dashvdash\pk5
        \bigoplus^\hvar_{x \in I} \pk2 (\Prop_x + \PropB_x)
    }
    \label{lrule:bigmix-sum}

    \infer*[lab = bigmix-unambig]{
      \all{x \in I} \pk2
      \paren{\Prop_x \col \unambig}
    }{ \textstyle
      \bigoplus^\hvar_{x \in I} \Prop_x \col \unambig
    }
    \label{lrule:bigmix-unambig}
  \end{proofrules} \vspace{.1em}
  \begin{proofrules}
    \infer*[lab = hoare-measure-mix]{}{
      \hoare{\pk3 \bar{\qit} \mapsto \qket \pk2}[^\hvar]
        { \Measure^\hvar[\bar{\qit}] }
        {\pk3 \textstyle \bigoplus^\hvar_i \pk3
          \bar{\qit} \mapsto \Measure^{(i)} \nk2 \qket \pk2}
    }
    \label{lrule:hoare-measure-mix}
  \end{proofrules}
  \caption{Tagged mixing and its derived proof rules.}
  \label{fig:mix}
  \Description{}
\end{mathfig}

Finally, we formally model the \emph{tagged mixing}, introduced in the overview \cref{sect:overview:measurement}.
\Cref{fig:mix} shows its definition and proof rules.
Generalizing binary tagged mixing $\Prop \mix{\hvar} \PropB$, we also introduce tagged mixing $\bigoplus^\hvar_{x \in I} \Prop_x$ indexed over any set $I$.
Tagged mixing is derived from bare mixing \cref{sect:logic:bmix} by tagging each argument with a classical points-to token $\hvar \mapsto x$
whose value $x$ allows for the identification of the outcomes.
Tagged mixing is very well-behaved with respect to the other connectives.
Remarkably, sum $+$ can be taken over tagged mixing in parallel (\ref{lrule:bigmix-sum}), thanks to the classical points-to token automatically ensuring the incompatibility conditions of \ref{lrule:bigbmix-sum}.
This is one of the key ideas of \ourlogic{}, as explained in \cref{sect:overview:measurement}.
Also, tagged mixing is unambiguous simply when the arguments are all unambiguous (\ref{lrule:bigmix-unambig}).
Notably, the Hoare rule for measurement (\ref{lrule:hoare-measure}) can be reformulated using tagged mixing (\ref{lrule:hoare-measure-mix}).
Tagged mixing also enjoys natural proof rules derived from rules on bare mixing;
see \appendixref{app:sect:logic:mix} for the details.
In summary, using tagged mixing, we can enjoy natural reasoning about mixed quantum states introduced by measurements.

\subsection{Handling Complexity}
\label{sect:logic:complexity}

\paragraph{Frameability}

\begin{mathfig}
  \begin{gather*} \textstyle
    \Prop \col \frameable \pk8\defeq\pk8
      \Prop \col \precise, \pk2 \unambig, \pk2 \nonnb
  \\[.3em]
    \frac{
      \Prop \col \frameable
    \quad
      \PropB \vdash \Prop
    }{
      \PropB \col \frameable
    }
  \hspace{3em}
    \frac{
      \Prop, \PropB \col \frameable
    }{
      \Prop * \PropB,\ \Prop + \PropB \col \frameable
    }
  \hspace{3em}
    \frac{
      I \neq \empset
    \quad
      \all{i \in I}\,
      \Prop_i \col \frameable
    }{
      \bigoplus^\iota_{i\in I} \Prop_i \col \frameable
    }
  \\[.6em]
    \begin{aligned}
      \text{Frameable subclass} \quad
      \dFProp \pk2\ni\pk2 \FProp
      &\sdefeq \var \mapsto \val
      \sor \qit \mapsto \qket
      \sor \all{x \in A} \FProp_x
    \\[-.2em] & \hspace{.9em} \textstyle
      \sor \FProp * \FPropB
      \sor \FProp + \FPropB
      \sor \bigoplusn^\iota_{i\in I (\neq \empset)} \FProp_{i}
    \end{aligned}
  \hspace{3em}
    \frac{
      \FProp \in \dFProp
    }{
      \FProp \col \frameable
    }
  \end{gather*}\vspace{.5em}
  \begin{proofrules}
    \infer*[lab=bigbmix-frame-frameable]{
      \PropB \col \frameable
    }{ \textstyle
      \bigoplus_{x \in I} \pk2 (\Prop_x * \PropB) \, \dashvdash\,
        (\bigoplus_{x \in I} \Prop_x) \pk2*\pk2 \PropB
    }
    \label{lrule:bigbmix-frame-frameable}

    \infer*[lab=sum-frame-frameable]{
\Prop, \PropC \col \frameable
    }{
      (\Prop * \PropC) + (\PropB * \PropC) \, \dashvdash\,
        (\Prop + \PropB) * \PropC
    }
    \label{lrule:sum-frame-frameable}
  \end{proofrules}
  \caption{Frameability.}
  \label{fig:frameable}
  \Description{}
\end{mathfig}

So far, we have introduced various proof rules to \ourlogic.
Our primary goal has been to identify the most general sound rules that achieve the desired modularity.
It is unsurprising that some of them (e.g., \ref{lrule:sum-unframe}) carry intricate side conditions, but one might worry that they are too complex for practical use.

In fact, these complex side conditions can be greatly simplified by focusing on a well-behaved class of SL assertions, which we call \emph{frameable}.
The rules are summarized in \cref{fig:frameable}.
Frameability $\frameable$ is simply the conjunction of $\precise$, $\unambig$ and $\nonnb$.
As shown in \cref{fig:frameable}, frameable assertions have nice closure properties.
This enables us to define a handy class $\dFProp$ of frameable SL assertions, which are obviously frameable by construction.\footnote{
  For simplicity, we describe the class syntactically in BNF.
  The mixing here is tagged and runs over a non-empty domain.
}
Notably, frameable assertions can be a frame for mixing and sum (\ref{lrule:bigbmix-frame-frameable}, \ref{lrule:sum-frame-frameable}), simplifying the unframing rules \ref{lrule:bigbmix-unframe} and \ref{lrule:sum-unframe}.

\paragraph{Abstraction}

We can go further and provide \emph{abstraction}, as discussed in \cref{sect:overview:abstraction}.
We want to abstract away the complicated `factor' $\Prop_{\oplus}$ of the outcome of the program $\mCNOT$, so that clients of $\mCNOT$ can safely forget the exact form of $\Prop_{\oplus}$ and equate $\mCNOT$ with $\gCNOT$.

Frameability is a great match for this purpose.
As long as clients know that $\Prop_{\oplus}$ is frameable, they can treat it as a frame for mixing and sum to achieve the outcome- and basis-locality (we will explain this in more detail later in \cref{sect:cases:meas-cnot}).

\begin{mathfig}
  \[
    \Prop \col \Prob \prob \pk8\defeq\pk8
      \all{\mult \in \Prop}
      \sum_{(\qket, \Store) \pk1\mIn\pk1 \mult}\nk5 \norm{\qket}^2 \pk2=\pk2 \prob
  \hspace{2em}
    \frac{
      \Prop \col \Prob \prob
    \hspace{1em}
      \PropB \col \Prob \probB
    }{
      \Prop \pk1*\pk1 \PropB \col \Prob \prob \probB
    }
  \hspace{2em}
    \frac{
      \Prop_i \col \Prob \prob_i
    }{
      \bigoplusn_i \Prop_i \col \Prob \sum_i\prob_i
    }
  \]
  \caption{Probability.}
  \label{fig:prob}
  \Description{}
\end{mathfig}

For a detailed analysis of probability, we introduce the probability predicate $\Prop \col \Prob \prob$, shown in \cref{fig:prob}.\footnote{
  Recall that \ourlogic{} can verify the termination probability, which is non-trivial for programs with loops, as discussed in \nameref{sect:logic:program:loop}, \cref{sect:logic:program}.
}
The predicate simply means that every state represented by $\Prop$ has the total probability $\prob$.
A factor like $\mCNOT$'s $\Prop_\oplus$ typically satisfies $\Prob 1$.
So in a simple setting we can abstract such a factor as $\frameable$ and $\Prob 1$.
The probability behaves well over separating conjunction $*$ and mixing $\bigoplus$.
To reason about the probability of sums $+$, we can also introduce the inner product (and orthogonality as a special case) of SL assertions, naturally extending the usual vector calculus; see \appendixref{app:sect:logic:inner-prod} for the details.
 \section{Case Studies}
\label{sect:cases}

This section presents a number of practical examples of \ourlogic{} verification of quantum programs from the literature,
confirming that our logic:
\begin{enumerate}[label = \textbf{G\arabic*:}, ref = \textbf{G\arabic*}]
\item \label{g:modular}
  Enables modular reasoning by using the three locality principles;
\item \label{g:sum-mixing}
  Effectively applies basis-locality to real-world programs involving measurements;
\item \label{g:hiding}
  Demonstrates that abstracting the global phase, as presented in \cref{sect:overview:abstraction}, simplifies the specification and makes the reasoning scalable;
\item \label{g:prob_while}
  Supports standard probabilistic features of quantum programs, such as proving almost-sure termination of while loops.
\end{enumerate}
We start with relatively simple case studies and gradually move to more complex ones.
All the details of the proofs and additional examples are available in \appendixref{app:sect:cases}.
Here, we sketch the key ideas of the proofs and highlight the significance of the examples.

\begin{remark}[Utility tagged mixing notation]
  For brevity, we introduce the following shorthand for tagged mixing, abusing a classical variable $\hvar$ to represent the value it stores:
  $\bigoplus^\hvar \Prop \pk6\defeq\pk6
    \bigoplus^\hvar_x \Prop[x / \hvar]$.
  Here, the domain of the value $x$ should be properly inferred and is typically set to $\{0, 1\}$.
  Precisely speaking, we think of some syntax for propositions $\Prop$ to define the substitution $\Prop[x / \hvar]$ here.
  Also, we write $\bigoplus^{\hvar_1, \ldots, \hvar_n} \Prop$ for $\bigoplus^{\hvar_1} \cdots \bigoplus^{\hvar_n} \Prop$.
\end{remark}

\subsection{Dirty Qubit: Implementation of CCCX by Toffoli Gates}
\label{sect:cases:dirty}

{\setlength\textfloatsep{0pt}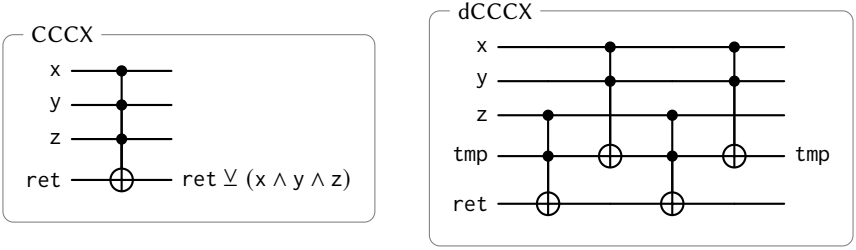
\begin{figure}[t]
  \adjustfigure[\small]
  \begin{quantikz}[with border, pic label = \gCCCX, row sep = 0.3cm]
    \lstick{\qit}    & \ctrl{3} & \\
    \lstick{\qitB}   & \ctrl{2} & \\
    \lstick{\qitC}   & \ctrl{1} & \\
    \lstick{\qitRet} & \targ{ } & \rstick{$\qitRet \xor (\qit \land \qitB \land \qitC)$}
  \end{quantikz}
  \qquad
  \begin{quantikz}[with border, pic label = \gdCCCX, row sep = 0.3cm]
    \lstick{\qit}    &          & \ctrl{3} &          & \ctrl{3} & \\
    \lstick{\qitB}   &          & \ctrl{2} &          & \ctrl{2} & \\
    \lstick{\qitC}   & \ctrl{2} &          & \ctrl{2} &          & \\
    \lstick{\qitTmp} & \ctrl{1} & \targ{ } & \ctrl{1} & \targ{ } & \rstick{\qitTmp} \\
    \lstick{\qitRet} & \targ{ } &          & \targ{ } &          &
  \end{quantikz}
  \caption{A \gCCCX{} gate (left) and its encoding \gdCCCX{} using dirty qubit \qitTmp{} (right).}
  \label{fig:dcccx-simple}
  \Description{}
\end{figure}}

The first example is an implementation of the \gCCCX{} gate using a dirty qubit, which is a concrete example of the program we mentioned in \cref{sect:overview:superposition}.
We show how \ourlogic{} enables the context-independent identification of
\gCCCX{} gate with its implementation (\ref{g:modular}).

As depicted in \cref{fig:dcccx-simple} (left), the \gCCCX{} gate is a 4-qubit multi-controlled-NOT gate that computes $\qitRet \leftarrow \qitRet \xor(\qit \land \qitB \land \qitC)$, where $\xor$ is the exclusive-OR operation.
The \gCCCX{} gate can be implemented using the \gCCX{} (called Toffoli) gate with a dirty qubit serving as the auxiliary qubit, as illustrated in \cref{fig:dcccx-simple} (right).
Note that the \gCCX{} gate computes $\qitRet \leftarrow \qitRet \xor (\qit \land \qitB)$.

It is important is that we must return a dirty qubit to its original state after its temporary use, because it may be required by another computation.
Thus, the goal here is to verify that \cref{fig:dcccx-simple} (right) enjoys this condition and implements \gCCCX{} correctly;
however, two challenges arise here.
First, the state of \qitTmp{} is unknown. Second, \qitTmp{} may be entangled with \qit{}, \qitB{}, \qitC{}, or other qubits.

The idea for overcoming these challenges is to decompose \qitTmp{} into a basis using \ref{lrule:hoare-sum}.
After the decomposition, \qitTmp{} can be regarded as disentangled from the other qubits for each basis state.
Consequently, the qubits other than \qit{}, \qitB{}, \qitC{} and \qitTmp{} can later be composed by using the frame rule.
Moreover, we can obtain reasoning that is independent of the initial state of \qitTmp{} because the reasoning results for each basis state can be linearly combined.

In fact, this idea makes it possible to verify the correctness of \cref{fig:dcccx-simple} (right).
First, the specification of \cref{fig:dcccx-simple} (right) for each basis can be expressed as
\begin{align*} &
  \pkcurly[\big]{
    (\qit, \qitB, \qitC) \mapsto \ket{i j k}
    * \qitTmp \mapsto \ket{\ell}
    * \qitRet \mapsto \ket{m}
  }
  \ \ {\gdCCCX[\qit, \qitB, \qitC, \qitTmp, \qitRet]}
\\ &
  \pkcurly[\big]{
    (\qit, \qitB, \qitC) \mapsto \ket{i j k}
    * \qitTmp \mapsto \ket{\ell}
    * \qitRet \mapsto \ket{m \xor (i \land j \land k)}
  }.
\end{align*}
Then, we scale the global phase \labelcref{dcccx:frame} and sum up derivations \labelcref{dcccx:sum} to prove for any entangled case:
\begin{derivation}
  \infer*[Right = {\labelstep{dcccx:frame}}]{
      \all{i, j, k, \ell, m \in \{0, 1\}}
    \\
      \pkcurly{ (\qit, \qitB, \qitC, \qitTmp, \qitRet) \mapsto \ket{i j k \ell m} }
          \pk8{ \gdCCCX }\pk8
    \hspace{6em}\\\\\hspace{10em}
      \pkcurly{ (\qit, \qitB, \qitC, \qitTmp, \qitRet) \mapsto
                (\gCCCX_{\qit, \qitB, \qitC, \qitRet} \otimes \gid_\qitTmp) \ket{i j k \ell m} }
    }{
  \infer*[Right = {\labelstep{dcccx:sum}}]{
      \all{i, j, k, \ell, m \in \{0, 1\}}
    \\\\
      \pkcurly{
        (\qit, \qitB, \qitC, \qitTmp, \qitRet) \mapsto
          \alpha_{i j k \ell m} \ket{i j k \ell m} *
        \bar{\qitD} \mapsto \qket[_{i j k \ell m}]
      }
      \ \ \gdCCCX
    \hspace{6em}\\
      \pkcurly{
        (\qit, \qitB, \qitC, \qitTmp, \qitRet) \mapsto
          (\gCCCX_{\qit, \qitB, \qitC, \qitRet} \otimes \gid_\qitTmp) \alpha_{i j k \ell m} \ket{i j k \ell m}
        * \bar{\qitD} \mapsto \qket[_{i j k \ell m}]
      }
    }{ \textstyle
    \pkcurly[Big]{
      (\qit, \qitB, \qitC, \qitTmp, \qitRet, \bar{\qitD}) \mapsto
        \sum_{i, j, k, \ell, m} \alpha_{i j k \ell m} \ket{i j k \ell m} \qket[_{i j k \ell m}]
    }
    \ \ \ \gdCCCX
  \hspace{6.5em}\\\textstyle
    \pkcurly[Big]{
      (\qit, \qitB, \qitC, \qitTmp, \qitRet, \bar{\qitD}) \mapsto
        (\gCCCX_{\qit, \qitB, \qitC, \qitRet} \otimes \gid_{\qitTmp, \bar{\qitD}})
        \sum_{i, j, k, \ell, m} \alpha_{i j k \ell m} \ket{i j k \ell m} \qket[_{i j k \ell m}]
    }
  }}
\end{derivation}
The state
$\sum_{i, j, k, \ell, m} \alpha_{i j k \ell m} \ket{i j k \ell m} \qket[_{i j k \ell m}]$
can represent any entangled state of the qubits, not only including the dirty qubit \qitTmp{},
but also uninvolved qubits $\bar{\qitD}$ introduced as a frame,
proving that the \gdCCCX{} gate is indeed a valid implementation of \gCCCX{} regardless of the context.

\subsection{Quantum Teleportation}
\label{sect:cases:teleportation}

The second example we demonstrate is the quantum teleportation protocol.
In this protocol, there are two parties: Alice and Bob,
and they communicate two bits of classical information
to teleport one qubit of quantum information from Alice to Bob.
Here, we prove the correctness of the protocol in a modular way (\ref{g:modular}).
That is, we analyse the behaviour of Alice and Bob separately
and then combine the results to prove the correctness of the whole protocol.
This protocol involves the use of a shared entangled state and measurements (\ref{g:sum-mixing}).

\newcommand*{\Bell}{\mathsf{Bell}}
\newcommand*{\Alice}{\mathsf{Alice}}
\newcommand*{\Bob}{\mathsf{Bob}}
\newcommand*{\Teleport}{\mathsf{Teleport}}
The quantum teleportation protocol can be implemented as follows:
\begin{align*}
  \Bell[\qit, \qitB] \ \ &\defeq\ \
    \gH[\qit];\, \gCX[\qit, \qitB]
\\
  \Alice^{\var, \varB}[\qit, \qitB] \ \ &\defeq\ \
    \gCX[\qit, \qitB];\, \gH[\qit];\,
    \var \store \gMZ[\qit];\, \varB \store \gMZ[\qitB]
\\
  \Bob^{\var, \varB}[\qitC] \ \ &\defeq\ \
    \ifonly{ \var }{ \gX[\qitC] };\ \ifonly{ \varB }{ \gZ[\qitC] }
\\
  \Teleport^{\var, \varB}[\qit, \qitB, \qitC] \ \ &\defeq\ \
    \Bell[\qitB, \qitC];\, \Alice^{\var, \varB}[\qit, \qitB];\, \Bob^{\var, \varB}[\qitC].
\end{align*}
The protocol begins by preparing the Bell state $\ketBell \defeq \tfrac{1}{\sqrt{2}}(\ket{00} + \ket{11})$
shared between Alice and Bob by using $\Bell[\qitB, \qitC]$.
Although Alice's qubit $\qitB$ and Bob's qubit $\qitC$ are entangled,
in \ourlogic{}, we can reason about $\Alice^{\var, \varB}[\qit, \qitB]$ and $\Bob^{\var, \varB}[\qitC]$ separately
by considering the case in which $\qitB$ and $\qitC$ are in some classical state.
Consequently, if we introduce a Boolean variable $x$ to represent the basis of the qubit $\qit$, we obtain:\footnote{
  We write $\delta_{a, b}$ for the Kronecker delta, returning $1$ if $a = b$ and $0$ otherwise.
}
\begin{gather*}
  \hoare{ (\qit, \qitB) \mapsto \ket{x i} }[^{\var, \varB}]
    { \Alice^{\var,\varB}[\qit, \qitB] }
    { \tfrac{1}{\sqrt2} \bigoplusn^{\var, \varB}
      (-1)^{x \var} \cdot \delta_{x \xor i, \varB} \cdot
      (\qit, \qitB) \mapsto \ket{\var \varB} }
\\
  \hoare{ \qitC \mapsto \ket{i} \,*\, \var \mapsto a \,*\, \varB \mapsto b }
    { \Bob^{\var, \varB}[\qitC] }
    { (-1)^{xa} (\qitC \mapsto \ket{b \xor i}) \,*\,
      \var \mapsto a \,*\, \varB \mapsto b }
\end{gather*}
Now, we verify the correctness of the whole protocol.
\begin{derivation}
  \infer*[Right = {\labelstep{teleport:bell-sum}}]{
    \hoare{ (\qit, \qitB, \qitC) \mapsto \ket{x i i} }[^{\var, \varB}]
          { \Alice[\qit, \qitB];\, \Bob^{\var, \varB}[\qitC] }
          { \tfrac{1}{\sqrt2} \bigoplusn^{\var, \varB} \delta_{x \xor i, \varB} * (\qit, \qitB, \qitC) \mapsto \ket{\var \varB (\varB \xor i)} }
    }{
  \infer*[Right = {\ \labelstep{teleport:input-sum}}]{
    \hoare{ (\qit, \qitB, \qitC) \mapsto \sumn_{i = 0, 1} \tfrac{1}{\sqrt2} \ket{x i i} }[^{\var, \varB}]
          { \Alice[\qit, \qitB];\, \Bob^{\var, \varB}[\qitC] }
          { \tfrac{1}{2} \bigoplusn^{\var, \varB} (\qit, \qitB, \qitC) \mapsto \ket{\var \varB x} }
    }{
  \infer*[Right = {\ \ \labelstep{teleport:mix-unframe}}]{
    \hoare{ \qit \mapsto \qket * (\qitB, \qitC) \mapsto \ketBell }[^{\var, \varB}]
          { \Alice[\qit, \qitB];\, \Bob^{\var, \varB}[\qitC] }
          { \tfrac{1}{2} \bigoplusn^{\var, \varB} (\qit, \qitB) \mapsto \ket{\var \varB} * \qitC \mapsto \qket }
    }{
    \hoare{ \qit \mapsto \qket * (\qitB, \qitC) \mapsto \ketBell }[^{\var, \varB}]
      { \Alice[\qit, \qitB]; \Bob^{\var, \varB}[\qitC] }
      { \qitC \mapsto \qket *
        \tfrac{1}{2} \bigoplusn^{\var, \varB} (\qit, \qitB) \mapsto \ket{\var \varB} }
    }}}
\end{derivation}
We frame $\qitC \mapsto \ket{i}$ into Alice's Hoare triple
and frame Alice's postcondition into Bob's Hoare triple to obtain the first line.
At \labelcref{teleport:bell-sum},
we use \ref{lrule:hoare-sum} to combine
them to obtain the Bell state in the precondition.
To generalize the input state $\ket{x}$ to $\qket = \alpha_0 \ket{0} + \alpha_1 \ket{1}$,
we multiply by the scalar $\alpha_{x}$ and sum over $x$ at \labelcref{teleport:input-sum}.
We use \ref{lrule:bigbmix-unframe} at \labelcref{teleport:mix-unframe}
to separate off the state
$\Prop \defeq \tfrac{1}{2}
\bigoplusn^{\var, \varB} (\qit, \qitB) \mapsto \ket{\var \varB} $,
which is $\frameable$ and $\Prob 1$.

\subsection{Lattice Surgery: Implementation of CNOT with Measurements}
\label{sect:cases:meas-cnot}

As we sketched in \cref{sect:overview:abstraction}, we can verify a sort of equivalence
between the \gCNOT{} gate and its implementation using 2-qubit measurements in \cref{fig:mcnot}.
This technique of replacing a \gCNOT{} gate with 1-qubit gates and 1- or 2-qubit measurements
is the core idea of lattice surgery~\cite{FowlerG19-lattice-surgery},
which can be used for fault-tolerant quantum computing.
This example illustrates how \ourlogic{}
achieves the verification of the program with measurements (\ref{g:sum-mixing}),
and how we can hide the information of the global phase that depends on the measurement outcomes
but does not affect the final state (\ref{g:hiding}).

The program we verify, $\mCNOT$, is defined by
\[\begin{aligned}
  \mCNOT^{\hvar, \hvarB, \hvarC}[\qit, \qitB, \qitC] \ \ \defeq\ \ \ &
    \ifonly{ \gMXX^\hvar[\qitB, \qitC] }{ \gZ[\qit] };
\\[-.2em] &
    \ifonly{ \gMZZ^\hvarB[\qit, \qitB] }{ \gX[\qitC] };\,\
    \gH[\qitB];\,\
    \ifonly{ \gMZ^\hvarC[\qitB] }{ \gZ[\qit] }.
\end{aligned}\]
We introduce shorthand $\ifonly{ \Measure^\hvar[\bar{\qit}] }{ \Cmd } \,\defeq\, \Measure^\hvar[\bar{\qit}];\, \ifonly{ \hvar }{ \Cmd }$.
Here, classical variables $\hvar$, $\hvarB$ and $\hvarC$ are used for storing the results of measurements.
Using our logic, we can easily derive the following Hoare triple by step-by-step reasoning:
\begin{align*}
  \all{a, b \in \{0, 1\}} \quad &
  \pkcurly[\big]{
    \qit \mapsto \ket{a} * \qitB \mapsto \ket{0} * \qitC \mapsto \gH \ket{b}
  }^{\hvar, \hvarB, \hvarC}
  \quad
  \mCNOT^{\hvar, \hvarB, \hvarC}[\qit, \qitB, \qitC]
\\ &
  \pkcurly[\big]{
    \tfrac{1}{2 \sqrt2}\, \bigoplusn^{\hvar, \hvarB, \hvarC}
      (-1)^{a b + \hvar \hvarB + \hvarB \hvarC} \cdot
      (\qit, \qitB) \mapsto \ket{a \hvarC} \,*\,
      \qitC \mapsto \gH \ket{b}
  }
\end{align*}
The postcondition can be simplified as follows:
\begin{align*} &
  \tfrac{1}{2 \sqrt2}\, \bigoplusn^{\hvar, \hvarB, \hvarC}
    (-1)^{a b + \hvar \hvarB + \hvarB \hvarC} \cdot
    (\qit, \qitB) \mapsto \ket{a \hvarC} \,*\,
    \qitC \mapsto \gH \ket{b}
\\ \vdash \quad &
  (-1)^{a b} \tfrac{1}{2 \sqrt2} \cdot
  (\qit, \qitC) \mapsto \ket{a} \otimes \gH \ket{b} \,*\,
  \bigoplusn^{\hvar, \hvarB, \hvarC}
    (-1)^{\hvar \hvarB + \hvarB \hvarC} \cdot \qitB \mapsto \ket{\hvarC}
\\ \vdash \quad &
  (\qit, \qitC) \mapsto \gCX\pk2 (\ket{a} \otimes \gH \ket{b}) \,*\,
  \tfrac{1}{2 \sqrt2} \bigoplusn^{\hvar, \hvarB, \hvarC}
    (-1)^{\hvar \hvarB + \hvarB \hvarC} \cdot \qitB \mapsto \ket{\hvarC}.
\end{align*}
Now, let $\Prop_\oplus \defeq \tfrac{1}{2 \sqrt2} \bigoplusn^{\hvar, \hvarB, \hvarC}
  (-1)^{\hvar \hvarB + \hvarB \hvarC} \cdot \qitB \mapsto \ket{\hvarC}$.
It satisfies the nice property $\Prop_\oplus \col \frameable,\pk1 \Prob 1$, \ie it is frameable and its total probability is $1$ (recall \cref{fig:frameable,fig:prob}, \cref{sect:logic:sum}).
With \ref{lrule:hoare-scale} and \ref{rule:hoare-sum}, we can derive the following general specification of $\mCNOT$ for any input state $\qket$:
\[
  \all{\qket}\
  \hoare{\pk2
    (\qit, \qitC) \mapsto \qket \,*\, \qitB \mapsto \ket{0}
  \pk2}[^{\hvar, \hvarB, \hvarC}]
    { \mCNOT^{\hvar, \hvarB, \hvarC}[\qit, \qitB, \qitC] }
    {\pk2 (\qit, \qitC) \mapsto \gCX \qket \,*\, \Prop_\oplus \pk2}.
\]

From this assertion, we can further derive the following, saying that \mCNOT{} behaves like \gCNOT{} even if $\qit$ and $\qitC$ are entangled with any other qubits:
\[
  \all{\qket}\
  \hoare{\pk2
    (\qit, \qitC, \bar{\qitD}) \mapsto \qket \,*\, \qitB \mapsto \ket{0}
  \pk2}[^{\hvar, \hvarB, \hvarC}]
    { \mCNOT^{\hvar, \hvarB, \hvarC}[\qit, \qitB, \qitC] }
    {\pk2 (\qit, \qitC, \bar{\qitD}) \mapsto \gCX_{\qit, \qitC} \qket \,*\, \Prop_\oplus \pk2}.
\]
Remarkably, for the derivation (which uses \ref{lrule:hoare-sum}), it suffices to know that $\Prop_\oplus$ is \emph{some} assertion satisfying $\frameable$, without knowing the exact form of $\Prop_\oplus$, thanks to the fact that frameable assertions can freely frame into and out of sum $+$ (\ref{lrule:sum-frame-frameable}).
We can further enrich this for mixing $\bigoplus$ (using \ref{lrule:hoare-bigbmix} and \ref{lrule:bigbmix-frame-frameable}).
This demonstrates the power of abstraction (\ref{g:hiding}).
Note that, when we have some succeeding program $\Cmd$ and want to prove the whole specification,
just framing $\Prop_\oplus$ suffices:
\begin{derivation}
  \infer*[Right = \ref{lrule:hoare-seq}]{
    \infer*[Right = \ref{lrule:hoare-frame}]{
      \hoare{ (\qit, \qitC, \bar{\qitD}) \mapsto \gCX_{\qit, \qitC} \qket }
        { \Cmd }{ \PropB }
    }{
      \hoare{ (\qit, \qitC, \bar{\qitD}) \mapsto \gCX_{\qit, \qitC} \qket \pk1*\pk1 \Prop_\oplus }
        { \Cmd }{ \PropB \pk1*\pk1 \Prop_\oplus }
    }}{
    \hoare{ (\qit, \qitC, \bar{\qitD}) \mapsto \qket \pk1*\pk1 \qitB \mapsto \ket{0} }[^{\hvar, \hvarB, \hvarC}]
      { \mCNOT^{\hvar, \hvarB, \hvarC}[\qit, \qitB, \qitC];\, \Cmd }
      { \PropB \pk1*\pk1 \Prop_\oplus }
  }
\end{derivation}

\subsection{Error Correction: Bit-Flip Code}
\label{sect:cases:bit-flip-code}

Fault tolerance is essential for achieving
reliable quantum computation, and error correction lies at its core.
From the perspective of program verification,
verifying such quantum error-correcting codes (QECCs) is a crucial task that must be addressed.
In this subsection and the next \cref{sect:cases:shor-code}, we present case studies of verifying two simple but fundamental QECCs using \ourlogic{},
the \emph{bit-flip code} and the \emph{Shor code},
while highlighting how \ourlogic{} achieves our goals (\ref{g:modular}, \ref{g:sum-mixing}, \ref{g:hiding}).

The \emph{bit-flip code} is the simplest QECC
that can correct a single bit-flip error (a.k.a. $\gX$ error), \ie a quantum error that flips one qubit flips from $\ket{0}$ to $\ket{1}$ or vice versa.
In order to make quantum information tolerant to such errors,
we encode a single qubit of quantum information into three actual qubits.
Concretely, we encode the state $\ket{0}$ as $\ket{000}$,
and $\ket{1}$ as $\ket{111}$.
The quantum information we want to protect is called a \emph{logical qubit},
and the three qubits that store the logical qubit are called \emph{physical qubits}.
The state of the physical qubits has some robustness,
and such robustness enables us to detect and correct the bit-flip error.

\tikzset{
  tight/.style = {inner ysep = -1pt},
  error/.style = {tight, dashed}
}
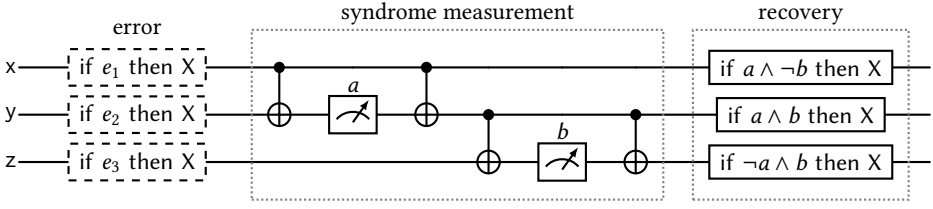
\begin{figure}[t]
  \adjustfigure[\small]
  \begin{quantikz}[row sep = 2pt, style = {inner sep = 0pt}]
    \lstick{\qit} &[.5em]
    \gate[style = error]{ \ifonly{ e_1 }{ \gX } }
      \gategroup[3, style = {draw = none}, label style = {yshift = -.4em}]{error} &[1em]
    \ctrl{1}
      \gategroup[3, steps = 6, style = {inner sep = 6pt, gray, densely dotted}]{syndrome measurement}
    & & \ctrl{1} & & &
    &[1em] \gate[style = tight]{ \ifonly{ \var \land \neg \varB }{ \gX } }
      \gategroup[3, steps = 1, style = {inner sep = 6pt, gray, densely dotted}]{recovery} &
  \\
    \lstick{\qitB} & \gate[style = error]{ \ifonly{ e_2 }{ \gX } } &
    \targ{} & \meter[label style = {yshift = .1em}]{\var} & \targ{} & \ctrl{1} & & \ctrl{1} &
    \gate[style = tight]{ \ifonly{ \var \land \varB }{ \gX } } &
  \\
    \lstick{\qitC} & \gate[style = error]{ \ifonly{ e_3 }{ \gX } } &
    & & & \targ{} & \meter[label style = {yshift = .1em}]{\varB} & \targ & &
    \gate[style = tight]{ \ifonly{ \neg \var \land \varB }{ \gX } } &
  \end{quantikz}
  \caption{Bit-flip code error correction procedure.}
  \label{fig:bitflip-code}
  \Description{}
\end{figure}

\newcommand*{\BitEC}{\mathsf{BitEC}}
The circuit for the bit-flip code error correction procedure is shown in
\cref{fig:bitflip-code}.
The left-most column models the occurrence of bit-flip errors,
where each $\ifonly{e_i}{\gX}$ applies an $\gX$ gate to the qubit if the error $e_i \in \{0, 1\}$ occurs.
We assume $e_1 + e_2 + e_3 \leq 1$ to ensure that at most one qubit is flipped.
The procedure of error correction is as follows:
we first measure the \emph{syndrome}---the effect caused by the error---and then
use the measurement outcomes to recover the logical qubit.
After the recovery procedure,
the syndrome affects only the global phase of the logical qubit,
and thus we can ignore it.

As we did in the previous example,
we can specify and verify the correctness of the bit-flip code error correction procedure $\BitEC$,
using the power of \ourlogic{}:
\[\begin{aligned}
  \ex{\, \Prop \col \frameable,\pk1 \Prob 1}\,\
  \all{\alpha, \beta} \hspace{1em} &
  \pkcurly[\big]{\pk2
     (\qit, \qitB, \qitC) \mapsto
      \gX^{e_1} \gX^{e_2} \gX^{e_3} (\alpha \ket{000} + \beta \ket{111})
  \pk2}^{\var, \varB}
\\[-.2em] & \hspace{2em}
    \BitEC^{\var, \varB}[\qit, \qitB, \qitC] \hspace{1em}
    \pkcurly[\big]{\pk2
      (\qit, \qitB, \qitC) \mapsto \paren{\alpha \ket{000} + \beta \ket{111}} \,*\,
      \Prop
    \pk2}.
\end{aligned}\]
The postcondition $(\qit, \qitB, \qitC) \mapsto (\alpha \ket{000} + \beta \ket{111}) \pk1*\pk1 \Prop$ describes the state of the system after the error correction procedure.
Here, $(\qit, \qitB, \qitC) \mapsto (\alpha \ket{000} + \beta \ket{111})$ captures the restored logical state of the qubits, while $\Prop$ abstracts the effect of the measured syndrome---that contains the difference of global phase induced by measurement outcomes (\ref{g:hiding}).

Crucially, the use of separating conjunction $*$ makes the distinction between the logical state and the syndrome information explicit and formal.
This separation reflects the fact that the syndrome affects only the global phase or auxiliary context, and not the logical content of the qubits.
Such a formulation not only captures the correctness of the procedure but also demonstrates how separation in the logic naturally mirrors separation in the underlying physical system.

\newcommand*{\PhaseEC}{\mathsf{PhaseEC}}
We can also think of a \emph{phase-flip error}, applying the $\gZ$ gate instead of $\gX$.
A phase-flip error-correction code can be obtained immediately from the bit-flip code.
Instead of the $\gZ$ basis $\angl{\ket{0}, \ket{1}}$, the new code uses the $\gX$ basis $\angl{\ket{\plus}, \ket{\minus}}$, encoding a qubit $\ket{\plus}$ as $\ket{\plus\plus\plus}$ and $\ket{\minus}$ as $\ket{\minus\minus\minus}$.
The error correction procedure $\PhaseEC$ can be obtained by adding Hadamard $\gH$ gates before and after $\BitEC$ for the basis transformation.
Its specification can be naturally derived from that of $\BitEC$.

\subsection{Error Correction: Shor's Code}
\label{sect:cases:shor-code}

The \emph{Shor code}~\cite{Shor95-scheme-reducing} is a QECC that
encodes one logical qubit with nine physical qubits and can correct any single-qubit unitary error.
In general, if a QECC can correct both a single $\gX$ error and a single phase-flip error (a.k.a. $\gZ$ error),
it can correct any single-qubit error $\Unitary$, because $\Unitary$ can be expressed as a linear combination of $\gid$, $\gX$, $\gZ$, and $\gX\gZ$.
Shor’s code achieves the correction of both $\gX$ and $\gZ$ errors by concatenating a 3-qubit bit-flip code with a 3-qubit phase-flip code:
first it encodes three $\gZ$-error-tolerant qubits using the 3-qubit phase-flip code, and then applies the 3-qubit bit-flip code to those encoded qubits to obtain a single logical qubit.
Let us denote the logical qubit state that encodes the state $\qket$ in the phase-flip code as $\qket[_\Log]$.
Concretely, the Shor code encodes the logical qubit state $\alpha \ket{0} + \beta \ket{1}$ into
\[
  \alpha \ket{0_\Log 0_\Log 0_\Log} + \beta \ket{1_\Log 1_\Log 1_\Log}.
\]
As a result, any $\gX$ error on the logical qubit can be corrected by the top-level bit-flip code, and any $\gZ$ error can be corrected by the bottom-level phase-flip code.

\newcommand*{\ShorEC}{\mathsf{ShorEC}}
The error correction procedure of the Shor code is roughly given as follows,
where $\PhaseEC$ is the $\gZ$ error-correction procedure defined above and $\BitEC_\Log$ is a variant of the $\gX$ error-correction procedure $\BitEC$ that works on logical qubits (we omit classical variables here for simplicity):
\[
  \ShorEC[\bar{\qit}, \bar{\qitB}, \bar{\qitC}] \ \ \defeq\ \
    \PhaseEC[\bar{\qit}];\
    \PhaseEC[\bar{\qitB}];\
    \PhaseEC[\bar{\qitC}];\
    \BitEC_\Log[\bar{\qit}, \bar{\qitB}, \bar{\qitC}].
\]
The Shor code error correction procedure can be proved by combining the specifications of $\PhaseEC$ and $\BitEC_\Log$ (\ref{g:modular}).
In each error correction procedure, we perform syndrome measurement and accumulate some global phase information.
However, framing such a syndrome by separating conjunction allows us to abstract the global phase information,
enabling scalable reasoning without any blow-up in the size of the formula (\ref{g:hiding}).

\subsection{Probabilistic Choice and Almost Sure Termination}
\label{sect:cases:coin-toss}

Finally, we present a case study of verifying a quantum program
that involves probabilistic choice and almost sure termination (\ref{g:prob_while}).

Since we can use the norm of a quantum state as a probability, we can define the probabilistic choice connective $\Prop \oplus_{\prob} \PropB$ by $\sqrt{\prob}\, \Prop \oplus \sqrt{1 - \prob}\, \PropB$, choosing the state $\Prop$ with probability $\prob$ and $\PropB$ with probability $1 - \prob$ for $\prob \in [0, 1]$.
This $\oplus_\prob$ satisfies the following as expected:
\begin{gather*}
  \Prop \oplus_\prob \PropB \, \dashvdash\,
    \PropB \oplus_{1 - \prob} \Prop
\hspace{3em}
  (\Prop \oplus_\prob \PropB) \oplus_\probB \PropC \, \dashvdash\,
    \Prop \oplus_{\prob \probB} (\PropB \oplus_{\frac{\probB}{1 - \prob \probB}} \PropC)
\\[.3em]
  \frac{
    \Prop \vdash \PropB
  \hspace{1.5em}
    \Prop' \vdash \PropB'
  }{
    \Prop \oplus_\prob \Prop' \, \vdash\, \PropB \oplus_\prob \PropB'
  }
\hspace{3em}
  \frac{
    \hoare{ \Prop }{ \Cmd }{ \PropB }
  \hspace{1.5em}
    \hoare{ \Prop' }{ \Cmd }{ \PropB' }
  }{
    \hoare{ \Prop \oplus_\prob \Prop' }
      { \Cmd }{ \PropB \oplus_\prob \PropB' }
  }
\end{gather*}

Also, the probabilistic choice program $\Cmd \oplus^\var_\prob \Cmd'$ can be defined as follows.
First, we introduce a primitive $\coin^\var_\prob$ that stores $0$ or $1$ to $\var$ with probability $\prob$ or $1 - \prob$, respectively.\footnote{
  Technically, we do not need to extend our program language, because we can mathematically encode $\coin^\var_\prob$ as $\gM_\prob^\var[]$, where $\gM_\prob$ is the 0-qubit measurement that applies $\sqrt{\prob}\, \gid$ returning $0$ or $\sqrt{1 - \prob}\, \gid$ returning $1$.
}
We have
\[
  \hoare{\pk2 \emp \pk2}[^\var]{\pk2 \coin^\var_\prob \pk2}
    {\pk2 \var \mapsto 0 \,\oplus_\prob\, \var \mapsto 1 \pk2}.
\]
From this, we can derive the probabilistic choice by $\Cmd \oplus^\var_\prob \Cmd' \,\defeq\, \coin^\var_\prob;\ \ifelse{ \var }{ \Cmd' }{ \Cmd }$,
for which the following rule can be derived:
\[
  \frac{
    \hoare{ \Prop }{ \Cmd }{ \PropB }
  \hspace{1.5em}
    \hoare{ \Prop }{ \Cmd' }{ \PropC }
  }{
    \hoare{\pk2 \Prop \pk2}[^\var]
      {\pk2 \Cmd \oplus^\var_\prob \Cmd' \pk2}
      {\pk2 \var \mapsto 0 \pk2*\pk2 \PropB \,\oplus_\prob\,
        \var \mapsto 1 \pk2*\pk2 \PropC \pk2}
  }
\]

\newcommand*{\cointoss}{\mathsf{cointoss}}
For example, let us think of the following repeat-until-success program:
\[
  \cointoss_\prob \ \ \defeq\ \
    \coin^\var_\prob;\,\
    \whilex{ \var }{ \paren[\big]{\pk2
      \coin^\var_\prob;\, \varC \store \varC + 1
    \pk2} }.
\]
Intuitively, the program repeats tossing a coin that returns $0$ with probability $\prob \ne 0$ until the coin finally returns $0$.
The counter variable $\varC$ stores the number of iterations.
Our logic can prove that this program almost surely terminates.
More specifically, we can prove the following, abstracting over an assertion $\Prop$ that is frameable and has the probability $1$:
\[
  \ex{\, \Prop \col \frameable,\pk1 \Prob 1}\
  \hoare{\pk2 \varC \mapsto 0 \pk2}[^\var]
    { \cointoss_\prob }
    {\pk2 \var \mapsto 0 \pk2*\pk2 \Prop \pk2}.
\]
 \section{Discussion}
\label{sect:discussion}

\paragraph{Automation}
One point worth noting about the case studies in the previous section
is that most of the proofs follow a uniform three-step pattern:
(i) prove the specification on the computational basis (or any chosen basis)
via simple symbolic execution;
(ii) apply \ref{lrule:hoare-scale} and \ref{lrule:hoare-frame};
and (iii) sum up the derivations with the \ref{lrule:hoare-sum} to lift the result
to general preconditions.

This observation also suggests a route to automation.
Step (i) is essentially a standard symbolic-execution pass
augmented with basic complex number arithmetic;
steps (ii)--(iii), despite applying rules with intricate side conditions,
can still be automated using the frameability notion of \cref{fig:frameable}, \cref{sect:logic:complexity}.
This kind of pattern appears in recent work on automated verification of quantum programs, such as AutoQ 2.0~\cite{ChenCHHLLT25-AutoQ2};
see \cref{sect:related} for the details.

\paragraph{Density matrices}

A limitation of such a simple methodology is that the size of the propositions
can grow exponentially with the number of measurements.
As we explained in \cref{sect:overview:abstraction}, the abstraction method can help mitigate the growth,
or we typically find compact representations
(\eg $\Mix^\iota_{i=0,1} (-1)^{\iota} \cdot \qit \mapsto \ket{a \xor \iota}$)
as shown in the case studies considered so far.
This method should works for any Clifford circuits, but is not guaranteed to work for general programs.
In such situations where an explosion occurs,
it may be preferable to work with \emph{density matrices} instead.

Although we adopt a global-phase-sensitive logic
rather than work with density matrices,
this choice does not preclude encoding density matrices in \ourlogic{}.
One such encoding is:
\[
  \qit \mapsto^{\iota} \dens \quad\defeq\quad
    \ex{I}\ \
    \ex{\pkcurly{ \qket[_i] }_{i \in I} }\ \
    \dens = \sum\nolimits_{i\in I} \ketbra{\psi_i}{\psi_i} \ *\
    \bigoplusn^{\iota}_{i\in I} \qit \mapsto \qket[_i].
\]
We can switch between the usual vector representation and this density-matrix representation in \ourlogic{}.
Where basis-locality is not necessary, one can use the density matrix representation $\qit \mapsto^{\iota} \dens$.
This roughly amounts to the existing approach of \citet{DengWX24-qsl} and enjoys entanglement- and outcome-locality.
Still, the existential quantification in $\qit \mapsto^{\iota} \dens$ makes the assertion non-precise, which is at odds with some advanced proof rules (\eg \ref{lrule:bigbmix-unframe}).
 \section{Related Work}
\label{sect:related}

\paragraph{Quantum Hoare logics}

The pioneering work on Hoare logic~\cite{Ying11-qhl, ZhouYY19-applied-qhl, Unruh19-qhl-ghost, LewisSZ21-qhl-survey} established a good foundation for
deductive verification of quantum programs.
In particular, \citet{LiY18-termination-qhl} has proposed a Hoare logic that can specify almost sure termination of quantum programs, as we also do.
However, as in the classical case, Hoare logic does not scale
for more complex programs, due to its very global style of specification.
This problem has been attacked by developing separation logics, as discussed below.

\paragraph{Quantum separation logics}

To bring local reasoning to quantum verification,
disentanglement has been proposed as a suitable notion of separation,
obtaining Quantum Separation Logic (QSL)~\cite{ZhouBHYY21-qsl, LeLSS22-qsl, DengWX24-qsl}.
As we remarked in \cref{sect:overview}, none of the existing QSLs have achieved all three locality principles we proposed (\cref{fig:locality-rules}, \cref{sect:intro}) in one logic.

Of particular note is \citet{DengWX24-qsl}'s logic, which is the only existing QSL that supports a probabilistic mixture of assertions $\oplus$.
Unlike our mixing $\oplus$, modelled as multiset sum, their model of $\oplus$ is the sum of density matrices, insensitive to global phases and the branching structure introduced by measurements.
Consequently, as discussed in \cref{sect:overview:measurement}, their approach to $\oplus$ is incompatible with our basis-locality principle, and a simple extension of their logic cannot achieve that locality.

Also, \citet{SuZFY24-qsl} have recently proposed a separation logic
for the modular verification of algorithms that utilize dirty qubits.
Their model is global-phase-insensitive, and thus does not support
our locality principles.
In particular, because of the lack of basis-locality, their framework does not handle cases where dirty qubits are entangled with the input, unlike \ourlogic{}.

\paragraph{Automated verification of quantum programs}

Automated verification of quantum programs is a highly important yet challenging research problem, and various approaches have been explored.
Qbricks~\cite{CharetonBBPV21-AutomatedDeductive} demonstrated some degree of automation by employing the path-sum representation together with Why3~\cite{FilliatreP13-Why3}, a semi-automated verification platform for functional programs.
Qafny~\cite{LiZCNLCW24-Qafny} automatically verifies annotated quantum programs by translating them into classical separation logic for arrays implemented in Dafny~\cite{LeinoM14-Dafny}.\footnote{
  The work by \citet{LiZCNLCW24-Qafny} was not cited in the published POPL version, as we became aware of it only after the version had been finalized.
}
\citet{WangY24-symbolic-execution} proposed a symbolic execution framework for quantum programs and applied it for automated verification of quantum error-correction codes.

Particularly relevant to our work is AutoQ 2.0~\cite{ChenCHHLLT25-AutoQ2}.
It is a framework for verifying Hoare-style assertions on quantum programs with measurement, branching, and loop constructs, with automation using level-synchronized tree automata~\cite{AbdullaCCHLLLT25-quantum-lsta}.
Much like our treatment, it represents post-measurement states as a map of non-normalized vectors and employs reasoning principles similar to \ourlogic{}'s outcome- and basis-locality.

\paragraph{Quantum relational logics}

Quantum relational logics~\cite{BartheHYYZ20-qrl, Unruh19-qrhl} verify relations between program behaviours rather than functional specifications of single programs.
For modularity, relational \emph{separation} logics have been studied in the classical probabilistic setting~\cite{BaoDF25-bluebell}, but not yet in the quantum setting.

\paragraph{Outcome logics}

Notably, the need to talk about a collection of outcomes at once
has emerged before in a completely different context:
Outcome Logic~\cite{ZilbersteinDS23-outcome-logic}
studied general reasoning principles about branching effects---\ie non-deterministic and probabilistic (but not quantum) computation---when assertions are over the whole collection of possible outcomes.
What we need here is a similar jump, conceptually.
Technically, however, there are serious challenges,
as Outcome Logic only considers one layer of locality,
while we have to handle three, deeply interacting, layers.
Outcome Separation Logic~\cite{ZilbersteinSS24-outcome-sl}
makes some first steps in incorporating separation of heaps
in a language with probabilistic branching,
but it uses a very syntactic approach to framing and does not hint at
design principles that might apply more generally.
In \ourlogic, we harmonize general framing with the other two connectives.
To make the logic work, in particular,
the connective that composes outcomes, mixing, must be non-commutative
and satisfy a strong interchange law with sum.

\paragraph{Quantum programming languages}

\emph{Quantum control} has been actively studied in the form of \textsf{qif} and symmetric pattern matching~\cite{SabryVV18-qif},
and adopted by several quantum languages~\cite{AltenkirchG05-QML, SvoreGTAGHKMPR18-Qsharp, BichselBGV20-Silq, HirataH25-Qurts, HeunenLMR2025-qif}.
Roughly speaking, a program $\qifelse{\expr}{\Cmd_1}{\Cmd_0}$ executes both branches $\Cmd_0$ and $\Cmd_1$
under negative or positive quantum control over the qubit $\expr$.
We expect that programs with such quantum control can be effectively verified using basis-locality of \ourlogic{}.

Also, in light of the no-cloning theorem, some recent quantum languages~\cite{KochBRLSSD25-Guppy, HirataH25-Qurts} adopt ownership types, popularized by Rust~\cite{MatsakisK14-Rust}.
We expect that \ourlogic{} can serve as a semantic foundation for such ownership types.
 \section{Conclusion and Future Work}
\label{sect:concl}

We proposed \ourlogic{}, the first quantum program logic to unify the three locality principles---entanglement-locality, outcome-locality, and basis-locality.
We proved its soundness and demonstrated its power by verifying several practical quantum programs.

Although we have provided a pen-and-paper proof of the soundness of our logic, it would be better to mechanize the proof, as has been done for some quantum program logics~\cite{ZhouBSLY23-CoqQ, DengWX25-qsl-arXiv}.
It would help explore advanced features without the fear of making mistakes.

One important direction for future work is the application of \ourlogic{} for automated verification.
We have started exploring this direction, as discussed in \cref{sect:discussion}.
More specifically, we expect that the insights from \ourlogic{} can be used to extend AutoQ 2.0~\cite{ChenCHHLLT25-AutoQ2}, discussed in \cref{sect:related}, with the entanglement-locality principle for more scalable automation.

We are also interested in extending our logic to relational verification.
A possible direction for future work is to adopt the perspective of Bluebell~\cite{BaoDF25-bluebell}, which unified
unary and relational reasoning in a single separation logic for probabilistic programs.
 
\begin{acks}
We would like to thank Yu-Fang Chen, Takeshi Tsukada and Ugo Dal Lago for their insightful discussions and valuable feedback.
This research was supported in part by the Hakubi Project at Kyoto University and JSPS KAKENHI Grant Number JP24KJ0133 for the first author, JST SPRING, Grant Number JPMJSP2110 and JST ACT-X, Grant Number JPMJAX23CT for the third author.
\end{acks}
 
\bibliographystyle{ACM-Reference-Format}

\appendix
\section{More on the Logic}
\label{app:sect:logic}

\subsection{Standard Connectives}
\label{app:sect:logic:standard}

We have the following standard connectives.
They satisfy the usual proof rules, presented below:
\begin{gather*}
  \textstyle
  \all{x \in A} \Prop_x \pk6\defeq\pk6
  \bigcap_{x \in A} \Prop_x
\hspace{2.5em}
  \ex{x \in A} \Prop_x \pk6\defeq\pk6
  \bigcup_{x \in A} \Prop_x
\\[.2em]
  \Prop_0 \land \Prop_1 \pk6\defeq\pk6
  \all{\pk1 i \in \nk1\curly{0, 1}} \pk1 \Prop_i
\hspace{2.5em}
  \Prop_0 \lor \Prop_1 \pk6\defeq\pk6
  \ex{\pk1 i \in \nk1\curly{0, 1}} \pk1 \Prop_i
\\[.2em]
  \neg \Prop \pk6\defeq\pk6 \Prop^{\pk1 \complement}
\hspace{2.5em}
  \lift{\varphi} \pk6\defeq\pk6 \set{ \mult }{ \varphi }
\\[.2em]
  \ffrac{
    \all{x \in A} \pk2 (\Prop \vdash \PropB_x)
  }{
    \Prop \vdash \all{x \in A} \PropB_x
  }
\hspace{3em}
  \ffrac{
    \all{x \in A} \pk2 (\Prop_x \vdash \PropB)
  }{
    (\ex{x \in A} \Prop_x) \vdash \PropB
  }
\\[.2em]
  \Prop \land \neg \Prop \, \vdash\, \PropB
\hspace{3em}
  \PropB \, \vdash\, \Prop \lor \neg \Prop
\hspace{3em}
  \ffrac{
    \varphi \, \limp\, \paren{\Prop \vdash \PropB}
  }{
    \lift{\varphi} \land \Prop \, \vdash\, \PropB
  }
\end{gather*}

\subsection{Right Adjoint of SL Connectives}
\label{app:sect:logic:right-adjoint}

In \cref{rem:right-adjoint}, we mentioned that the connectives $\land$, $*$, $\oplus$ and $+$ all have the right adjoint, written as $\limp$, $\wand$, $\pine$ and $\cross$.

The following give the models for these connectives:\footnote{
  We use $\downarrow$ to denote definedness (\eg $\mult \cdot \mult' \defined$ means that $\mult \cdot \mult'$ is defined).
}
\begin{gather*}
  \Prop \limp \PropB \pk6\defeq\pk6
  \Prop^{\pk1 \complement} \cup \PropB
\hspace{2.5em}
  \Prop \wand \PropB \pk6\defeq\pk6
  \set{ \mult }{
    \all{\mult' \in \Prop \st \mult \cdot \mult' \defined} \pk3
    \mult \cdot \mult' \in \PropB }
\\[.3em]
  \Prop \pine \PropB \pk6\defeq\pk6
  \set{ \mult }{ \all{\mult' \in \Prop} \pk2 \mult \uplus \mult' \in \PropB }
\\[.3em]
  \Prop \cross \PropB \pk6\defeq\pk6
  \set{ \mult }{ \all{\mult' \in \Prop} \pk2
    \all{\mbij \col \mult \mbiject \mult' \st \mult +_\mbij \mult' \defined} \pk2
    \mult +_\mbij \mult' \in \PropB }
\end{gather*}

As expected, we have the following adjunction rules.
\begin{gather*}
  \ffrac{
    \Prop \land \PropB \pk2\vdash\pk2 \PropC
  }{
    \Prop \pk2\vdash\pk2 \PropB \limp \PropC
  }
  \hspace{3em}
  \ffrac{
    \Prop * \PropB \pk2\vdash\pk2 \PropC
  }{
    \Prop \pk2\vdash\pk2 \PropB \wand \PropC
  }
  \hspace{3em}
  \ffrac{
    \Prop \oplus \PropB \pk2\vdash\pk2 \PropC
  }{
    \Prop \pk2\vdash\pk2 \PropB \pine \PropC
  }
  \hspace{3em}
  \ffrac{
    \Prop + \PropB \pk2\vdash\pk2 \PropC
  }{
    \Prop \pk2\vdash\pk2 \PropB \cross \PropC
  }
\end{gather*}

\subsection{Additional Proof Rules}
\label{app:sect:logic:rules}

We have the following proof rules for the quantum and classical points-to tokens.
\begin{proofrules*}
  \infer*[lab = qpoints-swap]{}{
    (\bar{\qitC}, \qit, \qitB, \bar{\qitC'}) \mapsto \qket \pk6\dashvdash\pk6
    (\bar{\qitC}, \qitB, \qit, \bar{\qitC'}) \mapsto \gSWAP_{\qit, \qitB} \qket
  }
  \label{arule:qpoints-swap}
\\
  \infer*[lab = qpoints-distinct]{}{
    \bar{\qit} \mapsto \qket \pk3\vdash\pk3
    \distinct(\bar{\qit})
  }
  \label{arule:qpoints-distinct}

  \infer*[lab = cpoints-disj]{}{
    \var \mapsto \val \pk3*\pk3 \varB \mapsto \valB
    \pk2\vdash\pk2 \var \ne \varB
  }
  \label{arule:cpoints-disj}
\end{proofrules*}
The rule \ref{arule:qpoints-swap} is for permuting the qubits of a quantum points-to token.
Here, $\gSWAP$ is the unitary operator that maps $\ket{i j}$ to $\ket{ji}$ for each $i, j \in \curly{0, 1}$.

We have the following proof rules for $\Prop \col \nonnb$, defined in \cref{fig:incomp-unambig}.
\begin{gather*}
  \emp, \, \bar{\qit} \mapsto \qket, \, \var \mapsto \val \col \nonnb
\hspace{3em}
  \frac{
    \PropB \col \nonnb
    \hspace{1.5em}
    \Prop \vdash \PropB
  }{
    \Prop \col \nonnb
  }
\\[.5em]
  \frac{
    \all{x \in I} \pk2 \paren{\Prop_x \col \nonnb}
    \quad
    I \neq \empset
  }{
    \bigoplus_{x \in I} \Prop_x \col \nonnb
  }
\hspace{3em}
  \frac{
    \Prop, \PropB \col \nonnb
  }{
    \Prop * \PropB \col \nonnb
  }
\hspace{3em}
  \frac{
    \Prop \col \nonnb
  }{
    \Prop + \PropB \col \nonnb
  }
\end{gather*}

We have the following proof rules for $\Prop \col \Prob \prob$, defined in \cref{fig:prob}.
\begin{gather*}
  \emp, \var \mapsto \val \col \Prob 1
  \hspace{3em}
  \bar{\qit} \mapsto \qket \col \Prob\pk2 \norm{\qket}^2
\hspace{3em}
  \frac{
    \PropB \col \Prob \prob
    \hspace{1.5em}
    \Prop \vdash \PropB
  }{
    \Prop \col \Prob \prob
  }
\\[.5em]
  \frac{
    \all{x} \pk2 \paren{\Prop_x \col \Prob \prob_x}
  }{
    \bigoplus_{x \in I} \Prop_x \col \Prob\pk2 \paren{\sum_{x \in I} \prob_x}
  }
\hspace{3em}
  \frac{
    \Prop \col \Prob \prob
    \hspace{1.5em}
    \PropB \col \Prob \probB
  }{
    \Prop * \PropB \col \Prob \prob \probB
  }
\end{gather*}

\subsection{More on Tagged Mixing}
\label{app:sect:logic:mix}

We have the following derived rules for tagged mixing, other than those in \cref{fig:mix}:
\begin{proofrules*}
  \infer*[lab = bigmix-mono]{
    \all{x \in I}\, \paren{\Prop_x \vdash \PropB_x}
  }{ \textstyle
    \bigoplus^\hvar_{x \in I} \Prop_x \, \vdash\,
    \bigoplus^\hvar_{x \in I} \PropB_x
  }
  \label{arule:bigmix-mono}

  \infer*[lab = bigmix-comm]{
    \text{$f \col I \to J$ is a bijection}
  }{ \textstyle
    \bigoplus^\hvar_{x \in I} \Prop_{\pk1 f\pk2 x} \, \dashvdash\,
    \bigoplus^\hvar_{y \in J} \Prop_y
  }
  \label{arule:bigmix-comm}

  \infer*[lab = bigmix-bigmix]{}{ \textstyle
    \bigoplus^\hvar_{x \in I} \bigoplus^\hvarB_{y \in J} \Prop_{x, y} \, \dashvdash\,
    \bigoplus^\hvarB_{y \in J} \bigoplus^\hvar_{x \in I} \Prop_{x, y}
  }
  \label{arule:bigmix-bigmix}

  \infer*[lab = nb-bigmix]{}{ \textstyle
    \nb \, \dashvdash\, \bigoplus^\hvar_{\_ \in \empset}
  }
  \label{arule:nb-bigmix}

  \infer*[lab = mix-bigmix]{}{ \textstyle
    \Prop_0 \mix{\hvar} \Prop_1 \, \dashvdash\,
    \bigoplus^\hvar_{i \in \{0, 1\}} \Prop_i
  }
  \label{arule:mix-bigmix}

  \infer*[lab = bigmix-scale]{}{ \textstyle
    \scl{\alpha}{(\bigoplus^\hvar_{x \in I} \Prop_x)} \, \dashvdash\,
    \bigoplus^\hvar_{x \in I} \pk2 \scl{\alpha}{\Prop_x}
  }
  \label{arule:bigmix-scale}

  \infer*{
    \all{x} \pk2 \paren{\Prop_x \col \precise}
  }{ \textstyle
    \bigoplus^\hvar_{x \in I} \Prop_x \col \precise
  }

  \infer*[lab = hoare-bigmix]{
    \all{x \in I} \pk2
    \hoare{ \Prop_x }{ \Cmd }{ \PropB_x }
  }{ \textstyle
    \hoare{\pk2 \bigoplus^\hvar_{x \in I} \Prop_x \pk2}{ \Cmd }{\pk2 \bigoplus^\hvar_{x \in I} \PropB_x \pk2}
  }
  \label{arule:hoare-bigmix}

  \infer*[lab = bigmix-frame]{}{ \textstyle
    (\bigoplus^\hvar_{x \in I} \Prop_x) \pk2*\pk2 \PropB \, \vdash\,
      \bigoplus^\hvar_{x \in I} \pk2 (\Prop_x * \PropB)
  }
  \label{arule:bigmix-frame}

\infer*[lab = bigmix-unframe]{
    \PropB \col \precise
  }{ \textstyle
    \bigoplus^\hvar_{x \in I} \pk2 (\Prop_x * \PropB) \, \vdash\,
      (\bigoplus^\hvar_{x \in I} \Prop_x) \pk2*\pk2 \PropB
  }
  \label{arule:bigmix-unframe}

  \infer*[lab=bigmix-frame-frameable]{
    \PropB \col \frameable
  }{ \textstyle
    \bigoplus^\hvar_{x \in I} \pk2 (\Prop_x * \PropB) \, \dashvdash\,
      (\bigoplus^\hvar_{x \in I} \Prop_x) \pk2*\pk2 \PropB
  }
  \label{arule:bigmix-frame-frameable}

  \infer*{
    \all{x \in I} \pk2 \paren{\Prop_x \col \nonnb}
    \quad
    I \neq \empset
  }{ \textstyle
    \bigoplus^\hvar_{x \in I} \Prop_x \col \nonnb
  }

  \infer*{
    \all{x} \pk2 \paren{\Prop_x \col \Prob \prob_x}
  }{ \textstyle
    \bigoplus^\hvar_{x \in I} \Prop_x \col \Prob\pk2 \paren{\sum_{x \in I} \prob_x}
  }
\end{proofrules*}

\subsection{More on Incompatibility Relation}
\label{app:sect:logic:incomp}

In \cref{fig:incomp-unambig}, we defined the incompatibility relation of resources by
\[
  (\qket, \Store) \pk1\hash\pk1 (\qketB, \Store') \pk9\defeq\pk9
  \ex{\var \pk1\in\pk1 \dom \Store \pk1\cap\pk1 \dom \Store'} \pk2
  \Store[\var] \pk1\ne\pk1 \Store'[\var].
\]
We state the following property that characterizes it algebraically.

\begin{lemma}[Algebraic characterization of incompatibility]
  \label{lem:incom-basic-prop}
  For any $x, y$ in the resource ring $R = \Res$,
  the following three conditions are equivalent:
  \begin{enumerate}[label=(\arabic*)]
    \item $x \hash y$,
    \label{cond:incom-basic-prop:0}
    \item $\all{k, h \in R}\  k \cdot x + h \cdot y = \bot$,
    \label{cond:incom-basic-prop:1}
    \item $\all{z \in R}\ (\all{k \in R}\ k \cdot x + z = \bot) \lor (\all{h \in R}\ h \cdot y + z = \bot)$.
    \label{cond:incom-basic-prop:2}
  \end{enumerate}
\end{lemma}
\begin{proof}
  For simplicity, we denote $x[a]$ to mean $\Store[a]$ for $x = (\qket, \Store)$.

  \labelcref{cond:incom-basic-prop:0} $\Rightarrow$ \labelcref{cond:incom-basic-prop:1}.
  If there exists $a$ such that $x[a] \neq y[a]$ and both $k\cdot x$ and $h\cdot y$ are defined,
  then $(k\cdot x)[a] = x[a] \neq y[a] = (h\cdot y)[a]$. So $k \cdot x + h \cdot y$ is undefined, that is, it equals $\bot$.

  \labelcref{cond:incom-basic-prop:1} $\Rightarrow$ \labelcref{cond:incom-basic-prop:2}.
  In $\Res$, if $x + y \neq \bot$, then $x + z = \bot$ if and only if $y + z = \bot$ for any $z$.
  Therefore, if $k\cdot x + z \neq \bot$ and $k\cdot x + h \cdot y = \bot$,
  then $z + h \cdot y = \bot$.

  \labelcref{cond:incom-basic-prop:2} $\Rightarrow$ \labelcref{cond:incom-basic-prop:0}.
  We prove the contrapositive.
  Let $x = (\Qit\mapsto\qket, \Store)$ and $y = (\QitB\mapsto\qketB, \Store')$.
  Since there is no $a \in \dStore$ such that $\Store[a] \neq \Store'[a]$,
  we can take the supremum $\Store'' \defeq \sup(\Store, \Store')$.
  Let us define $\QitC \defeq \Qit\cup\QitB$
  and $z \defeq (\QitC \mapsto 0, \Store'')$.
  Such $z$ satisfies
  $
  z
  = z + z
  = x \cdot (\QitC\setminus\Qit \mapsto 0, \Store''\setminus\Store)
  = y \cdot (\QitC\setminus\QitB \mapsto 0, \Store''\setminus\Store')
  $
  thus it contradicts \labelcref{cond:incom-basic-prop:2}.
\end{proof}

\begin{lemma}[Frame on incompatibility]
  \label{lem:incomp-closed-under-mult-add}
  If $x, y \in R$ satisfy $x \hash y$, then $x + y = \bot$.
  In particular, $x\neq y$.
  Also, for any $z \in R$, $x\cdot z \hash y$ and $x + z \hash y$ if the left-hand side is defined.
\end{lemma}
\begin{proof}
  Immediate from \labelcref{cond:incom-basic-prop:1}.
\end{proof}

\begin{remark}
  In this paper, we defined the incompatibility relation only for the specific resource ring $\Res$.
  This definition extends to any resource ring $R$
  by adopting \labelcref{cond:incom-basic-prop:1} $\land$ \labelcref{cond:incom-basic-prop:2} as axioms.
  These conditions \labelcref{cond:incom-basic-prop:1} and \labelcref{cond:incom-basic-prop:2}
  are not equivalent in general, so both are required.
  \Cref{lem:incomp-closed-under-mult-add} holds for arbitrary $R$,
  and all soundness proof in this subsection
  can be done with this general definition.
  We note that the condition \labelcref{cond:incom-basic-prop:1}
  captures the essence of incompatibility,
  whereas \labelcref{cond:incom-basic-prop:2} is required for the soundness of \ref{lrule:sum-precise} and \ref{lrule:sum-unframe}.
\end{remark}

\subsection{Soundness Proof}
\label{app:sect:logic:proof}

We give proofs of the correctness of some non-trivial proof rules.

\begin{proof}[Proof of \ref{lrule:bigbmix-frame}, \cref{fig:pcm-sep}]
  All elements on the left-hand side are of the form $(\biguplus_x \mult_x) \cdot \mult'$ where
  $\mult_x\in {\Prop_x}$ and $\mult'\in {\PropB}$.
  By an easy calculation, this is equivalent to
  $\biguplus_x \mult_1 \cdot \mult'$,
  which is on the right-hand side.\footnote{
    The opposite direction is not generally true when ${\PropB}$ is not a singleton set.
  }
\end{proof}

\begin{proof}[Proof of \ref{lrule:bigbmix-unframe}, \cref{fig:pcm-sep}]
  If ${\PropB}$ is empty, then both sides are empty.
  If ${\PropB}$ is a singleton $\{ \mult \}$, then the left-hand side is $\set{ \biguplus_x \mult_x \cdot \mult }{\mult_x \in {\Prop_x} }$.
  Since $\biguplus_x \mult_x \cdot \mult = (\biguplus_x \mult_x)\cdot \mult$, this is included in the right-hand side.
\end{proof}

\begin{proof}[Proof of \ref{lrule:sum-assoc}, \cref{fig:sum}]
  All elements on the left-hand side are of the form $(\mult_1 +_\mbij \mult_2) +_\mbijB \mult_3$ where
  $\mult_1\in {\Prop}$,
  $\mult_2\in {\PropB}$,
  $\mult_3\in {\PropC}$,
  with multiset bijections
  $\mbij \col \mult_1 \mbiject \mult_2$
  and
  $\mbijB \col (\mult_1 + \mult_2) \mbiject \mult_3$.
  Since multiset bijections compose and there are canonical multiset bijections
  $\mult \mbiject \mult' \mbiject \mult +_{h} \mult'$
  when $\mult +_h \mult'$ is defined,
  we have a canonical multiset bijection $\mult_1 \mbiject \mult_2 \mbiject \mult_3$.
  Along these bijections, we can define $\mult_1 +_{\mbij'}(\mult_2 +_{\mbijB'} \mult_3)$
  with appropriate $\mbij'$ and $\mbijB'$, satisfying
  $\mult_1 +_{\mbij'}(\mult_2 +_{\mbijB'} \mult_3)
  = (\mult_1 +_{\mbij} m_2) +_{\mbijB} \mult_3$.
\end{proof}

\begin{proof}[Proof of \ref{lrule:sum-frame}, \cref{fig:sum}]
  All elements on the left-hand side are of the form $(\mult_1 +_\mbij \mult_2) \cdot \mult'$ where
  $\mult_1\in {\Prop}$, $\mult_2\in {\PropB}$, and $\mult'\in {\PropC}$,
  with a multiset bijection $\mbij \col \mult_1 \mbiject \mult_2$.
  Then,
  \begin{align*}
    (\mult_1 +_\mbij \mult_2) \cdot \mult'
     & = \mset{ (x + x') \cdot y
    \mid (x, x') \mIn \mbij, y \mIn \mult' }                      \\
     & = \mset{ (x \cdot y) + (x' \cdot y)
    \mid (x,x') \mIn \mbij, y \mIn \mult' }                       \\
     & = (\mult_1 \cdot \mult') +_{\mbij'} (\mult_2 \cdot \mult')
    \quad\in \text{RHS},
  \end{align*}
  where $\mbij'$ is the extended multiset bijection of $\mbij$
  that relates $x \cdot y$ and $x' \cdot y$ for each $(x,x') \mIn \mbij$.\footnote{
    The opposite direction is not generally true because not all the multiset bijections between $\mult_1 \cdot \mult'$ and $\mult_2 \cdot \mult'$ are of the form $\mbij'$.
  }
\end{proof}

\begin{proof}[Proof of \ref{lrule:sum-bigbmix}, \cref{fig:sum}]
  All elements on the left-hand side are of the form
  $\biguplus_x \mult_x +_{\mbij_x} \mult'_x$ where
  $\mult_x\in {\Prop_x}$ and $\mult'_x\in {\PropB_x}$,
  with multiset bijections $\mbij_x \col \mult_x \mbiject \mult'_x$.
  Let $\biguplus_x \mbij_x \col (\biguplus_x \mult_x) \mbiject (\biguplus_x \mult'_x)$ be the multiset bijection
  that is the union of all the $\mbij_x$, \ie
  if $\mult_x, \mult'_x, \mbij_x \colon J_x \to \Res$,
  \[
    \Big(\biguplus\nolimits_x \mbij_x\Big)
    \colon\
    \bigsqcup\nolimits_{x \in I} J_x \longrightarrow \Res
    ;\quad
    (x \in I, y \in J_x) \longmapsto (a, b) = r_x(y).
  \]
  Then, we have $\biguplus_x \mult_x +_{\mbij_x} \mult'_x =
    \mset{ a + b \mid x\in I, (a,b) \mIn \mbij_x }
    =
    \mset{ a + b \mid (a,b) \mIn \biguplus_x \mbij_x }
    \in \text{RHS}$.\footnote{
    The opposite direction is not generally true because not all the multiset bijections
    between $\biguplus_x \mult_x$ and $\biguplus_x \mult'_x$ are of the form $\biguplus_x \mbij_x$.
  }
\end{proof}

\begin{proof}[Proof of \ref{lrule:sum-precise}, \cref{fig:incomp-unambig}]
  Let $\Prop = \{ \mult \}$ and $\PropB = \{ \mult' \}$.
  Assume that we have $\mbij, \mbijB \col \mult \mbiject \mult'$
  such that $\mult +_\mbij \mult'$ and $\mult +_\mbijB \mult'$.
  If $(a, b) \mIn \mbij$ and $(a', b) \mIn \mbijB$,
  then $a + b \defined$ and $a' + b \defined$ hold.
  From \labelcref{cond:incom-basic-prop:2} in \cref{lem:incom-basic-prop},
  $\neg (a \hash a')$.
  Since $\mult$ is unambiguous, $a = a'$.
  Therefore, $\mbij = \mbijB$.
\end{proof}

\begin{proof}[Proof of \ref{lrule:bigbmix-sum}, \cref{fig:incomp-unambig}]
  All elements on the left-hand side are of the form
  $(\biguplus_{x\in I} \mult_x) +_{\mbij} (\biguplus_{x \in I} \mult'_x)$
  where
  $\mult_x \in {\Prop_x}$,
  $\mult'_x \in {\PropB_x}$, and
  $\mbij \col (\biguplus_x \mult_x) \mbiject (\biguplus_x \mult'_x)$.
  Let $\mult \defeq \biguplus_x \mult_x$
  and $\mult' \defeq \biguplus_x \mult'_x$.
  This $\mbij$ can be regarded as a bijection of the set of indices $\bigsqcup_{x\in I}J_x$ and $\bigsqcup_{x\in I}J'_x$
  where $\mult_x \colon J_x \to \Res$ and $\mult'_x \colon J'_x \to \Res$.
  For each $(x,j) \in \bigsqcup_{x\in I}J_x$, let $(y, j') \defeq \mbij(x,j)$.
  Since the sum $\mult(x,j) + \mult'(\mbij(x,j)) = \mult_x(j) + \mult'_y(j')$ is defined,
  $x = y$ because of the assumption $x \neq y \Rightarrow \Prop_x \hash \PropB_y$
  and \cref{lem:incomp-closed-under-mult-add}.
  Therefore, for each $i \in I$, the image of $\mbij$ restricted to $J_i$ is included in $J'_i$.
  The bijection $\mbij$ must be of the form
  $\biguplus_{x\in I} \mbij_x \col \biguplus_{x\in I}\mult_x \to \biguplus_{x\in I}\mult'_x$
  with some multiset bijections $\mbij_x \col \mult_x \mbiject \mult'_x$.
  Therefore, $(\biguplus_x \mult_x) +_\mbij (\biguplus_x \mult'_x)$ can now be rewritten as
  $\biguplus_x \mult_x +_{\mbij_x} \mult'_x$, which is in the right-hand side.
\end{proof}

\begin{proof}[Proof of \ref{lrule:sum-unframe}, \cref{fig:incomp-unambig}]
  If ${\PropC}$ is empty, then both sides are empty.
  Let ${\PropC}$ be $\{ \mult \}$.
  All elements on the left-hand side are of the form
  $(\mult_1 \cdot \mult) +_\mbij (\mult_2 \cdot \mult)$
  where $\mult_1 \in {\Prop}$, $\mult_2 \in {\PropB}$,
  and $\mbij \col \mult_1\cdot \mult \mbiject \mult_2\cdot \mult$.
  Because of the unambiguity of $\Prop$ and $\PropC$,
  all the multiplicities of $\mult_1$ and $\mult$ are either $0$ or $1$,
  so the multisets $\mult_1$ and $\mult$ can be regarded as mere sets of resources.
  The multiset bijection $\mbij$ can now be regarded as a bijection of the sets $\mult_1 \times \mult$ and $I \times \mult$,
  where $I$ is the set of indices of $\mult_2\colon I \to \Res$.
  For each $x \in \mult_1$ and $z \in \mult$,
  if $(i, z') \defeq \mbij(x, z)$,
  then $x\cdot z + \mult_2(i) \cdot z' \neq \bot$ follows because
  $(\mult_1 \cdot \mult) +_\mbij (\mult_2 \cdot \mult) \defined$.
  Because of the unambiguity of $\PropC$,
  $z = z'$ follows from \labelcref{cond:incom-basic-prop:1} in \cref{lem:incom-basic-prop}.
  Therefore, for each $z \in \mult$,
  there is a bijection $\mbij_z \defeq \mathrm{fst} \circ \mbij(-, z) \col \mult_1 \to I$
  that satisfies $\mbij(x, z) = (\mbij_z(x), z)$.
  We prove that this bijection $\mbij_z$ does not depend on $z$.
  Since $\mult$ is non-empty because of the assumption $\PropC \col \nonnb$, we choose one $z_0 \in \mult$.
  For each $z \in \mult$ and $i = \mbij_{z_0}(x) = \mbij_z(x')$,
  both $x\cdot z + \mult_2(i) \cdot z$ and $x'\cdot z + \mult_2(i) \cdot z$ are defined.
  Therefore, from
  \labelcref{cond:incom-basic-prop:2} in \cref{lem:incom-basic-prop},
  $\neg(x\hash x')$.
  Because $\Prop$ is unambiguous, $x$ and $x'$ must be equal, which proves
  $\mbij_z^{-1} \circ \mbij_{z_0} = \id_{\mult_1} (\Leftrightarrow \mbij_{z_0} = \mbij_{z})$.
  We can now rewrite $(\mult_1 \cdot \mult) +_\mbij (\mult_2 \cdot \mult)$ as
  \begin{align*}
    (\mult_1 \cdot \mult) +_\mbij (\mult_2 \cdot \mult)
     & = \mset{ x\cdot z + y \cdot z' \ \mid\  ((x,z), (y,z')) \mIn \mbij }                                 \\
     & = \mset{ x\cdot z + \mult_2(i) \cdot z' \ \mid\  x \in \mult_1, z \in \mult, (i, z') = \mbij(x, z) } \\
     & = \mset{ (x + \mult_2(\mbij_{z_0}(x)))\cdot z \ \mid\ x \in \mult_1, z \in \mult }                            \\
     & = (\mult_1 +_{\mbij_{z_0}} \mult_2) \cdot \mult \qquad
    \quad\in \text{RHS}.
    \qedhere
  \end{align*}
\end{proof}

\begin{proof}[Proof of the rest of the rules in \cref{fig:incomp-unambig}]
  Most rules are straightforward from the definition
  and \cref{lem:incomp-closed-under-mult-add}.
  We prove $\Prop, \PropB \col \unambig \Rightarrow \Prop * \PropB \col \unambig$
  as an example.
  Let $\mult\cdot \mult' \in \Prop * \PropB$ where $\mult\in\Prop$ and $\mult'\in\PropB$,
  and $\mset{x\cdot x',\, y\cdot y'} \msubseteq \mult \cdot \mult'$
  where $x,y \mIn \mult$ and $x',y' \mIn \mult'$.
  The indices of $x \cdot x'$ and $y\cdot y'$ are not the same,
  thus $\mset{x,\,y} \msubseteq \mult$ or $\mset{x',\,y'} \msubseteq \mult'$
  from the definition of $\mult \cdot \mult'$.
  Without loss of generality, we assume $\mset{x,\,y} \msubseteq \mult$.
  Then, from \cref{lem:incomp-closed-under-mult-add},
  $x \cdot x' \hash y$ and $x \cdot x' \hash y \cdot y'$.
\end{proof}

The remaining rules to be shown are the ones regarding Hoare triples.
From here on, we denote $c \Rightarrow t$ to mean
$c = (\bar{\Cmd}, \qket, S)$ and $t = \sem{\bar{\Cmd}}(\qket, S)$.

\begin{lemma}[Frame on denotations]\label{lem:denosem-frame}
  Let $c = (\bar{\Cmd}, a)$ be a configuration and $t$ be a tree such that $c\Rightarrow t$.
  If $x \in \Res$ satisfies $(\bar{\Cmd}, a\cdot_\Res x)\Rightarrow t'$,
  then the tree $t'$ is obtained from $t$
  by replacing all the leaves $\Leaf(b)$ with $\Leaf(b\cdot_\Res x)$.
\end{lemma}
\begin{proof}
  By induction on the definition of the denotational semantics.
  Note that no commands change the domain of resources, \ie
  if $c'$ appears as a leaf of $\sem{\bar{\Cmd}}(c)$
  where
  $c = (\Qit \mapsto \qket, \Store)$
  and
  $c' = (\QitB \mapsto \qket, \Store')$,
  then $\Qit = \QitB$ and $\dom(\Store) = \dom(\Store')$.
\end{proof}

\begin{proof}[Proof of \ref{lrule:hoare-frame}, \cref{fig:hoare}]
  Follows from \cref{lem:denosem-frame}.
\end{proof}

\begin{proof}[Proof of \ref{lrule:hoare-bigbmix}, \cref{fig:hoare}]
  Straightforward from the definition of weakest precondition.
\end{proof}

\begin{lemma}[Adding denotations]\label{lem:denosem-addition}
  Let $c_0$ and $c_1$ be two configurations such that $c_i = (\bar{\Cmd}, \Qit \mapsto \qket[_i], \Store)$.
  If $c_i \Rightarrow t_i$ holds for $i \in \{0, 1\}$, then $c_0 + c_1 \Rightarrow t_0 + t_1$ holds, where
  $c_0 + c_1 \defeq (\bar{\Cmd}, (\Qit \mapsto \qket[_0] + \qket[_1]), \Store)$
  and $t_0 + t_1$ is the tree of the same structure as $t_0$ and $t_1$ (the two agree in the structure) whose leaves are $\Leaf(a_0 +_\Res a_1)$ letting $\Leaf(a_i)$ be the leaf of $t_i$ at the same position.
\end{lemma}
\begin{proof}
  By induction on the definition of the denotational semantics.
  This can also be proven using \cref{thm:opsem-denosem} and the linearity of the operational semantics.
\end{proof}

\begin{proof}[Proof of \ref{lrule:hoare-sum}, \cref{fig:sum}]
  Follows from \cref{lem:denosem-addition}
  and the fact that there is a canonical multiset bijection between
  $\Leaves(t)$ and $\Leaves(t')$ whenever $t$ and $t'$ have the same shape.
\end{proof}

\begin{proof}[Proof of \ref{lrule:hoare-seq}, \cref{fig:hoare}]
  Let $\mset{a_i \in \Res \mid i \in I} \in {\Prop}$ and $(\Cmd, a_i) \Rightarrow t_i$.
  Then, we have $\biguplus_{i \in I} \Leaves(t_i) \in {\PropB}$.
  We also know that for each $b_{i j} \mIn \Leaves(t_i)$,
  $\ex{t_{i j}} (\Cmd', b_{i j}) \Rightarrow t_{i j}$,
  satisfying $\biguplus_{i} \biguplus_{j} \Leaves(t_{i j}) \in {\PropC}$.
  Now, from the definition of the denotational semantics,
  we have $(\Cmd;\Cmd', a_i) \Rightarrow t_i[t_{i j}/\Leaf(b_{i j})]$.
  The multiset of leaves of the tree $t_i[t_{i j}/\Leaf(b_{i j})]$ is $\biguplus_{j} \Leaves(t_{i j})$.
  Therefore, $\biguplus_i \Leaves(t_i[t_{i j}/\Leaf(b_{i j})]) = \biguplus_i\biguplus_j \Leaves(t_{i j}) \in {\PropC}$.
\end{proof}

\begin{proof}[Proof of \ref{lrule:hoare-while}, \cref{fig:hoare}]
  From the definition of the denotation of the while loop,
  the following holds.
  \[
    \sem{\whilex{\expr}{\Cmd}}
    \ =\ \sem{\ifonly{\expr}{\big(\Cmd;\ \whilex{\expr}{\Cmd}\big)}}
  \]
  Since the denotation of the while loop is defined by the least fixed point,
  this can be rewritten as
  \begin{align*}
    \sem{\whilex{\expr}{\Cmd}}
    &\ =\ \sup\nolimits_{n \in \NN}
    \sem{
      \ifonly{\expr}{\big( \Cmd;\
      \dots
      \ifonly{\expr}{\big( \Cmd;\
        \ifonly{\expr}{\eloop}
      \big)}
      \dots
      \big)}
    }
  \\
    &\ =\ \sup\nolimits_{n \in \NN}
    \sem{\big(\ifonly{\expr}{\Cmd} \big)^n,\ \ifonly{\expr}{\eloop}}
  \end{align*}
  where $\eloop$ is the infinite loop whose denotation is $\lambda \_. \Nil$,
  and $\iter_n \defeq \big(\ifonly{\expr}{\Cmd} \big)^n$ denotes that
  the command $\big(\ifonly{\expr}{\Cmd}\big)$ is executed $n$ times in sequence.
  Here, we used the following equation and the compositionality of
  the denotational semantics.
  \begin{align*}
    & \sem{\ifonly{\expr}{\big(\Cmd;\ \ifonly{\expr}{\Cmd'}\big)}} (\qket, \Store)
  \\ &\ =\
    \begin{cases}
      \sem{\Cmd;\ \ifonly{\expr}{\Cmd'}}(\qket, \Store) & \text{if } \sem{e}_S = 1 \\
      \sem{\cskip}(\qket, \Store) & \text{otherwise}
    \end{cases}
  \\ & \ =\
    \begin{cases}
      \sem{\Cmd,\ \ifonly{\expr}{\Cmd'}}(\qket, \Store) & \text{if } \sem{e}_S = 1 \\
      \sem{\cskip,\ \ifonly{\expr}{\Cmd'}}(\qket, \Store) & \text{otherwise}
    \end{cases}
  \\ & \ =\
    \sem{\ifonly{\expr}{\Cmd},\ \ifonly{\expr}{\Cmd'}} (\qket, \Store).
  \end{align*}

  Let $\mset{a_i \in \Res \mid i \in I} \in {\Prop_0}$,
  and $(\whilex{\expr}{\Cmd}, a_i) \Rightarrow t_{i, \infty}$.
  Then, there exists a sequence of trees $t_{i, n}$ such that
  $(\iter_n, a_i) \Rightarrow t_{i, n}$.
  We define multisets from such trees.
  \begin{align*}
     & \mult^P_{i, 0} \defeq \mset{ a_i },
     \qquad
     \mult^P_{i, n + 1} \defeq
    \mset{ b \mid a \mIn \mult^R_{i, n}, b \mIn \Leaves(\sem{C}(a)) }
    \\
     & \mult^Q_{i, n} \defeq
    \mset{ (\qket, \Store) \mid (\qket, \Store) \mIn \mult^P_{i, n}, \sem{\expr}_\Store = 0 }
    \\
     & \mult^R_{i, n} \defeq
    \mset{ (\qket, \Store) \mid (\qket, \Store) \mIn \Leaves(t_{i, n}), \sem{\expr}_\Store = 1 }
  \end{align*}
  Note that this is not a recursive definition:
  the definition of $\mult^P$ and $\mult^Q$ depends on $\mult^R$,
  but $\mult^R$ depends only on $t_{i,n}$.
  We first prove that
  $\biguplus_i \mult^R_{i, n} \in {\PropC_n}$
  by induction on $n$.
  For the base case $n = 0$,
  since $\iter_0 = \cskip$, the tree $t_{i,0}$ is $\Leaf(a_i)$.
  So,
  $\biguplus_i \mult^R_{i, 0}
  = \mset{a_i \mid i\in I, a_i = (\qket, \Store), \sem{\expr}_\Store = 1}
  \in \PropC_0$
  since $\mset{a_i \mid i \in I} \in \Prop_0$.
  For the step case,
  \begin{align*}
    \biguplus\nolimits_i \mult^R_{i, n + 1}
  & \ =\
    \biguplus\nolimits_i \mset{ (\qket, \Store) \mid (\qket, \Store) \mIn \Leaves(t_{i, n + 1}), \sem{\expr}_\Store = 1 }
  \\ & \ =\
    \biguplus\nolimits_i \{\nk4\vert (\qket['], \Store') \mid (\qket, \Store) \mIn \Leaves(t_{i, n}),
    \sem{\expr}_\Store  = 1,
  \\[-.3em] & \hspace{10em}
      (\qket['], \Store') \mIn \Leaves(\sem{C}(\qket, \Store)), \sem{\expr}_{\Store'} = 1
    \vert\nk4\}
  \\ & \ =\
    \biguplus\nolimits_i \mset{ (\qket['], \Store') \mid
    a \mIn \mult^R_{i, n},
    (\qket['], \Store') \mIn \Leaves(\sem{C}(a)), \sem{\expr}_{\Store'} = 1
    }
  \\ & \ =\
    \biguplus\nolimits_i \biguplus\nolimits_{a \mIn \mult^R_{i, n}} \mset{ (\qket['], \Store') \mid
      (\qket['], \Store') \mIn \Leaves(\sem{C}(a)), \sem{\expr}_{\Store'} = 1
    }
  \\ & \ =\
    \biguplus\nolimits_{a \mIn \biguplus\nolimits_i \mult^R_{i, n}} \mset{ (\qket['], \Store') \mid
      (\qket['], \Store') \mIn \Leaves(\sem{C}(a)), \sem{\expr}_{\Store'} = 1
    }
  \\ & \ \in\ \PropC_{n + 1}.
  \end{align*}
  The statements
   $\biguplus_i \mult^P_{i, n} \in \Prop_n$ and
   $\biguplus_i \mult^Q_{i, n} \in \PropB_n$ follow from this.

  By an easy induction, one can prove that
  the leaves of $t_{i,n}$ can be represented as
  $\mult^P_{i,n} \uplus \biguplus_{k < n} \mult^Q_{i,k}
  =
  \mult^R_{i,n} \uplus \biguplus_{k \leq n} \mult^Q_{i,k}$.
  Therefore, the leaves of $\sem{(\iter_n; \ifonly{\expr}{\eloop})}(a_i)$
  are $\biguplus_{k \leq n} \mult^Q_{i,k}$.
  Taking the supremum of these trees, we conclude that
  \[
    \biguplus_{i\in I} \Leaves(t_{i,\infty})
  \ =\
    \biguplus_{i\in I} \sup_{n\in \NN} \sem{(\iter_n; \ifonly{\expr}{\eloop})}(a_i)
  \ =\
    \biguplus_{i\in I, n\in \NN} \mult^Q_{i,n}
  \ \in\ \bigoplus_{n\in\NN} \PropB_n
    .
    \qedhere
  \]
\end{proof}

\subsection{Inner Product and Orthogonality}
\label{app:sect:logic:inner-prod}

Finally, to reason about the probability of the sum, we introduce the inner product and orthogonality, based on elementary linear algebra.

\begin{definition}[Inner product]\label{app:def:inner-prod}
  We define the \emph{inner product} of resources $\iprod[]{\pk1}{\pk1} \col A \times A \pto \CC$
  for each PCM $A \in \{ \Qstate, \dStore, \Res \}$ as follows.
  \begin{gather*}
    \iprod{\Qit \mapsto \qket}{\Qit \mapsto \qketB}_{\Qstate}
    \ \defeq\ \iprod{\qket}{\qketB}
    \ =\ \braket{\psi}{\phi}
  \hspace{2cm}
    \iprod{\Store}{\Store}_\dStore \defeq 1
  \\[.2em]
    \iprod{(\Qit \mapsto \qket, \Store)}{(\QitB \mapsto \qketB, \Store')}_\Res
    \ \defeq\
    \iprod{\Qit \mapsto \qket}{\QitB \mapsto \qketB}_{\Qstate}
    \cdot \iprod{\Store}{\Store'}_\dStore
  \end{gather*}
  The inner product is just the usual one on vectors
  whenever the classical information (including the set of owned qubits) agrees, and is undefined otherwise.

  Just like for the sum $+$ (\cref{fig:sum}, \cref{sect:logic:sum}),
  we define the inner product on multisets via multiset bijections $\mbij$
  and lift it to propositions by collecting the resulting values.
  \begin{gather*}
    \iprod[]{\mult}{\mult'}_\mbij \ \defeq\
    \sum_{(a,b)\mIn \mbij} \iprod[]{a}{b}
  \hspace{1.2cm}
    \iprod[]{\Prop}{\PropB} \ \defeq\
    \set{ \iprod[]{\mult}{\mult'}_\mbij }
        { \mult \in \Prop, \, \mult' \in \PropB, \, \mbij \col \mult \mbiject \mult' }
  \end{gather*}
  For utility, we also introduce the predicate $\iprod[]{\Prop}{\PropB} \col \alpha$ to assert that the inner product evaluates to $\alpha$ whenever defined:
  \[
    \iprod[]{\Prop}{\PropB} \col \alpha
    \quad\defeq\quad
    \iprod[]{\Prop}{\PropB} \subseteq \curly{\alpha}.
\tag*{\qed}
  \]
\end{definition}

\begin{remark}\label{app:rem:prob-by-inner-prod}
  The probability predicate can be characterized through the inner product as
  \[
    \Prop \col \Prob \prob
    \quad\iff\quad
    \all{\mult \in \Prop} \iprod[]{\mult}{\mult}_\id = \prob,
  \]
  where $\id$ is the canonical identity multiset bijection.
  Also, under the constraint $\Prop \col \precise,\unambig$, $\Prop \col \Prob \prob$ is equivalent to $\iprod[]{\Prop}{\Prop} \col \prob$.
\end{remark}

\begin{definition}[Orthogonality]
  The orthogonality $\Prop \orth \PropB$ of SL assertions is defined as follows:
  \[
    \Prop \orth \PropB \ \ \defeq\ \
      \iprod[]{\Prop}{\PropB} \col 0.
    \tag*{\qed}
  \]
\end{definition}

As an auxiliary notion for reasoning about the inner product and orthogonality, we also introduce the \emph{coherence}.

\begin{definition}[Coherence]
  We say $a, b \in \Res$ are \emph{coherent} and write $a \coh b$ if their inner product is defined, \ie $\iprod[]{a}{b} \defined$.
  This is equivalent to $a + b \defined$.
  We say multisets $\mult, \mult' \in \Mult(\Res)$ are \emph{coherent} if their inner product is defined for some multiset bijection $\mbij$, \ie $\ex{\mbij \col \mult \mbiject \mult'} \iprod[]{\mult}{\mult'}_{\mbij} \defined$.

  The coherence relation $\Prop \coh \PropB$ for SL assertions is defined as follows:
  \[
\Prop \coh \PropB
    \ \ \defeq\ \
    \all{\mult \in \Prop}\
    \all{\mult' \in \PropB}\
    \mult \coh \mult'
    \tag*{\qed}
  \]
\end{definition}

We have the following rules for the coherence $\coh$, whose soundness can be proved easily:
\begin{proofrules*}
\infer*[]{}{
    \bar{\qit} \mapsto \qket
    \,\coh\,
    \bar{\qit} \mapsto \qketB
  }

  \infer*[]{
    \Prop \col \precise
  }{
    \Prop \coh \Prop
  }

  \infer*[]{
    \Prop \coh \PropB
  }{
    \PropB \coh \Prop
  }

  \infer*[]{
    \Prop \coh \PropB
    \\
    \Prop' \vdash \Prop
  }{
    \Prop' \coh \PropB
  }

  \infer*[]{
    \all{x \in I}\pk2 \paren[\big]{\pk2 \Prop_x \coh \PropB \pk2}
  }{
    \paren{\ex{x \in I} \Prop_x} \,\coh\, \PropB
  }

  \infer*[]{
    \Prop \coh \PropB
    \\
    \Prop' \coh \PropB'
  }{
    \Prop * \Prop' \,\coh\, \PropB * \PropB'
  }

  \infer*[]{
    \Prop \coh \PropB
    \\
    \Prop' \coh \PropB'
  }{
    \Prop + \Prop' \,\coh\, \PropB + \PropB'
  }

  \infer*[]{
    \all{x \in I} \Prop_x \coh \PropB_x
  }{
    \bigoplusn_{x \in I} \Prop_x \,\coh\,
    \bigoplusn_{x \in I} \PropB_x
  }
\end{proofrules*}

Now we list the rules for the inner product and orthogonality.
\begin{proofrules*}
  \infer*[]{}{
    \iprod{\bar{\qit} \mapsto \qket}{\bar{\qit} \mapsto \qketB}
    \colon \langle \psi | \phi \rangle
  }

\infer*[]{}{
    \iprod[]{\Prop}{\PropB}
    \,=\, \overline{\iprod[]{\PropB}{\Prop}}
  }

  \infer*[lab=inner-prod-unique]{
    \Prop, \PropB \colon \precise
    \\
    \Prop \colon \unambig
  }{
    \ex{\alpha \in \CC}\
    \iprod[]{\Prop}{\PropB} \colon \alpha
  }
  \label{arule:inner-prod-unique}

  \infer*[]{
    \alpha \in \CC
  }{
    \iprod[]{\Prop}{\scl{\alpha}{\PropB}}
    \,=\, \alpha\pk2 \iprod[]{\Prop}{\PropB}
  }

  \infer*[lab=inner-prod-sum]{}{
    \iprod[]{\Prop}{\PropB + \PropC}
    \,\subseteq\,
    \iprod[]{\Prop}{\PropB} + \iprod[]{\Prop}{\PropC}
  }
  \label{arule:inner-prod-sum}

  \infer*[lab=inner-prod-sum-exact]{
    \Prop \colon \precise
  }{
    \iprod[]{\Prop}{\PropB + \PropC}
    \,=\, \iprod[]{\Prop}{\PropB} + \iprod[]{\Prop}{\PropC}
  }
  \label{arule:inner-prod-sum-exact}

  \infer*[lab=inner-prod-bigbmix]{}{
    \sum\nolimits_{x \in I} \iprod[]{\Prop_x}{\PropB_x}
    \,\subseteq\,
    \iprod{\bigoplusn_{x \in I} \Prop_x}{\bigoplusn_{x \in I} \PropB_x}
  }
  \label{arule:inner-prod-bigbmix}

  \infer*[lab=inner-prod-bigbmix-exact]{
    \all{x,y \in I \st x \neq y} \Prop_x \hash \PropB_y
  }{
    \sum\nolimits_{x \in I} \iprod[]{\Prop_x}{\PropB_x}
    \,=\, \iprod{\bigoplusn_{x \in I} \Prop_x}{\bigoplusn_{x \in I} \PropB_x}
  }
  \label{arule:inner-prod-bigbmix-exact}

  \infer*[lab=inner-prod-frame]{
    \all{\mult \in \Prop, \mult' \in \Prop'}
    \mult \cdot \mult' \defined
  }{
    \iprod[]{\Prop}{\PropB} \cdot \iprod[]{\Prop'}{\PropB'}
    \,\subseteq\,
    \iprod[]{\Prop * \Prop'}{\PropB * \PropB'}
  }
  \label{arule:inner-prod-frame}

  \infer*[lab=inner-prod-unframe]{
    \Prop \colon \unambig
    \\
    \Prop \coh \PropB
    \\
    \iprod[]{\PropC}{\PropC'} \col \alpha
  }{
    \iprod[]{\Prop * \PropC}{\PropB * \PropC'}
    \,\subseteq\,
    \alpha \iprod[]{\Prop}{\PropB}
  }
  \label{arule:inner-prod-unframe}

  \infer*[lab=orth-prob]{
    \Prop, \PropB \col \Prob \prob
    \\
    \Prop \orth \PropB
    \\
    \alpha, \beta \in \CC
    \\
    \abs{\alpha}^2 + \abs{\beta}^2 = 1
  }{
    \scl{\alpha}{\Prop} + \scl{\beta}{\PropB} \col \Prob \prob
  }
  \label{arule:orth-prob}

  \infer*[lab=orth-sum]{
    \Prop \orth \PropC
    \\
    \PropB \orth \PropC
  }{
    \Prop + \PropB \,\orth\, \PropC
  }
  \label{arule:orth-sum}

  \infer*[lab=orth-bigbmix]{
    \all{x \in I} \Prop_x \orth \PropB_x
    \\
    \hspace*{-.5em}
    \all{x, y \in I \st x \ne y} \Prop_x \hash \PropB_y
  }{
    \bigoplusn_{x \in I} \Prop_x \,\orth\, \bigoplusn_{x \in I} \PropB_x
  }
  \label{arule:orth-bigbmix}

  \infer*[lab=orth-frame]{
    \Prop, \PropC \colon \unambig
    \\
\Prop \coh \PropB
    \\
\Prop \orth \PropB
  }{
    \Prop \,*\, \PropC \,\orth\, \PropB \,*\, \PropC'
  }
  \label{arule:orth-frame}
\end{proofrules*}
We have the coherence side condition $\Prop \coh \PropB$ for the rules \ref{arule:inner-prod-unframe} and \ref{arule:orth-frame}.
Some of the above rules are used in the case study of the Shor code in \cref{app:sect:cases:shor-code}.

Before proving the soundness of these rules, we discuss useful properties.

\begin{lemma}[Basic properties of the inner product on $\Res$]
  The following hold for any $a, b, c, d \in \Res$:
  \begin{gather*}
    \iprod[]{a}{b} \,=\, \overline{\iprod[]{b}{a}}
  \hspace{4em}
    \iprod[]{1}{\alpha} \,=\, \alpha
  \\[.3em]
    \iprod[]{a + b}{c}
    \,=\, \iprod[]{a}{c} + \iprod[]{b}{c}
  \hspace{4em}
    \iprod[]{c}{a + b}
    \,=\, \iprod[]{c}{a} + \iprod[]{c}{b}
  \\[.3em]
    a \coh b\ \implies\
    \iprod[]{a \cdot c}{b \cdot d}
    \,=\, \iprod[]{a}{b} \cdot \iprod[]{c}{d}.
  \end{gather*}
\end{lemma}
\begin{proof}
  Immediately derived from the usual vector calculus.
\end{proof}
By combining these rules, we can, for example, prove the anti-linearity
of the first argument as
$\iprod[]{\scl{\alpha}{a} + \scl{\beta}{b}}{c}
= \iprod[]{\scl{\alpha}{a}}{c} + \iprod[]{\scl{\beta}{b}}{c}
= \iprod[]{\alpha \cdot a}{1 \cdot c}
  + \iprod[]{\beta \cdot b}{1 \cdot c}
= \iprod[]{\alpha}{1} \cdot \iprod[]{a}{c}
  + \iprod[]{\beta}{1} \cdot \iprod[]{b}{c}
= \overline{\iprod[]{1}{\alpha}} \cdot \iprod[]{a}{c}
  + \overline{\iprod[]{1}{\beta}} \cdot \iprod[]{b}{c}
= \bar{\alpha} \iprod[]{a}{c}
  + \bar{\beta} \iprod[]{b}{c}
$.

\begin{lemma}[More properties on $\Res$ and $\Mult(\Res)$]
  \label{app:lem:res-more}
  The following hold for any $a, b, c \in \Res$ and
  $\mult_1, \mult_2, \mult_3 \in \Mult(\Res)$.
  \begin{itemize}
\item If $a \coh b$ and $a \cdot c \defined$,
          then $(a + b) \cdot c \defined$.
    \item If $a \coh b$ and $a \hash c$,
          then $b \hash c$.
    \item If $\mult_1 +_\mbij \mult_2 \defined$
          and $\mult_1 \cdot \mult_3 \defined$,
          then $\mult_2 \cdot \mult_3 \defined$.
\end{itemize}
\end{lemma}
\begin{proof}
  Straightforward.
\end{proof}

Now we prove the soundness of the proof rules for the inner product and orthogonality.

\begin{proof}[Proof of \ref{arule:inner-prod-unique}]
  This can be proven similarly to
  \ref{lrule:sum-precise}.
\end{proof}

\begin{proof}[Proof of \ref{arule:inner-prod-sum}]
  An element in the left-hand side can be written as
  $\iprod[]{\mult}{\mult_1 +_\mbij \mult_2}_{\mbij'}$
  where $\mult \in \Prop$, $\mult_1 \in \PropB$,
  and $\mult_2 \in \PropC$
  with multiset bijections
  $\mbij \colon \mult_1 \mbiject \mult_2$
  and $\mbij' \colon \mult \mbiject (\mult_1 +_\mbij \mult_2)$.
  From these mutiset bijections,
  we obtain a canonical multiset bijection
  $\mbijB \colon \mult \mbiject \mult_1$ and
  $\mbijB' \colon \mult \mbiject \mult_2$.
  Then
  $\iprod[]{\mult}{\mult_1 +_\mbij \mult_2}_{\mbij'}
  =
  \iprod[]{\mult}{\mult_1}_{\mbijB}
  + \iprod[]{\mult}{\mult_2}_{\mbijB'} \in \text{RHS}$.
\end{proof}

\begin{proof}[Proof of \ref{arule:inner-prod-sum-exact}]
  If $\Prop = \{ \mult \}$,
  an element in the right-hand side can be written as
  $
  \iprod[]{\mult}{\mult_1}_{\mbijB}
  + \iprod[]{\mult}{\mult_2}_{\mbijB'}
  =
  \iprod[]{\mult}{\mult_1 +_\mbij \mult_2}_{\mbij'}
  \in \text{LHS}$.
\end{proof}

\begin{proof}[Proof of \ref{arule:inner-prod-bigbmix}]
  Any element in the left-hand side is written as
  $\sum_{x \in I} \iprod[]{\mult_x}{\mult'_x}_{\mbij_x}$,
  which is equal to
  $\iprod{
    \biguplus_x \mult_x
  }{
    \biguplus_x \mult'_x
  }_{
    \biguplus_{x \in I} \mbij_x
  }
  \in \text{RHS}
  $.
\end{proof}

\begin{proof}[Proof of \ref{arule:inner-prod-bigbmix-exact}]
  Let $\mult_x \in \Prop_x$ and $\mult'_x \in \PropB$.
  As in the proof of \ref{lrule:bigbmix-sum},
  any multiset bijection $\mbij$ between
  $\mult \defeq \biguplus_x \mult_x$
  and $\mult' \defeq \biguplus_x \mult'_x$
  that satisfy $\mult +_{\mbij} \mult' \defined$
  has the form $\biguplus_x \mbij_x$ for some
  $\mbij_x \colon \mult_x \mbiject \mult'_x$.
  Therefore,
  $\iprod[]{\biguplus_x \mult_x}{\biguplus_x \mult'_x}_{\mbij}
  = \sum_x \iprod{ \mult_x }{ \mult'_x }_{ \mbij_x }
  \in \text{LHS}$.
\end{proof}

\begin{proof}[Proof of \ref{arule:inner-prod-frame}]
  Any element in the left-hand side can be written as
  $\iprod[]{\mult_1 +_\mbij \mult_2}{\mult_3 +_{\mbijB} \mult_4}$
  where
  $\mult_1 \in \Prop$,
  $\mult_2 \in \PropB$,
  $\mult_3 \in \Prop'$, and
  $\mult_4 \in \PropB'$
  with multiset bijections
  $\mbij \colon \mult_1 \mbiject \mult_2$ and
  $\mbijB \colon \mult_3 \mbiject \mult_4$.
  Thanks to the assumption, $\mult_1 \cdot \mult_3$ is defined.
  Because $\mult_1 +_\mbij \mult_3$ and
  $\mult_2 +_\mbijB \mult_4$ are defined,
  $\mult_2 \cdot \mult_4 \defined$ follows
  from \cref{app:lem:res-more}.
  There is an induced multiset bijection
  $\mbij \times \mbijB \colon
  \mult_1 \cdot \mult_3 \mbiject \mult_2 \cdot \mult_4$
  that satisfies
  \begin{align*}
    \iprod[]{\mult_1 +_\mbij \mult_2}{\mult_3 +_{\mbijB} \mult_4}
    &\ =\
    \mset{
      \iprod[]{a}{b} \cdot \iprod[]{c}{d}
      \mid (a, b) \mIn \mbij, (c, d) \mIn \mbijB
    }
  \\ &\ =\
    \mset{
      \iprod[]{a \cdot c}{b \cdot d}
      \mid (a \cdot c, b \cdot d) \mIn \mbij \times \mbijB
    }
  \ \in\ \text{RHS}.
    \qedhere
  \end{align*}
\end{proof}

\begin{proof}[Proof of \ref{arule:inner-prod-unframe}]
  Let $\iprod[]{\mult \cdot \mult_1}{\mult' \cdot \mult_2}_\mbij$
  be an element of left-hand side where
  $\mult \in \Prop$,
  $\mult' \in \PropB$,
  $\mult_1 \in \PropC$,
  $\mult_2 \in \PropC'$,
  and $\mbij \col \mult \cdot \mult_1 \mbiject \mult' \cdot \mult_2$.
  Since $\Prop \coh \PropB$,
  there exists some multiset bijection
  $\mbijB_0 \col \mult \mbiject \mult'$ such that
  $\mult +_{\mbijB_0} \mult' \defined$.
  Since $\Prop$ is unambiguous,
  similarly to the proof of \ref{lrule:sum-precise},
  we can prove that such $\mbijB_0$ is unique.
  From \cref{app:lem:res-more},
  for any $a \mIn \mult$, $b \mIn \mult'$,
  $(a, b) \mIn \mbijB_0$ or $a \hash b$.
Because
  $\iprod[]{\mult \cdot \mult_1}{\mult' \cdot \mult_2}_\mbij$
  is defined,
  for any $(a \cdot c, b \cdot d) \mIn \mbij$,
  $\iprod[]{a \cdot c}{b \cdot d} \defined$,
  which is equivalent to say $a \cdot c + b \cdot d \defined$.
  So $\neg (a \hash b)$, thus $(a, b) \mIn \mbijB_0$.
  Therefore,
  $(a \cdot c, b \cdot d) \mIn \mbij \implies (a, b) \mIn \mbijB_0$,
  and $\iprod[]{a \cdot c}{b \cdot d} =
  \iprod[]{a}{b} \cdot \iprod[]{c}{d}$.
  For each $a \mIn \mult$,
  let $\mbij_{a}$ be the multiset bijection between
  $\mult_1$ and $\mult_2$ such that
  $(a \cdot c, b \cdot d) \mIn \mbij \implies (c, d) \mIn \mbij_a$
  Then
  \begin{align*}
    &
    \sum\nolimits_{(a \cdot c, b \cdot d) \mIn \mbij}
    \iprod{a \cdot c}{b \cdot d}
  \\ &\ =\
    \sum\nolimits_{(a, b) \mIn \mbijB_0}
    \sum\nolimits_{(c, d) \mIn \mbij_a}
    \iprod{a}{b} \cdot \iprod{c}{d}
  \\ &\ =\
    \sum\nolimits_{(a, b) \mIn \mbijB_0} \iprod{a}{b}
    \cdot
    \sum\nolimits_{(c, d) \mIn \mbij_a} \iprod{c}{d}
  \\ &\ =\
    \sum\nolimits_{(a, b) \mIn \mbij_0} \alpha \iprod{a}{b}
  \\ &\ =\
    \alpha\pk2 \iprod{\mult}{\mult'}_{\mbij_0}
    \ \in\ \alpha \iprod{\Prop}{\PropB}.
    \qedhere
  \end{align*}
\end{proof}

\begin{proof}[Proof of \ref{arule:orth-prob}]
  Let $\mult \in \Prop$, $\mult' \in \PropB$,
  and $\mbij \colon \mult \mbiject \mult'$.
  Then
  \begin{align*}
    \iprod{\alpha\mult +_\mbij \beta\mult'}
          {\alpha\mult +_\mbij \beta\mult'}_\id
    &\ =\
      \bar{\alpha}\alpha\iprod[]{\mult}{\mult}_\id
      +
      \bar{\alpha}\beta\iprod[]{\mult}{\mult'}_\mbij
      +
      \bar{\beta}\alpha\iprod[]{\mult'}{\mult}_{\mbij^{-1}}
      +
      \bar{\beta}\beta\iprod[]{\mult'}{\mult'}_\id
  \\ &\ =\
      |\alpha|^2 + 0 + 0 + |\beta|^2 \ =\ 1.
    \qedhere
  \end{align*}
\end{proof}

\begin{proof}[Proof of \ref{arule:orth-sum}]
  Follows from \ref{arule:inner-prod-sum}.
\end{proof}

\begin{proof}[Proof of \ref{arule:orth-bigbmix}]
  Follows from \ref{arule:inner-prod-bigbmix-exact}.
\end{proof}

\begin{proof}[Proof of \ref{arule:orth-frame}]
  Let
  $\mult \in \Prop$,
  $\mult' \in \PropB$,
  $\mult_1 \in \PropC$,
  $\mult_2 \in \PropC'$,
  and $\mbij \col \mult \cdot \mult_1 \mbiject \mult' \cdot \mult_2$,
  such that
  $\iprod[]{\mult \cdot \mult_1}{\mult' \cdot \mult_2}_\mbij \defined$.
  Similarly to the proof of \ref{arule:inner-prod-unframe},
  we can prove that,
  there exist $\mbijB_0 \colon \mult \mbiject \mult'$
  and $\mbij_a \colon \mult_1 \mbiject \mult_2$
  such that
  $(a \cdot c, b \cdot d) \mIn \mbij \implies
  (a, b) \mIn \mbijB_0 \,\land\, (c, d) \mIn \mbij_a$,
  $\mult +_{\mbijB_0} \mult' \defined$, and
  $\mult_1 +_{\mbij_a} \mult_2 \defined$.
  Moreover, since $\PropC$ is unambiguous,
  $\{ \mbijB' \colon \mult_1 \mbiject \mult_2
  \mid \mult_1 +_{\mbijB'} \mult_2 \defined \}$
  is a singleton or empty.
  If not empty, we denote the unique element as $\mbijB'_{0}$, and then
  \begin{align*}
    \sum\nolimits_{(a \cdot c, b \cdot d) \mIn \mbij}
    \iprod{a \cdot c}{b \cdot d}
    \ &=\
    \sum\nolimits_{(a, b) \mIn \mbijB_0} \iprod{a}{b}
    \cdot
    \sum\nolimits_{(c, d) \mIn \mbij_a} \iprod{c}{d}
    \\ & =\
    \sum\nolimits_{(a, b) \mIn \mbijB_0} \iprod{a}{b}
    \cdot
    \sum\nolimits_{(c, d) \mIn \mbijB'_0} \iprod{c}{d}
    \\ &\in\
    \{ 0 \}.
  \end{align*}
  If such $\mbijB'_0$ does not exist, then $\mult = \mult' = \mset{\ }$.
  The orthogonality is trivial in this case.
\end{proof}
 \section{More on Case Studies}
\label{app:sect:cases}

\subsection{Dirty Qubit: Implementation of CCCX by Toffoli Gates}
\label{app:sect:cases:dirty}

Here we give the details of \cref{sect:cases:dirty}.

The \gCCCX{} gate is a 4-qubit gate that computes $\qitRet \store \qitRet \xor (\qit \land \qitB \land \qitC)$ where $\xor$ denotes exclusive OR.
We can implement this \gCCCX{} gate using only \gCCX{} gates (computing $\qitRet \store \qitRet \xor (\qit \land \qitB)$, also called Toffoli gates) with a dirty auxiliary qubit $\qitTmp$, a qubit whose state is unknown, as follows:
\[
  \gdCCCX[\qit, \qitB, \qitC, \qitTmp, \qitRet] \ \defeq\
    \gCCX[\qitC, \qitTmp, \qitRet]; \gCCX[\qit, \qitB, \qitTmp];
    \gCCX[\qitC, \qitTmp, \qitRet]; \gCCX[\qit, \qitB, \qitTmp].
\]
The $\gCCCX$ gate and the circuit $\gdCCCX$ are visualized in \cref{fig:dcccx-simple}.

We would like to prove the equivalence between $\gCCCX$ and $\gdCCCX$ in our logic in any context, \ie with any other qubits entangled with the qubits involved.
We show the following Hoare triple:
\[
  \hoare{ (\qit, \qitB, \qitC, \qitTmp, \qitRet, \bar{\qitD}) \mapsto \qket }{ \gdCCCX }{
    (\qit, \qitB, \qitC, \qitTmp, \qitRet, \bar{\qitD}) \mapsto (\gCCCX_{\qit, \qitB, \qitC, \qitRet} \otimes \gid_{\qitTmp, \bar{\qitD}}) \qket
  }.
\]
To this end, we first prove the following Hoare triple for any classical state $i, j, k, \ell, m \in \{0, 1\}$:
\[
  \hoare{ (\qit, \qitB, \qitC, \qitTmp, \qitRet) \mapsto \ket{i j k \ell m} }{ \gdCCCX }{
    (\qit, \qitB, \qitC, \qitTmp, \qitRet) \mapsto \ket{i j k \ell(m \xor(i \land j \land k))}
  }.
\]
This can be done as follows.
Here, we simply write $\land$ as multiplication.
Note that this part is no different from symbolic execution of the usual classical bitwise operations.
\begin{align*} &
  \pkcurly[\big]{ \qit \mapsto \ket{i} \pk1*\pk1 \qitB \mapsto \ket{j} \pk1*\pk1 \qitC \mapsto \ket{k} \pk1*\pk1
    \qitTmp \mapsto \ket{\ell} \pk1*\pk1 \qitRet \mapsto \ket{m} }
\\[.2em] &
  \gCCX[\qitC, \qitTmp, \qitRet]
\\ &
  \pkcurly[\big]{ \qit \mapsto \ket{i} \pk1*\pk1 \qitB \mapsto \ket{j} \pk1*\pk1 \qitC \mapsto \ket{k} \pk1*\pk1
    \qitTmp \mapsto \ket{\ell} \pk1*\pk1 \qitRet \mapsto \ket{m \xor k \ell} }
\\[.2em] &
  \gCCX[\qit, \qitB, \qitTmp]
\\ &
  \pkcurly[\big]{ \qit \mapsto \ket{i} \pk1*\pk1 \qitB \mapsto \ket{j} \pk1*\pk1 \qitC \mapsto \ket{k} \pk1*\pk1
    \qitTmp \mapsto \ket{\ell \xor i j} \pk1*\pk1 \qitRet \mapsto \ket{m \xor k \ell} }
\\[.2em] &
  \gCCX[\qitC, \qitTmp, \qitRet]
\\ &
  \pkcurly[\big]{ \qit \mapsto \ket{i} \pk1*\pk1 \qitB \mapsto \ket{j} \pk1*\pk1 \qitC \mapsto \ket{k} \pk1*\pk1
    \qitTmp \mapsto \ket{\ell \xor i j} \pk1*\pk1 \qitRet \mapsto \ket{m \xor k \ell \xor k (\ell \xor i j)} }
\\ &
  \pkcurly[\big]{ \qit \mapsto \ket{i} \pk1*\pk1 \qitB \mapsto \ket{j} \pk1*\pk1 \qitC \mapsto \ket{k} \pk1*\pk1
    \qitTmp \mapsto \ket{\ell \xor i j} \pk1*\pk1 \qitRet \mapsto \ket{m \xor i j k} }
\\[.2em] &
  \gCCX[\qit, \qitB, \qitTmp]
\\ &
  \pkcurly[\big]{ \qit \mapsto \ket{i} \pk1*\pk1 \qitB \mapsto \ket{j} \pk1*\pk1 \qitC \mapsto \ket{k} \pk1*\pk1
    \qitTmp \mapsto \ket{\ell \xor i j \xor i j} \pk1*\pk1 \qitRet \mapsto \ket{m \xor i j k} }
\\ &
  \pkcurly[\big]{\qit \mapsto \ket{i} \pk1*\pk1 \qitB \mapsto \ket{j} \pk1*\pk1 \qitC \mapsto \ket{k} \pk1*\pk1
    \qitTmp \mapsto \ket{\ell} \pk1*\pk1 \qitRet \mapsto \ket{m \xor i j k} }
\end{align*}
Now that we have proven the Hoare triple for the concrete state $\ket{i j k \ell m}$, we can generalize it to any quantum state $\qket$ as follows:
\begin{derivation}
  \infer*[Right = \labelstep{ax:dcccx:sum}]{
    \infer*[Right = \labelstep{ax:dcccx:frame}]{
      \all{i, j, k, \ell, m \in \{0, 1\}}
    \\\\
      \hoare{ (\qit, \qitB, \qitC, \qitTmp, \qitRet) \mapsto \ket{i j k \ell m} }
        { \gdCCCX }
        { (\qit, \qitB, \qitC, \qitTmp, \qitRet) \mapsto
          (\gCCCX_{\qit, \qitB, \qitC, \qitRet} \otimes \gid_\qitTmp) \ket{i j k \ell m} }
    }{
      \all{i, j, k, \ell, m \in \{0, 1\}}
    \\\\
      \pkcurly[\big]{
        (\qit, \qitB, \qitC, \qitTmp, \qitRet) \mapsto
          \alpha_{i j k \ell m} \ket{i j k \ell m} *
        \bar{\qitD} \mapsto \qket[_{i j k \ell m}]
      }
      \ \ \gdCCCX
    \\
      \pkcurly[\big]{
        (\qit, \qitB, \qitC, \qitTmp, \qitRet) \mapsto
          (\gCCCX_{\qit, \qitB, \qitC, \qitRet} \otimes \gid_\qitTmp) \alpha_{i j k \ell m} \ket{i j k \ell m}
        * \bar{\qitD} \mapsto \qket[_{i j k \ell m}]
      }
    }
  }{ \textstyle
    \pkcurly[\big]{
      (\qit, \qitB, \qitC, \qitTmp, \qitRet, \bar{\qitD}) \mapsto
        \sum_{i, j, k, \ell, m} \alpha_{i j k \ell m} \ket{i j k \ell m} \qket[_{i j k \ell m}]
    }
    \ \ \ \gdCCCX
  \\ \textstyle
    \pkcurly[\big]{
      (\qit, \qitB, \qitC, \qitTmp, \qitRet, \bar{\qitD}) \mapsto
        (\gCCCX_{\qit, \qitB, \qitC, \qitRet} \otimes \gid_{\qitTmp, \bar{\qitD}})
        \sum_{i, j, k, \ell, m} \alpha_{i j k \ell m} \ket{i j k \ell m} \qket[_{i j k \ell m}]
    }
  }
\end{derivation}
We first use \ref{lrule:hoare-frame} to add the auxiliary qubits $\bar{\qitD} \mapsto \qket[_{i j k \ell m}]$ and the global phase $\alpha_{i j k \ell m}$ in the pre/postconditions
at \cref{ax:dcccx:frame}.
Finally, we use \ref{lrule:hoare-sum} to sum over all $i j k \ell m \in \{0, 1\}$ to obtain the last line \cref{ax:dcccx:sum}.
We used the linearity of the \gCCCX{} gate to commute the \gCCCX{} gate with the sum $\sum_{i, j, k, \ell, m}$ and the global phase $\alpha_{i j k \ell m}$.
Since the choice of $\alpha_{i j k \ell m}$ and $\qket[_{i j k \ell m}]$ is arbitrary, and every quantum state can be written in the form $\sum_{i, j, k, \ell, m} \alpha_{i j k \ell m} \ket{i j k \ell m} \qket[_{i j k \ell m}]$, we conclude that the Hoare triple holds for any quantum state $\qket$.

\subsection{Program with Measurements: EPR Paradox}
\label{app:sect:cases:epr}

\newcommand*{\EPR}{\mathsf{EPR}}
The EPR paradox is a famous quantum phenomenon that shows that the measurement results of two entangled qubits are correlated, no matter how far apart the two qubits are.
The \emph{Bell state} or \emph{EPR pair} is a state of 2-qubits ($\qit$ and $\qitB$) written as $\tfrac{1}{\sqrt2} (\ket{00} + \ket{11})$.
We suppose that those entangled qubits are shared between Alice and Bob.
If Alice measures her qubit $\qit$ and Bob measures his qubit $\qitB$, the EPR paradox states that their measurement results agree with probability 1, even though the outcomes are not determined prior to measurement.
The program can be written as follows, where the Bell state is prepared in the first step:
\[
  \EPR^{\var, \varB, \varR}[\qit, \qitB] \ \ \defeq\ \
    \Bell[\qit, \qitB];\
    \var \store \gMZ[\qit];\
    \varB \store \gMZ[\qitB];\
    \varR \store \var \xor\nk1 \varB.
\]
The classical boolean variable $\varR$ is added to the program to indicate the difference between the measurement results of Alice and Bob.
Therefore, the correctness of the program can be specified as the following assertion:
\begin{align*}
&
  \ex{\pk3 \Prop \col \frameable, \Prob 1,\pk4 \PropB \col \frameable, \Prob 0}
\\[-.2em] & \hspace{3em}
  \hoare{ (\qit, \qitB) \mapsto \ket{00} }[^{\var, \varB, \varR}]
    { \EPR^{\var, \varB, \varR}[\qit, \qitB] }
    { (\varR \mapsto 0 * \Prop) \pk2\oplus\pk2 (\varR \mapsto 1 * \PropB) }.
\end{align*}
This assertion can be proved as follows,
setting $\PropC_{k, \ell} \defeq \var \mapsto k * \varB \mapsto \ell * (\qit, \qitB) \mapsto \ket{k \ell}$:
\begin{align*} &
  \pkcurly[\big]{ (\qit, \qitB) \mapsto \ket{00} }^{\var, \varB, \varR}
\\ & \textstyle
  \Bell[\qit, \qitB]
  \quad \pkcurly[\big]{
    \tfrac{1}{\sqrt2} \sum_{i = 0, 1} (\qit, \qitB) \mapsto \ket{i i}
  }^{\var, \varB, \varR}
\\ &
  \var \store \gMZ[\qit]
  \quad \pkcurly[\big]{
    \tfrac{1}{\sqrt2} \bigoplusn_{i = 0, 1}
    \var \mapsto i * (\qit, \qitB) \mapsto \ket{i i}
  }^{\varB, \varR}
\\ &
  \varB \store \gMZ[\qitB]
  \quad \pkcurly[\big]{
    \tfrac{1}{\sqrt2} \bigoplusn_{i = 0, 1} \bigoplusn_{j = 0, 1}
    \delta_{i j} * \var \mapsto i * \varB \mapsto i * (\qit, \qitB) \mapsto \ket{i i}
  }^\varR
\\ &
  \pkcurly[\Big]{
    \paren*{\pk2
      \tfrac{1}{\sqrt2} \bigoplusn_{k = 0, 1} \PropC_{k, k}
    \pk2} \ \oplus\
    \paren*{\pk2
      0 \cdot \bigoplusn_{k = 0, 1} \PropC_{k, \neg k}
    \pk2}
  }^\varR
\\ &
  \varR \store \var \xor\nk1 \varB
\\ &
  \pkcurly[\Big]{
    \paren*{\pk2
      \tfrac{1}{\sqrt2}
      \bigoplusn_{k = 0, 1} \varR \mapsto 0 * \PropC_{k, k}
    \pk2} \ \oplus\
    \paren*{\pk2
      0 \cdot \bigoplusn_{k = 0, 1} \varR \mapsto 1 * \PropC_{k, \neg k}
    \pk2}
  }
\\ &
  \pkcurly[\Big]{
    \paren*{\pk2
      \varR \mapsto 0 \,*\,
      \tfrac{1}{\sqrt2} \bigoplusn_{k = 0, 1} \PropC_{k, k}
    \pk2}
    \ \oplus\
    \paren*{\pk2
      \varR \mapsto 1 \,*\,
      0 \cdot \bigoplusn_{k = 0, 1} \PropC_{k, \neg k}
    \pk2}
  }.
\end{align*}
We can finally set
$\Prop \defeq \tfrac{1}{\sqrt2} \bigoplusn_{k = 0, 1} \PropC_{k, k}$ and
$\PropB \defeq 0 \cdot \bigoplusn_{k = 0, 1} \PropC_{k, \neg k}$.
Note that $\Prop \col \Prob 1$ holds because $\PropC_{k, \ell} \col \Prob 1$ and $\paren[\big]{\frac{1}{\sqrt2}}^2 \cdot (1 + 1) = 1$.

\subsection{Quantum Teleportation}
\label{app:sect:cases:teleportation}

Here we give the details of \cref{sect:cases:teleportation}.

Quantum teleportation is a protocol that allows one to send the state of a qubit $\qket$ from Alice to Bob using only classical communication and pre-shared entangled qubits.
The protocol is written as the following circuit:
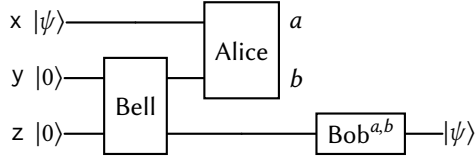
\[
  \begin{quantikz}[row sep = 0.2cm]
    \lstick{\qit} \qket & & \gate[2]{\Alice} \rstick{\var}
  \\
    \lstick{\qitB} \ket{0} & \gate[2]{\Bell} & \rstick{\varB}
  \\
    \lstick{\qitC} \ket{0} & & & \gate{\Bob^{\smash{\var, \varB}}} & \qket
  \end{quantikz}
\]
In our logic, the correctness of the teleportation program can be specified as follows:
\begin{align*} &
  \ex{\, \Prop \col \frameable, \Prob 1}\
  \all{\qket}
\\[-.3em] & \hspace{3em}
  \hoare{ \qit \mapsto \qket * \qitB \mapsto \ket{0} * \qitC \mapsto \ket{0} }[^{\var, \varB}]
    { \Teleport^{\var, \varB}[\qit, \qitB, \qitC] }
    { \qitC \mapsto \qket * \Prop }.
\end{align*}
The specification states that the final state of the qubit $\qitC$ will be the same as the initial state of the qubit $\qit$.

The first step of the protocol is to distribute the entangled qubits $\qitB$ and $\qitC$ to Alice and Bob.
The specification of this preparation step is as follows:
\[
  \hoare{\, \qitB \mapsto \ket{0} * \qitC \mapsto \ket{0} \, }
    { \Bell[\qitB, \qitC] }
    {\, (\qitB, \qitC) \mapsto \tfrac{1}{\sqrt2}(\ket{00} + \ket{11}) \, }.
\]
Then, Alice performs the following program to generate the classical bits $\var$ and $\varB$ that will be sent to Bob:
\[
  \Alice^{\var, \varB}[\qit, \qitB]
\hspace{.5em}= \hspace{.5em}
  \begin{quantikz}[baseline = {(0, -0.2ex)}, row sep = 0.2cm]
    \qit\, & \ctrl{1} & \gate{\gH} & \meter{} \rstick{\var}\\
    \qitB\, & \targ{} & \qw & \meter{} \rstick{\varB}
  \end{quantikz}
\hspace{.5em}= \hspace{.5em}
  \gCX[\qit, \qitB];\, \gH[\qit];\, \var \store \gMZ[\qit]; \varB \store \gMZ[\qitB].
\]
Finally, Bob receives the classical bits $\var$ and $\varB$ from Alice, and performs the following program to recover the state $\qket$:
\[
  \Bob^{\var, \varB}[\qitC]
\hspace{.5em}= \hspace{.5em}
  \begin{quantikz}[baseline = {(0, -.4ex)}]
    \qitC\, & \gate{\ifonly{\varB}{\gX}} & \gate{\ifonly{\var}{\gZ}} &
  \end{quantikz}
\hspace{.5em}= \hspace{.5em}
  \ifonly{\varB}{\gX[\qitC]};\ \ifonly{\var}{\gZ[\qitC]}.
\]

We prove the correctness of the protocol modularly, \ie we prove the specifications of $\Alice$ and $\Bob$ separately, and then combine them to prove the specification of the whole teleportation program.
We first prove the specification of $\Alice$ assuming the inputs are in the classical state $\ket{x i}$ for $x, i \in \{0, 1\}$:
\[
  \hoare{ (\qit, \qitB) \mapsto \ket{x i} }[^{\var, \varB}]
    { \Alice^{\var, \varB}[\qit, \qitB] }
    { \tfrac{1}{\sqrt2} \bigoplusn^{\var, \varB} (-1)^{x \land \var}\,
      \delta_{x \xor i, \varB} \,\cdot\,
      (\qit, \qitB) \mapsto \ket{\var \varB} }.
\]
This proof can be done symbolically as follows:
\begin{align*} &
  \pkcurly[\big]{
    (\qit, \qitB) \mapsto \ket{x i}
  }^{\var, \varB}
  \ \
  \gCX[\qit, \qitB]
  \ \
  \pkcurly[\big]{
    (\qit, \qitB) \mapsto \ket{x (x \xor i)}
  }^{\var, \varB}
\\ &
  \gH[\qit] \ \
  \pkcurly[\Big]{ \tfrac{1}{\sqrt2}
    \sumn_{k = 0, 1} (-1)^{x \land k} * (\qit, \qitB) \mapsto \ket{k (x \xor i)} }^{\var, \varB}
\\ &
  \var \store \gMZ[\qit] \ \
  \pkcurly[\Big]{ \tfrac{1}{\sqrt2}
    \bigoplusn^{\var}
    (-1)^{x \land \var} * (\qit, \qitB) \mapsto \ket{\var (x \xor i)} }^\varB
\\ &
  \varB \store \gMZ[\qitB] \ \
  \pkcurly[\Big]{ \tfrac{1}{\sqrt2}
    \bigoplusn^{\var, \varB}
      (-1)^{x \land \var} \cdot \delta_{(x \xor i), \varB} *
      (\qit, \qitB) \mapsto \ket{\var \varB} }.
\end{align*}
Note that we are also identifying the classical variables $\var$ and $b$
with the boolean values they store.
We can also prove the following specification of $\Bob$ in a straightforward way:
\[
  \hoare{ \qitC \mapsto \ket{i} \,*\,
    \var \mapsto a \,*\, \varB \mapsto x \xor i }
    { \Bob^{\var, \varB}[\qitC] }
    { (-1)^{x \land a}\, \qitC \mapsto \ket{x} }.
\]
By combining the above specifications, we prove the following specification of the whole teleportation program for the classical input $\ket{x}$ for $x \in \{0, 1\}$:
\[
  \ex{\, \Prop \col \frameable, \Prob 1}\ \
  \hoare{\, \qit \mapsto \ket{x} * \qitB \mapsto \ket{0} * \qitC \mapsto \ket{0} \,}[^{\var, \varB}]
    {\, \Teleport^{\var, \varB}[\qit, \qitB, \qitC] \,}
    {\, \qitC \mapsto \ket{x} * \Prop \,}.
\]
We can prove this by the following derivation:
\begin{align*} &
  \pkcurly[\big]{ \qit \mapsto \ket{x} \,*\, (\qitB, \qitC) \mapsto \ket{00}}^{\var, \varB}
\\ &
  \Bell[\qitB, \qitC]
\\ &
  \pkcurly[\big]{ \qit \mapsto \ket{x} \,*\,
    (\qitB, \qitC) \mapsto \tfrac{1}{\sqrt2}(\ket{00} + \ket{11}) }^{\var, \varB}
  \ \
  \pkcurly[\big]{ \tfrac{1}{\sqrt2} \sumn_{i = 0, 1}
    \qit \mapsto \ket{x} \,*\, (\qitB, \qitC) \mapsto \ket{i i} }^{\var, \varB}
\\ &
  \Alice^{\var, \varB}[\qit, \qitB]
\\ &
  \pkcurly[\Big]{
    \tfrac{1}{2} \sumn_{i = 0, 1} \bigoplusn^{\var, \varB}
    (-1)^{x \land \var}\, \delta_{(x \xor i), \varB} \,\cdot\,
    (\qit, \qitB) \mapsto \ket{\var \varB} \,*\,
    \qitC \mapsto \ket{i} }
\\ &
  \Bob^{\var, \varB}[\qitC]
\\ &
  \pkcurly[\Big]{
    \tfrac{1}{2} \sumn_{i = 0, 1} \bigoplusn^{\var, \varB}
    \delta_{(x \xor i), \varB} \,\cdot\,
    (\qit, \qitB) \mapsto \ket{\var \varB} \,*\,
    \qitC \mapsto \ket{x}
  }
\\ &
  \pkcurly[\Big]{
    \tfrac{1}{2} \bigoplusn^{\var, \varB} \sumn_{i = 0, 1}
    \delta_{(x \xor i), \varB} \,\cdot\,
    (\qit, \qitB) \mapsto \ket{a \varB} \,*\,
    \qitC \mapsto \ket{x}
  }
\\ &
  \pkcurly[\Big]{
    \tfrac{1}{2} \bigoplusn^{\var, \varB}
    (\qit, \qitB) \mapsto \ket{a \varB} \,*\,
    \qitC \mapsto \ket{x}
  }
  \ \
  \pkcurly[\Big]{
    \qitC \mapsto \ket{x} \,*\,
    \bigoplusn^{\var, \varB} (\qit, \qitB) \mapsto \tfrac{1}{2} \ket{a b}
  }.
\end{align*}
Now we can set $\Prop \defeq \bigoplusn^{\var, \varB} (\qit, \qitB) \mapsto \tfrac{1}{2} \ket{\var \varB}$ because it is $\frameable$ and $\Prob 1$.
In the general case, the input state of $\qit$ can be any quantum state $\qket = \alpha_0 \ket{0} + \alpha_1 \ket{1}$.
The Hoare triple we have just proved can be immediately generalized using \ref{lrule:hoare-scale} and \ref{lrule:hoare-sum}:
\begin{derivation}
  \infer*[Right = \ref{lrule:hoare-sum}]{
    \infer*[Right = \ref{lrule:hoare-scale}]{
      \all{x}\
      \hoare{ \qit \mapsto \ket{x} * (\qitB, \qitC) \mapsto \ket{00} }[^{\var, \varB}]
        { \Teleport^{\var, \varB}[\qit, \qitB, \qitC] }
        { \qitC \mapsto \ket{x} * \Prop }
    }{
      \all{x}\
      \hoare{ \qit \mapsto \alpha_{x} \ket{x} * (\qitB, \qitC) \mapsto \ket{00} }[^{\var, \varB}]
        { \Teleport^{\var, \varB}[\qit, \qitB, \qitC] }
        { \qitC \mapsto \alpha_x \ket{x} * \Prop }
    }
  }{
    \hoare{ \qit \mapsto \qket * (\qitB, \qitC) \mapsto \ket{00} }[^{\var, \varB}]
      { \Teleport^{\var, \varB}[\qit, \qitB, \qitC] }
      { \qitC \mapsto \qket * \Prop }
  }
\end{derivation}

\subsection{Lattice Surgery: Implementation of CNOT with Measurements}
\label{app:sect:cases:meas-cnot}

Here we give the details of \cref{sect:cases:meas-cnot}.

Lattice surgery is a technique for fault-tolerant quantum computing that uses surface codes~\cite{DennisKLP02-surface-codes}.
In lattice surgery, the \gCNOT{} gate is implemented without using two-qubit gates, but using only two-qubit measurements, as shown in \cref{fig:mcnot}.
The right-hand side circuit \mCNOT{} implements the \gCNOT{} gate using two 2-qubit measurements, $\gMXX$ and $\gMZZ$,
and some single-qubit gates with an auxiliary qubit $\qitB$ initialized to $\ket{0}$.
The program $\mCNOT^{\hvar, \hvarB, \hvarC}[\qit, \qitB, \qitC]$ for the circuit has been given in \cref{sect:cases:meas-cnot}.
We want to formally prove that this circuit indeed implements the \gCNOT{} gate. That is, we want to prove the following specification:
\begin{align*} &
  \ex{\, \Prop \col \frameable, \Prob 1}\
  \all{\qket}
\\[-.3em] & \hspace{3em}
  \hoare{ (\qit, \qitC) \mapsto \qket * \qitB \mapsto \ket{0} }[^{\hvar, \hvarB, \hvarC}]
    { \mCNOT^{\hvar, \hvarB, \hvarC}[\qit, \qitB, \qitC] }
    { (\qit, \qitC) \mapsto \gCX \qket * \Prop }.
\end{align*}

The specification of 2-qubit measurements $\gMXX$ and $\gMZZ$ is given as follows.
\begin{align*}
  \all{\pk2 s \in \{+, -\}} &\
  \hoare{ (\qit, \qitB) \mapsto \ket{s s} }
    { \gMXX^\var[\qit, \qitB] }
    { (\qit, \qitB) \mapsto \ket{s s} \mix{\var} (\qit, \qitB) \mapsto 0 }
\\
  \all{\pk2 s, t \in \{+, -\} \st s \neq t} &\
  \hoare{ (\qit, \qitB) \mapsto \ket{s t} }
    { \gMXX^\var[\qit, \qitB] }
    { (\qit, \qitB) \mapsto 0 \mix{\var} (\qit, \qitB) \mapsto \ket{s t} }
\\
  \all{\pk2 i \in \{0, 1\}} &\
  \hoare{ (\qit, \qitB) \mapsto \ket{i i} }
    { \gMZZ^\var[\qit, \qitB] }
    { (\qit, \qitB) \mapsto \ket{i i} \mix{\var} (\qit, \qitB) \mapsto 0 }
\\
  \all{\pk2 i, j \in \{0, 1\} \st i \neq j} &\
  \hoare{ (\qit, \qitB) \mapsto \ket{i j} }
    { \gMZZ^\var[\qit, \qitB] }
    { (\qit, \qitB) \mapsto 0 \mix{\var} (\qit, \qitB) \mapsto \ket{i j} }
\end{align*}
By linearity, we can assume that the initial state of $\qit$ and $\qitC$ is some disentangled state.
Here, we take the initial state to be $(\qit, \qitC) \mapsto \ket{a}(\gH \ket{b})$ for $a, b \in \{0, 1\}$.
The qubit $\qitC$ is not in the classical state $\ket{0}$ or $\ket{1}$,
but we choose $\ket{\plus} \defeq \gH \ket{0}$ and $\ket{\minus} \defeq \gH \ket{1}$ for simplicity.
Since $\angl{\ket{\plus}, \ket{\minus}}$ spans the whole 2-dimensional Hilbert space $\CC^2$,
this is enough to prove the correctness of the program.
The derivation is as follows:
\begin{align*} &
  \pkcurly[\big]{
    \qit \mapsto \ket{a} \,*\,
    \qitB \mapsto \ket{0} \,*\, \qitC \mapsto \gH \ket{b}
  }^{\hvar, \hvarB, \hvarC}
\\[.4em] &
  \pkcurly[\big]{
    \qit \mapsto \ket{a} \,*\,
    \tfrac{1}{\sqrt2} \sumn_{c}
      \qitB \mapsto \gH \ket{c} \,*\, \qitC \mapsto \gH \ket{b}
  }^{\hvar, \hvarB, \hvarC}
\\[.4em] &
\gMXX^\hvar[\qitB, \qitC]
\\[.4em] &
  \pkcurly[\big]{
    \qit \mapsto \ket{a} \,*\,
    \tfrac{1}{\sqrt2} \bigoplusn^{\hvar} \sumn_{c}
      \delta_{\hvar, c \xor b} \cdot
      \qitB \mapsto \gH \ket{c} \,*\, \qitC \mapsto \gH \ket{b}
  }^{\hvarB, \hvarC}
\\[.4em] &
  \pkcurly[\big]{
    \qit \mapsto \ket{a} \,*\,
    \tfrac{1}{\sqrt2} \bigoplusn^{\hvar}
      \qitB \mapsto \gH \ket{b \xor \hvar} \,*\, \qitC \mapsto \gH \ket{b}
  }^{\hvarB, \hvarC}
  \hspace{1.5em} (c \coloneq b \xor \hvar)
\\[.4em] &
\ifonly{ \hvar }{ \gZ[\qit] }
\\[.4em] &
  \pkcurly[\big]{
    \tfrac{1}{\sqrt2} \bigoplusn^{\hvar}
      (-1)^{a \hvar} \cdot
      \qit \mapsto \ket{a} \,*\,
      \qitB \mapsto \gH \ket{b \xor \hvar} \,*\, \qitC \mapsto \gH \ket{b}
  }^{\hvarB, \hvarC}
\\[.4em] &
  \pkcurly[\big]{
    \tfrac{1}{2} \bigoplusn^{\hvar}
      \sumn_d (-1)^{a \hvar + (b \xor \hvar) d} \cdot
      (\qit, \qitB) \mapsto \ket{a d} \,*\,
      \qitC \mapsto \gH \ket{b}
  }^{\hvarB, \hvarC}
\\[.4em] &
\gMZZ^\hvarB[\qit, \qitB]
\\[.4em] &
  \pkcurly[\big]{
    \tfrac{1}{2} \bigoplusn^{\hvar, \hvarB} \sumn_{d}
      \delta_{\hvarB, a \xor d} \cdot (-1)^{a \hvar + (b \xor \hvar) d} \cdot
      (\qit, \qitB) \mapsto \ket{a d} \,*\, \qitC \mapsto \gH \ket{b}
  }^\hvarC
\\[.4em] &
  \pkcurly[\big]{
    \tfrac{1}{2} \bigoplusn^{\hvar, \hvarB}
      (-1)^{a \hvar + (b \xor \hvar) (a \xor \hvarB)} \cdot
      (\qit, \qitB) \mapsto \ket{a\pk1 (a \xor \hvarB)} \,*\,
      \qitC \mapsto \gH \ket{b}
  }^\hvarC
  \hspace{1.5em} (d \coloneq a \xor \hvarB)
\\[.4em] &
\ifonly{ \hvarB }{ \gX[\qitC] }
\\[.4em] &
  \pkcurly[\big]{
    \tfrac{1}{2} \bigoplusn^{\hvar, \hvarB}
      (-1)^{a \hvar + (b \xor \hvar) (a \xor \hvarB) + b \hvarB} \cdot
      (\qit, \qitB) \mapsto \ket{a\pk1 (a \xor \hvarB)} \,*\,
      \qitC \mapsto \gH \ket{b}
  }^\hvarC
\\[.4em] &
  \pkcurly[\big]{
    \tfrac{1}{2} \bigoplusn^{\hvar, \hvarB}
      (-1)^{a b + \hvar \hvarB} \cdot
      (\qit, \qitB) \mapsto \ket{a\pk1 (a \xor \hvarB)} \,*\,
      \qitC \mapsto \gH \ket{b}
  }^\hvarC
\\[.4em] &
\gH[\qitB]
\\[.4em] &
  \pkcurly[\big]{
    \tfrac{1}{2} \bigoplusn^{\hvar, \hvarB}
      (-1)^{a b + \hvar \hvarB} \cdot
      \qit \mapsto \ket{a} \,*\, \qitB \mapsto \gH \ket{a \xor \hvarB} \,*\,
      \qitC \mapsto \gH \ket{b}
  }^\hvarC
\\[.4em] &
\gMZ^\hvarC[\qitB]
\\[.4em] &
  \pkcurly[\big]{
    \tfrac{1}{2 \sqrt2} \bigoplusn^{\hvar, \hvarB, \hvarC}
      (-1)^{a b + \hvar \hvarB + (a \xor \hvarB) \hvarC} \cdot
      (\qit, \qitB) \mapsto \ket{a \hvarC} \,*\,
      \qitC \mapsto \gH \ket{b}
  }
\\[.4em] &
\ifonly{ \hvarC }{ \gZ[\qit] }
\\[.4em] &
  \pkcurly[\big]{
    \tfrac{1}{2 \sqrt2} \bigoplusn^{\hvar, \hvarB, \hvarC}
      (-1)^{a b + \hvar \hvarB + (a \xor \hvarB) \hvarC + a \hvarC} \cdot
      (\qit, \qitB) \mapsto \ket{a \hvarC} \,*\,
      \qitC \mapsto \gH \ket{b}
  }
\\[.4em] &
  \pkcurly[\big]{
    \tfrac{1}{2 \sqrt2} \bigoplusn^{\hvar, \hvarB, \hvarC}
      (-1)^{a b + \hvar \hvarB + \hvarB \hvarC} \cdot
      (\qit, \qitB) \mapsto \ket{a \hvarC} \,*\,
      \qitC \mapsto \gH \ket{b}
  }
\\[.4em] &
  \pkcurly[\Big]{
    (-1)^{a b} \tfrac{1}{2 \sqrt2} \cdot
    (\qit, \qitC) \mapsto \ket{a} \otimes \gH \ket{b} \,*\,
    \bigoplusn^{\hvar, \hvarB, \hvarC}
      (-1)^{\hvar \hvarB + \hvarB \hvarC} \cdot \qitB \mapsto \ket{\hvarC}
  }
\\[.4em] &
  \pkcurly[\Big]{
    (\qit, \qitC) \mapsto \gCX\pk2 (\ket{a} \otimes \gH \ket{b}) \,*\,
    \tfrac{1}{2 \sqrt2} \bigoplusn^{\hvar, \hvarB, \hvarC}
      (-1)^{\hvar \hvarB + \hvarB \hvarC} \cdot \qitB \mapsto \ket{\hvarC}
  }.
\end{align*}
Although the above proof is a bit long, each step is straightforward.
Note that at the second last step, we used the unframe rule \ref{lrule:bigbmix-unframe}.
This proves that the program behaves as the \gCNOT{} gate for the initial state $(\qit, \qitC) \mapsto \ket{a} \otimes \gH \ket{b}$.
Finally, we can set $\Prop \defeq \tfrac{1}{2 \sqrt2} \bigoplusn^{\hvar, \hvarB, \hvarC}
  (-1)^{\hvar \hvarB + \hvarB \hvarC} \cdot \qitB \mapsto \ket{\hvarC}$
(which satisfies $\frameable$ and $\Prob 1$)
and conclude the proof with the following derivation, generalizing the assertion to any initial state of $\qit$ and $\qitC$ using \ref{lrule:hoare-sum} and \ref{lrule:hoare-scale}:
\begin{derivation}
  \infer*[Right = \ref{lrule:hoare-sum}]{
    \infer*[Right = \ref{lrule:hoare-scale}]{
      \all{a, s}
      \hoare{
        (\qit, \qitC) \mapsto \ket{a s} *
        \qitB \mapsto \ket{0}
      }[^{\hvar, \hvarB, \hvarC}]
        { \mCNOT^{\hvar, \hvarB, \hvarC}[\qit, \qitB, \qitC] }
        { (\qit, \qitC) \mapsto \gCX \ket{a s} \pk1*\pk1 \Prop }
    }{
      \all{a, s}
      \hoare{
        (\qit, \qitC) \mapsto \alpha_{a, s} \ket{a s} *
        \qitB \mapsto \ket{0}
      }[^{\hvar, \hvarB, \hvarC}]
        { \mCNOT^{\hvar, \hvarB, \hvarC}[\qit, \qitB, \qitC] }
        { (\qit, \qitC) \mapsto \alpha_{a, s}\pk1 \gCX \ket{a s} \pk1*\pk1 \Prop }
    }
  }{
    \hoare{
      (\qit, \qitC) \mapsto \qket * \qitB \mapsto \ket{0}
    }[^{\hvar, \hvarB, \hvarC}]
      { \mCNOT^{\hvar, \hvarB, \hvarC}[\qit, \qitB, \qitC] }
      { (\qit, \qitC) \mapsto \gCX \qket \pk1*\pk1 \Prop }
  }
\end{derivation}
Here, $a$ and $s$ range over $\{0, 1\}$ and $\{\plus, \minus\}$, respectively.
This works because we can decompose any vector $\qket$ into $\sum_{a, s} \alpha_{a, s} \ket{a s}$ for some coefficients $\alpha_{a, s} \in \CC$.

\subsection{Error Correction: Bit-Flip Code}
\label{app:sect:cases:bit-flip-code}

Here we give the details of \cref{sect:cases:bit-flip-code}.

Fault-tolerant quantum computation aims to protect quantum information from errors by introducing redundancy, and the error correction codes are a key component of this approach.
For example, instead of representing a qubit state $\ket{i}$ (for $i \in \{0, 1\}$) using a single qubit,
we can encode it as $\ket{i i i}$ using three qubits.
This duplication enables error detection and correction if one of the qubits flips due to noise.

More generally, such redundancy means that an abstract single-qubit state---referred to as a \emph{logical qubit}---is encoded across multiple \emph{physical qubits}.
Such an encoding is called an \emph{error correction code} and allows the system
to tolerate certain classes of errors without compromising the logical information.

One of the simplest quantum error correction codes we consider in this section is the three-qubit bit-flip code.
It encodes a logical qubit $\alpha \ket{0} + \beta \ket{1}$ as $\alpha \ket{000} + \beta \ket{111}$,
and protects against a single bit-flip error,
which is an error that flips the state of a qubit
from $\ket{0}$ to $\ket{1}$ or vice versa.
The error correction procedure can be expressed as the circuit in \cref{fig:bitflip-code}.
Formally, $\BitEC$ can be defined as the following program:
\[\begin{aligned}
  \BitEC^{\var, \varB}[\qit, \qitB, \qitC] \ \ \defeq\ \ \ &
    \gCX[\qitB, \qit];\, \gMZ^\var[\qit];\, \gCX[\qitB, \qit];\,
    \gCX[\qitC, \qitB];\, \gMZ^\varB[\qitB];\, \gCX[\qitC, \qitB];
\\ &
    \ifonly{ \var \land \neg \varB }{ \gX[\qit] };\,
    \ifonly{ \var \land \varB }{ \gX[\qitB] };\,
    \ifonly{ \neg \var \land \varB }{ \gX[\qitC] }.
\end{aligned}\]
The first line applies the \gCNOT{} gates and measurements are \emph{error detection} steps.
During this phase, we measure the effect of the errors, called the \emph{syndrome}, by measuring the qubits $\qitB$ and $\qitC$.
The second line applies corrective operations based on the measurement results to \emph{recover} the original logical qubit state.
The correctness of this error-correction procedure can be specified as the following assertion, which holds for any $e_1, e_2, e_3 \in \{0, 1\}$ satisfying $e_1 + e_2 + e_3 \leq 1$:
\[\begin{aligned}
  \all{\alpha, \beta} \hspace{1em} &
  \pkcurly[\big]{\pk2
     (\qit, \qitB, \qitC) \mapsto
      (\gX^{e_1} \gX^{e_2} \gX^{e_3}) (\alpha \ket{000} + \beta \ket{111})
  \pk2}^{\var, \varB}
\\[-.2em] & \hspace{2em}
    \BitEC^{\var, \varB}[\qit, \qitB, \qitC] \hspace{1em}
    \pkcurly[\big]{\pk2
      (\qit, \qitB, \qitC) \mapsto \paren{\alpha \ket{000} + \beta \ket{111}} \,*\,
      \Prop^{\,\BitEC^{\var, \varB}}_{e_1,e_2,e_3}
    \pk2}.
\end{aligned}\]
Here, $\Prop^{\,\BitEC^{\var, \varB}}_{e_1,e_2,e_3}$ is some SL assertion satisfying $\frameable$ and $\Prob 1$.

We can prove it by the following derivation for $i = 0, 1$:
\begin{align*} &
  \pkcurly[\big]{
    \qit \mapsto \ket{i \xor e_1} * \qitB \mapsto \ket{i \xor e_2} *
    \qitC \mapsto \ket{i \xor e_3}
  }^{\var, \varB}
\\[.3em] &
  \gCX[\qit, \qitB]
  \ \ \pkcurly[\big]{
    \qit \mapsto \ket{i \xor e_1} * \qitB \mapsto \ket{e_1 \xor e_2} *
    \qitC \mapsto \ket{i \xor e_3}
  }^{\var, \varB}
\\[.3em] &
  \gMZ^\var[\qitB]
  \ \ \pkcurly[\big]{
    \bigoplusn^\var \delta_{\var, e_1 \xor e_2} \cdot
      \qit \mapsto \ket{i \xor e_1} * \qitB \mapsto \ket{e_1 \xor e_2} *
      \qitC \mapsto \ket{i \xor e_3}
  }^{\var, \varB}
\\[.3em] &
  \gCX[\qit, \qitB]
  \ \ \pkcurly[\big]{
    \bigoplusn^\var \delta_{\var, e_1 \xor e_2} \cdot
      \qit \mapsto \ket{i \xor e_1} * \qitB \mapsto \ket{i \xor e_2} *
      \qitC \mapsto \ket{i \xor e_3}
  }^\var
\\[.3em] &
  \gCX[\qitB, \qitC];\
  \gMZ^\varB[\qitC];\ \gCX[\qitB, \qitC]
\\[-.1em] &
  \pkcurly[\big]{
    \bigoplusn^{\var, \varB}
      \delta_{\var, e_1 \xor e_2} \delta_{\varB, e_2 \xor e_3} \cdot
      \qit \mapsto \ket{i \xor e_1} * \qitB \mapsto \ket{i \xor e_2} *
      \qitC \mapsto \ket{i \xor e_3}
  }
\\[.4em] &
  \ifonly{ \var \land \neg \varB }{ \gX[\qit] };\
  \ifonly{ \var \land \varB } {\gX[\qitB] };\
  \ifonly{ \neg \var \land \varB }{ \gX[\qitC] };\
\\[-.1em] &
  \pkcurly[\big]{
    \bigoplusn^{\var, \varB}
      \delta_{\var, e_1 \xor e_2} \delta_{\varB, e_2 \xor e_3} \cdot
      \qit \mapsto \ket{i \xor e_1 \xor \var \nk2\cdot\nk2 \neg \varB} *
      \qitB \mapsto \ket{i \xor e_2 \xor \var \varB} *
      \qitC \mapsto \ket{i \xor e_3 \xor \neg \var \nk2\cdot\nk2 \varB}
  }
\\[.4em] &
  \big\{\pk2
    \Prop^{\,\BitEC^{\var, \varB}}_{e_1,e_2,e_3} \,*\,
    \qit \mapsto
      \ket{i \xor e_1 \xor (e_1 \xor e_2) (1 \xor e_2 \xor e_3)} \,*\,
    \qitB \mapsto
      \ket{i \xor e_2 \xor (e_1 \xor e_2) (e_2 \xor e_3)} \,*\,
\\[-.1em] & \hspace{3em}
    \qitC \mapsto
      \ket{i \xor e_3 \xor (1 \xor e_1 \xor e_2) (e_2 \xor e_3)}
  \pk2\big\}.
\end{align*}
Here, we let $\Prop^{\,\BitEC^{\var, \varB}}_{e_1,e_2,e_3} \,\defeq\, \bigoplusn^{\var, \varB} \delta_{\var, e_1 \xor e_2} \delta_{\varB, e_2 \xor e_3}$, which does not depend on $i$.
Thanks to $e_1 + e_2 + e_3 \le 1$, with an easy bit manipulation we can show the following (we use $e_i^2 = e_i$ and $e_i e_j = 0$ under $i \ne j$):
\begin{gather*}
  e_1 \xor (e_1 \xor e_2) (1 \xor e_2 \xor e_3)
  = e_1 \xor e_1 \xor e_2 \xor e_2
  = 0
\\
  e_2 \xor (e_1 \xor e_2) (e_2 \xor e_3)
  = e_2 \xor e_2
  = 0
\\
  e_3 \xor (1 \xor e_1 \xor e_2) (e_2 \xor e_3)
  = e_3 \xor e_2 \xor e_3 \xor e_2
  = 0.
\end{gather*}
Therefore, the postcondition can be simplified to $\Prop^{\,\BitEC^{\var, \varB}}_{e_1,e_2,e_3} \pk1*\pk1 (\qit, \qitB, \qitC) \mapsto \ket{i i i}$.

Finally, we can use \ref{lrule:hoare-sum} and \ref{lrule:hoare-scale} to generalize this proof to any initial state $\alpha \ket{000} + \beta \ket{111}$
(we here abbreviate $\BitEC^{\var, \varB}[\qit, \qitB, \qitC]$ as $\BitEC$ and $\Prop^{\,\BitEC^{\var, \varB}}_{e_1,e_2,e_3}$ as $\Prop^{\,\BitEC}$):
\begin{derivation}
  \infer*[Right = \ref{lrule:hoare-sum}]{
    \infer*[Right = \ref{lrule:hoare-scale}]{
      \all{i \in \{0, 1\}}\,
      \hoare{
        (\qit, \qitB, \qitC) \mapsto (\gX^{e_1} \gX^{e_2} \gX^{e_3}) \ket{i i i}
      }[^{\var, \varB}]
        { \BitEC }
        { \Prop^{\,\BitEC} \pk2*\pk2 (\qit, \qitB, \qitC) \mapsto \ket{i i i} }
    }{
      \hoare{
        (\qit, \qitB, \qitC) \mapsto (\gX^{e_1} \gX^{e_2} \gX^{e_3}) (\alpha \ket{000})
      }[^{\var, \varB}]
        { \BitEC }
        { \Prop^{\,\BitEC} \pk2*\pk2 (\qit, \qitB, \qitC) \mapsto \alpha \ket{000} }
    \\
      \hoare{
        (\qit, \qitB, \qitC) \mapsto (\gX^{e_1} \gX^{e_2} \gX^{e_3}) (\beta \ket{111})
      }[^{\var, \varB}]
        { \BitEC }
        { \Prop^{\,\BitEC} \pk2*\pk2 (\qit, \qitB, \qitC) \mapsto \beta \ket{111} }
    }}{
    \hoare{
      (\qit, \qitB, \qitC) \mapsto (\gX^{e_1} \gX^{e_2} \gX^{e_3}) (\alpha \ket{000} + \beta \ket{111})
    }[^{\var, \varB}]
      { \BitEC }
      {
        \Prop^{\,\BitEC} \pk2*\pk2
        (\qit, \qitB, \qitC) \mapsto \alpha \ket{000} + \beta \ket{111}
      }
  }
\end{derivation}

Our proof of the correctness of the bit-flip code has some similarities with verification of quantum error correction codes using symbolic execution as in \cite{WangY24-symbolic-execution},
though our approach is fundamentally different.
Their approach uses density matrices and stabilizer codes, while our approach uses vector spaces and separation logic.

There are other kinds of errors that can occur in quantum computation, such as \emph{phase-flip errors}.
A phase-flip error negates the relative phase between $\ket{0}$ and $\ket{1}$, \ie applies the $\gZ$ gate, turning $\ket{1}$ to $- \ket{1}$ but $\ket{0}$ to $\ket{0}$.
In other words, it transforms a qubit state $\ket{\plus}$ to $\ket{\minus}$ and vice versa.
In order to correct such phase-flip errors, we need a different error correction code.
However, this can be done easily by using the bit-flip code in a clever way.
Since a bit-flip code encodes a logical qubit $\ket{0}$ as $\ket{000}$ and $\ket{1}$ as $\ket{111}$ to protect against a bit-flip error,
we can instead take three copies with respect to the $\gX$ basis $\angl{\ket{\plus}, \ket{\minus}}$ rather than the $\gZ$ basis $\angl{\ket{0}, \ket{1}}$, \ie encode a qubit $\ket{\plus}$ as $\ket{\plus\plus\plus}$ and $\ket{\minus}$ as $\ket{\minus\minus\minus}$.
As the Hadamard gate $\gH$ transforms the $\gX$ basis to the $\gZ$ basis and vice versa,
we can derive the phase-flip code error correction procedure $\PhaseEC$ simply by putting $\gH$ gates before and after the phase-flip code error correction procedure $\BitEC$.

The phase-flip code error correction procedure $\PhaseEC$ can be derived from $\BitEC$ as follows:
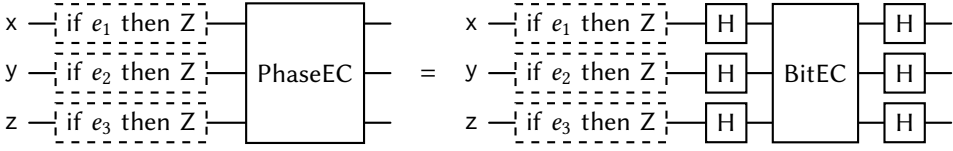
\[
  \begin{quantikz}[row sep = 2pt, column sep = 10pt]
    \lstick{\qit}  & \gate[style=error]{\ifonly{e_1}{\gZ}} &[.5em]
    \gate[3]{\PhaseEC} &
  \\
    \lstick{\qitB} & \gate[style=error]{\ifonly{e_2}{\gZ}} & &
  \\
    \lstick{\qitC} & \gate[style=error]{\ifonly{e_3}{\gZ}} & &
  \end{quantikz}
  \ \ =\ \
  \begin{quantikz}[row sep = 2pt, column sep = 10pt]
    \lstick{\qit}  & \gate[style=error]{\ifonly{e_1}{\gZ}} &[.5em]
    \gate{\gH} & \gate[3]{\BitEC} & \gate{\gH} &
  \\
    \lstick{\qitB} & \gate[style=error]{\ifonly{e_2}{\gZ}} & \gate{\gH} &                  & \gate{\gH} &
  \\
    \lstick{\qitC} & \gate[style=error]{\ifonly{e_3}{\gZ}} & \gate{\gH} &                  & \gate{\gH} &
  \end{quantikz}
\]
Formally, $\PhaseEC$ is defined as the following program:
\[
  \PhaseEC^{\var, \varB}[\qit, \qitB, \qitC] \ \ \defeq\ \
    \gH[\qit];\, \gH[\qitB];\, \gH[\qitC];\,
    \BitEC^{\var, \varB}[\qit, \qitB, \qitC];\,
    \gH[\qit];\, \gH[\qitB];\, \gH[\qitC].
\]
We can immediately derive the following specification of the procedure $\PhaseEC$ from the specification of $\BitEC$, for any $e_1, e_2, e_3 \in \{0, 1\}$ such that $e_1 + e_2 + e_3 \le 1$:
\[\begin{aligned}
  \all{\alpha, \beta}\hspace{1em} &
  \pkcurly[\big]{\pk2
     (\qit, \qitB, \qitC) \mapsto
      \gZ^{e_1} \gZ^{e_2} \gZ^{e_3} (\alpha \ket{\plus\plus\plus} + \beta \ket{\minus\minus\minus})
  \pk2}^{\var, \varB}
\\[-.1em] & \hspace{2em}
    \PhaseEC^{\var, \varB}[\qit, \qitB, \qitC] \hspace{1em}
    \pkcurly[\big]{\pk2
      (\qit, \qitB, \qitC) \mapsto \paren{\alpha \ket{\plus\plus\plus} + \beta \ket{\minus\minus\minus}} \,*\,
      \Prop^{\,\BitEC^{\var, \varB}}_{e_1,e_2,e_3}
    \pk2}.
\end{aligned}\]

\subsection{Error Correction: Shor Code}
\label{app:sect:cases:shor-code}

Here we give the details of \cref{sect:cases:shor-code}.

Now, we can combine the bit-flip code and the phase-flip code to construct a universal quantum error correction code capable of correcting any single-qubit error.
This construction is known as the \emph{Shor code}, which encodes a single logical qubit into 9 physical qubits.
The Shor code achieves this by nesting two types of codes: we first apply a phase-flip code to the logical qubit, and then apply a bit-flip code to each of the resulting qubits.
Concretely, the Shor code encodes the logical qubit state $\alpha \ket{0} + \beta \ket{1}$ into
\[
  \alpha \ket{0_\Log 0_\Log 0_\Log} + \beta \ket{1_\Log 1_\Log 1_\Log},
\]
where $\ket{0_\Log}$ and $\ket{1_\Log}$ are phase-flip-encoded states.
In this way, the Shor code uses a bit-flip code over logical qubits that have already been protected against phase-flip errors, ensuring robustness against both bit-flip and phase-flip errors.

Such codes that can correct both bit-flip and phase-flip errors are universal in the sense that they can correct any single-qubit unitary error.
The sketch of the intuitive proof of this is as follows.
Because the Pauli gates $\gid$, $\gX$, $\gZ$ and $\gY = - \sqrt{-1}\pk2 \gZ \gX$ form an orthogonal basis for the whole linear space of complex $2 \times 2$ matrices,
any single-qubit unitary error $\Unitary$ can be decomposed into a linear combination of $\gid$, $\gX$, $\gZ$ and $\gZ \gX$, with a coefficient vector of norm $1$.
So we can write $\Unitary = \lambda_0 \gid + \lambda_1 \gX + \lambda_2 \gZ + \lambda_3 \gZ \gX$ for some $(\lambda_i)_{i=0}^3$ such that $\sum_{i=0}^3 \abs{\lambda_i}^2 = 1$.
Note that a state $\qket$ with the error $\Unitary$ applied, $\Unitary \qket$, can be expressed as
\[
  \lambda_0 \qket + \lambda_1 \gX \qket + \lambda_2 \gZ \qket + \lambda_3 \gZ \gX \qket.
\]
By linearity of quantum computation,
it suffices to show that the matrices $\gX$, $\gZ$ and $\gZ \gX$, can be corrected by the Shor code.
Perhaps surprisingly, this is immediate,
because $\PhaseEC$ and $\BitEC_\Log$ are designed to correct $\gX$ and $\gZ$ errors, respectively.

\begin{figure*}
  \begin{quantikz}[row sep = 2pt]
    & \gate[style = error]{ \ifonly{e_1}{\Unitary} } &
    \gate[3]{ \PhaseEC } & \gate[9]{ \BitEC_\Log } &
  \\
    & \gate[style = error]{ \ifonly{e_2}{\Unitary} } & & &
  \\
    & \gate[style = error]{ \ifonly{e_3}{\Unitary} } & & &
  \\
    & \gate[style = error]{ \ifonly{e_4}{\Unitary} } &
    \gate[3]{ \PhaseEC } & &
  \\
    & \gate[style = error]{ \ifonly{e_5}{\Unitary} } & & &
  \\
    & \gate[style = error]{ \ifonly{e_6}{\Unitary} } & & &
  \\
    & \gate[style = error]{ \ifonly{e_7}{\Unitary} } &
    \gate[3]{ \PhaseEC } & &
  \\
    & \gate[style = error]{ \ifonly{e_8}{\Unitary} } & & &
  \\
    & \gate[style = error]{ \ifonly{e_9}{\Unitary} } & & &
  \end{quantikz}
  \caption{Error correction circuit for the Shor code.}
  \label{fig:shor-code}
  \Description{}
\end{figure*}
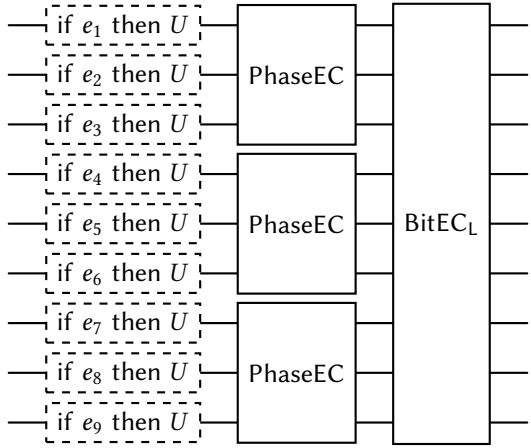

The Shor code is a perfect example to demonstrate the power of our separation logic.
The error correction circuit $\ShorEC$ for the Shor code is given in \cref{fig:shor-code}
and the formal definition has been given in \cref{sect:cases:shor-code}.
As we can see, the whole circuit has 9 qubits, meaning the dimension of the Hilbert space is $2^9 = 512$.
However, it can be decomposed into smaller modules,
which can be verified independently using separation logic.
This not only makes the verification process more manageable,
it also lets us use our linear-combination rule to simplify the proof,
since the universality of the Shor code is based on the linearity of quantum computation as sketched above.

Let us first explain how the last component $\BitEC_\Log$ is constructed.
This program is essentially the same as the 3-qubit bit-flip code $\BitEC$ we have already verified,
but with each physical qubit replaced by a logical qubit encoded by the phase-flip code.
Since we used \gX{} and \gCNOT{} gates and a measurement to implement the bit-flip code,
we first implement these gates for the phase-flip code.
That is, we verify the following assertions for programs $\gX_\Log$,
$\gCX_\Log$ and $\gMZ^\Log$ to use them as building blocks for the error correction of the bit-flip code.
\begin{align*}
  \ex{\, \PropB \col \frameable, \Prob 1}\
  \all{i \in \{0, 1\}}\ &
  \hoare{ \bar{\qit} \mapsto \ket{i_\Log} }
    { \gX_\Log[\bar{\qit}] }
    { \bar{\qit} \mapsto \ket{\neg i_\Log} * \PropB }
\\[.1em]
  \ex{\, \PropB \col \frameable, \Prob 1}\
  \all{i, j \in \{0, 1\}}\ &
  \hoare{ (\bar{\qit}, \bar{\qitB}) \mapsto \ket{i_\Log j_\Log} }
    { \gCX_\Log[\bar{\qit}, \bar{\qitB}] }
    { (\bar{\qit}, \bar{\qitB}) \mapsto
      \ket{i_\Log (i \xor\nk1 j)_\Log}
      \pk1*\pk1 \PropB }
\\[.1em]
  \ex{\, \PropB \col \frameable, \Prob 1}\ &
  \hoare{ \bar{\qit} \mapsto \ket{0_\Log} }[^\var]
    { \gMZ^{\Log\pk1 \var}[\bar{\qit}] }
    {
      \paren[\big]{\pk1
        \bar{\qit} \mapsto \ket{0_\Log} \,\mix{\var}\,
        \bar{\qit} \mapsto 0
      \pk1} \pk2*\pk2 \PropB
    }
\\[.1em]
  \ex{\, \PropB \col \frameable, \Prob 1}\ &
  \hoare{ \bar{\qit} \mapsto \ket{1_\Log} }[^\var]
    { \gMZ^{\Log\pk1 \var}[\bar{\qit}] }
    {
      \paren[\big]{\pk1
        \bar{\qit} \mapsto 0 \,\mix{\var}\,
        \bar{\qit} \mapsto \ket{0_\Log}
      \pk1} \pk2*\pk2 \PropB
    }
\end{align*}
Notably, we can abstract the factor $\PropB$ by the condition $\PropB \col \frameable, \Prob 1$.
Then, we can replace the basic gates $\gX$, $\gCX$ and $\gMZ$ in $\BitEC$
with the above programs to implement $\BitEC_\Log$.
More concretely, $\BitEC_\Log$ is defined as follows:
\[\begin{aligned}
  \BitEC_\Log^{\var, \varB}[\bar{\qit}, \bar{\qitB}, \bar{\qitC}] \ \ \defeq\ \ \ &
    \gCX_\Log[\bar{\qitB}, \bar{\qit}];\,
    \gMZ^{\Log\pk1 \var}[\bar{\qit}];\, \gCX_\Log[\bar{\qitB}, \bar{\qit}];\,
    \gCX_\Log[\bar{\qitC}, \bar{\qitB}];\,
    \gMZ^{\Log\pk1 \varB}[\bar{\qitB}];\, \gCX_\Log[\bar{\qitC}, \bar{\qitB}];
\\ &
    \ifonly{ \var \land \neg \varB }{ \gX_\Log[\bar{\qit}] };\,
    \ifonly{ \var \land \varB }{ \gX_\Log[\bar{\qitB}] };\,
    \ifonly{ \neg \var \land \varB }{ \gX_\Log[\bar{\qitC}] }.
\end{aligned}\]
We can also verify the correctness of $\BitEC_\Log$
exactly as we have done for the normal bit-flip code:
\begin{align*}
  &
  \ex{\, \PropB \col \frameable, \Prob 1}\,\
  \all{e_1, e_2, e_3 \st e_1 + e_2 + e_3 \le 1}\,\
  \all{\alpha, \beta} \hspace{1em}
\\[-.2em] & \hspace{2em}
  \pkcurly[\big]{\pk2
     (\bar{\qit}, \bar{\qitB}, \bar{\qitC}) \mapsto
      (\gX_\Log^{e_1} \gX_\Log^{e_2} \gX_\Log^{e_3})
      (\alpha \ket{0_\Log0_\Log0_\Log} + \beta \ket{1_\Log1_\Log1_\Log})
  \pk2}^{\var, \varB}
\\[-.2em] & \hspace{4em}
    \BitEC_\Log^{\var, \varB}[\bar{\qit}, \bar{\qitB}, \bar{\qitC}] \hspace{1em}
    \pkcurly[\big]{\pk2
      (\bar{\qit}, \bar{\qitB}, \bar{\qitC}) \mapsto
      \paren{\alpha \ket{0_\Log0_\Log0_\Log} + \beta \ket{1_\Log1_\Log1_\Log}} \,*\,
      \Prop^{\,\BitEC^{\var, \varB}}_{e_1,e_2,e_3} \,*\,
      \PropB
    \pk2}.
\end{align*}
This procedure of lifting the proof of correctness of the 3-qubit bit-flip code
to the 9-qubit version $\BitEC_\Log$ is not completely automatic.
However, we can easily redo the proof of correctness of a generalized
bit-flip code implemented on top of any logical qubit encoded by an arbitrary code.
This allows us to reuse the specification of the bit-flip code
as a building block in the verification of any program that uses it.
In this way, we can still avoid verifying the correctness of the bit-flip code
every time we use it.

Finally, we verify the correctness of the Shor code.
The specification of the correctness is similar to the ones we have seen (we omit classical variables for simplicity):
\begin{align*} \textstyle
  \ex{\, \Prop \col \frameable, \Prob 1}\
  \all{\alpha, \beta} \hspace{.8em}
  & \pkcurly[\big]{\,
    (\bar{\qit}, \bar{\qitB}, \bar{\qitC}) \mapsto
    (\Unitary^{e_1} \otimes \cdots \otimes \Unitary^{e_9})
    \paren*{ \alpha \ket{0_\Log 0_\Log 0_\Log} + \beta \ket{1_\Log 1_\Log 1_\Log} }
  \,}
\\ & \hspace{1.4em}
  \ShorEC[\bar{\qit}, \bar{\qitB}, \bar{\qitC}] \hspace{.8em}
  \pkcurly[\big]{\,
    (\bar{\qit}, \bar{\qitB}, \bar{\qitC}) \mapsto
      \paren{ \alpha \ket{0_\Log 0_\Log 0_\Log} + \beta \ket{1_\Log 1_\Log 1_\Log} }
    \,*\, \Prop
  \,}.
\end{align*}
This holds for any single-qubit unitary error $\Unitary$ and any $e_1, \ldots, e_9 \in \{0, 1\}$ under $\sum_{i = 1}^9 e_i = 1$.\footnote{
  We can safely exclude the case where $e_1 = \cdots = e_9 = 0$, because we can set $\Unitary = \id$ to model the error-free situation.
}

To simplify the proof, we first assume the case
where $\Unitary = \gZ^a \gX^b$ for $a, b \in \{0, 1\}$.
We can derive the following on $\PhaseEC[\bar{\qit}]$.
\begin{derivation}
  \infer*[Right = \ref{lrule:hoare-sum}]{
    \infer*[Right = \ref{lrule:hoare-frame}]{ \textstyle
      \all{c \in \{0, 1\}}\
      \pkcurly[\big]{
        \bar{\qit} \mapsto
          (\bigotimes_{i=1}^3 \gZ^{a e_i} \gX^{b e_i})
          \ket{c_\Log c_\Log c_\Log}
      }
      \hspace{1em} \PhaseEC[\bar{\qit}]
    \\ \hspace{14em} \textstyle
      \pkcurly[\big]{\,
        \bar{\qit} \mapsto
          (\bigotimes_{i=1}^3 \gX^{b e_i})
          \ket{c_\Log c_\Log c_\Log}
        \,*\, \Prop^{\,\BitEC}_{ae_1,ae_2,ae_3}
      \,}
    }{ \textstyle
      \all{c \in \{0, 1\}}\
      \pkcurly[\big]{
        (\bar{\qit}, \bar{\qitB}, \bar{\qitC}) \mapsto
          (\bigotimes_{i=1}^9 \gZ^{a e_i} \gX^{b e_i})
          \ket{c_\Log c_\Log c_\Log}
      }
      \hspace{1em} \PhaseEC[\bar{\qit}]
    \\ \hspace{9em} \textstyle
      \pkcurly[\big]{\,
        (\bar{\qit}, \bar{\qitB}, \bar{\qitC}) \mapsto
          (\bigotimes_{i=1}^3 \gX^{b e_i} \otimes
            \bigotimes_{i=4}^9 \gZ^{a e_i} \gX^{b e_i})
          \ket{c_\Log c_\Log c_\Log}
        \,*\, \Prop^{\,\BitEC}_{ae_1,ae_2,ae_3}
      \,}
    }
  }{ \textstyle
    \pkcurly[\big]{
      (\bar{\qit}, \bar{\qitB}, \bar{\qitC}) \mapsto
        (\bigotimes_{i=1}^9 \gZ^{a e_i} \gX^{b e_i})
        (\alpha \ket{0_\Log 0_\Log 0_\Log} + \beta \ket{1_\Log 1_\Log 1_\Log})
    }
    \hspace{1em} \PhaseEC[\bar{\qit}]
  \\ \hspace{3em} \textstyle
    \pkcurly[\big]{\,
      (\bar{\qit}, \bar{\qitB}, \bar{\qitC}) \mapsto
        (\bigotimes_{i=1}^3 \gX^{b e_i} \otimes
          \bigotimes_{i=4}^9 \gZ^{a e_i} \gX^{b e_i})
        (\alpha \ket{0_\Log 0_\Log 0_\Log} + \beta \ket{1_\Log 1_\Log 1_\Log})
      \,*\, \Prop^{\,\BitEC}_{ae_1,ae_2,ae_3}
    \,}
  }
\end{derivation}
Here, we used $a e_1 + a e_2 + a e_3 \le 1$.
Using this kind of derivation for $\PhaseEC$ three times and the specification of the enriched bit-flip code $\BitEC$, we can derive the following for any $\alpha, \beta \in \CC$,
letting $\qket[_\Log] \defeq \alpha \ket{0_\Log 0_\Log 0_\Log} + \beta \ket{1_\Log 1_\Log 1_\Log}$:
\begin{align*} & \textstyle
  \pkcurly[\big]{
    (\bar{\qit}, \bar{\qitB}, \bar{\qitC}) \mapsto
      (\bigotimes_{i=1}^9 \gZ^{a e_i} \gX^{b e_i}) \qket[_\Log]
  }
\\[.2em] &
  \PhaseEC[\bar{\qit}]
\\ & \textstyle
  \pkcurly[\big]{
    (\bar{\qit}, \bar{\qitB}, \bar{\qitC}) \mapsto
      (\bigotimes_{i=1}^3 \gX^{b e_i} \otimes
        \bigotimes_{i=4}^9 \gZ^{a e_i} \gX^{b e_i})
      \qket[_\Log]
    \,*\, \Prop^{(1)}_a
  }
\\[.2em] &
  \PhaseEC[\bar{\qitB}]
\\ & \textstyle
  \pkcurly[\big]{
    (\bar{\qit}, \bar{\qitB}, \bar{\qitC}) \mapsto
      (\bigotimes_{i=1}^6 \gX^{b e_i} \otimes
        \bigotimes_{i=7}^9 \gZ^{a e_i} \gX^{b e_i})
      \qket[_\Log]
    \,*\, \Prop^{(2)}_a
  }
\\[.2em] &
  \PhaseEC[\bar{\qitC}]
\\ & \textstyle
  \pkcurly[\big]{
    (\bar{\qit}, \bar{\qitB}, \bar{\qitC}) \mapsto
      (\bigotimes_{i=1}^9 \gX^{b e_i}) \qket[_\Log]
    \,*\, \Prop^{(3)}_a
  }
\\ &
  \pkcurly[\big]{
    (\bar{\qit}, \bar{\qitB}, \bar{\qitC}) \mapsto
      (
        \gX^{b(e_1 \lor e_2 \lor e_3)}_\Log \otimes
        \gX^{b(e_4 \lor e_5 \lor e_6)}_\Log \otimes
        \gX^{b(e_7 \lor e_8 \lor e_9)}_\Log
      )
      \qket[_\Log]
    \,*\, \Prop^{(3)}_{a}
  }
\\[.2em] &
  \BitEC_\Log[\bar{\qit}, \bar{\qitB}, \bar{\qitC}]
\\ &
  \pkcurly[\big]{
    (\bar{\qit}, \bar{\qitB}, \bar{\qitC}) \mapsto
      \qket[_\Log]
    \,*\, \Prop^{(4)}_{a,b}
  }
\end{align*}
Note that every $\Prop^{(k)}$ is a global phase guard that satisfies
$\Prop^{(k)} \col \frameable$, thanks to the fact that the separating conjunction of frameable assertions is frameable (\cref{fig:frameable}).
Importantly, each $\Prop^{(k)}$ depends on the error parameters $a, b$ (and each $e_i$),
but it does not depend on the initial state $\qket[_\Log]$.
Finally, we can use the linear combination rule to derive the final specification of the Shor code.
Let a single-qubit error $\Unitary$ be written as $\sum_{a, b \in \{0, 1\}} \lambda_{a, b} \gZ^a \gX^b$ for $\paren{\lambda_{a, b}}_{a, b}$ such that $\sum_{a, b \in \{0, 1\}} \abs{\lambda_{a, b}}^2 = 1$.
Then we have the following derivation (we here abbreviate $\ShorEC[\bar{\qit}, \bar{\qitB}, \bar{\qitC}]$ as $\ShorEC$):
\begin{derivation}
  \infer*[Right = \ref{lrule:hoare-sum}]{
    \infer*[Right = \ref{lrule:hoare-scale}]{ \textstyle
      \all{a, b}
      \hoare[\big]{
        (\bar{\qit}, \bar{\qitB}, \bar{\qitC}) \mapsto
        (\bigotimes_{i=1}^9 \gZ^{a e_i} \gX^{b e_i})
        \qket[_\Log]
      }{ \ShorEC }{
        (\bar{\qit}, \bar{\qitB}, \bar{\qitC}) \mapsto \qket[_\Log]
        \pk1*\pk1 \Prop^{(4)}_{a, b}
      }
    }{
      \all{a, b}
      \hoare[\big]{ \textstyle
        (\bar{\qit}, \bar{\qitB}, \bar{\qitC}) \mapsto
          (\bigotimes_{i=1}^9 \gZ^{a e_i} \gX^{b e_i})
          \lambda_{a, b} \qket[_\Log]
      }{ \ShorEC }{
        (\bar{\qit}, \bar{\qitB}, \bar{\qitC}) \mapsto \qket[_\Log]
        \pk1*\pk1 \lambda_{a, b} \Prop^{(4)}_{a, b}
      }
    }
  }{ \textstyle
    \hoare[\big]{
      (\bar{\qit}, \bar{\qitB}, \bar{\qitC}) \mapsto
        \sum_{a, b}
        (\bigotimes_{i=1}^9 \gZ^{a e_i} \gX^{b e_i})
        \lambda_{a, b} \qket[_\Log]
    }{ \ShorEC }{
      (\bar{\qit}, \bar{\qitB}, \bar{\qitC}) \mapsto \qket[_\Log]
      \pk1*\pk1 \sum_{a, b} \lambda_{a, b} \Prop^{(4)}_{a, b}
    }
  }
\end{derivation}
The precondition can be rewritten as
$
  \sum_{a, b}
    (\bigotimes_{i=1}^9 \gZ^{a e_i} \gX^{b e_i})
    \lambda_{a, b} \qket[_\Log] =
  (\bigotimes_{i=1}^9 \Unitary^{e_i}) \qket[_\Log]
$.
The frameability of
$\Prop \defeq \sum_{a, b} \lambda_{a, b} \Prop^{(4)}_{a, b}$
follows from the construction.

Finally, we prove $\Prop \col \Prob 1$.\footnote{
  This is quite trivial if we know that $\ShorEC$ is loop-free, because its termination probability is clearly $1$.
  However, here we prove $\Prop \col \Prob 1$ even without depending on that knowledge, leveraging orthogonality.
}
Since the coefficients $\lambda_{a, b}$ satisfy
$\sum_{a, b \in \{0, 1\}} \abs{\lambda_{a, b}}^2 = 1$,
it suffices to prove $\Prop^{(4)}_{a,b} \col \Prob 1$ and
that the assertions $\paren{ \Prop^{(4)}_{a,b} }_{a,b\in\{0,1\}}$ are orthogonal.
To begin with, $\Prop^{(4)}_{a,b}$ has the following form:
\begin{align*}
  \Prop^{\,\BitEC^{\hvarB, \hvarC}}_{e, e', e''}
  &\ \defeq\
  \textstyle\bigoplusn^{\hvarB, \hvarC} \delta_{\hvarB, e \xor e'} \delta_{\hvarC, e' \xor e''}
\\[.1em]
  \Prop^{(4)}_{a,b}
  &\ \defeq\
  \Prop^{\,\BitEC^{\hvarB_1, \hvarC_1}}_{ae_1, ae_2, ae_3}
  \,*\,
  \Prop^{\,\BitEC^{\hvarB_2, \hvarC_2}}_{ae_4, ae_5, ae_6}
  \,*\,
  \Prop^{\,\BitEC^{\hvarB_3, \hvarC_3}}_{ae_7, ae_8, ae_9}
  \,*\,
  \Prop^{\,\BitEC^{\hvarB_4, \hvarC_4}}_{
    b(e_1 \lor e_2 \lor e_3),
    b(e_4 \lor e_5 \lor e_6),
    b(e_7 \lor e_8 \lor e_9)
  }
  \,*\,
  \PropB.
\end{align*}
Here, $\PropB$ is some proposition satisfying $\frameable$ and $\Prob 1$.
We can easily prove $\Prop^{(4)}_{a,b} \col \Prob 1$.
Now we prove the orthogonality of $\paren{ \Prop^{(4)}_{a,b} }_{a,b\in\{0,1\}}$.
We can easily prove the orthogonality
$\Prop^{\,\BitEC^{\hvarB, \hvarC}}_{\bar{e}}
\orth \Prop^{\,\BitEC^{\hvarB, \hvarC}}_{\bar{e}'}$ for distinct $(\bar{e}), (\bar{e}') \in \{0,1\}^3$.
So, for any two distinct parameters $(a,b), (a',b') \in \{0,1\}^2$,
at least one of the first four conjuncts in $\Prop^{(4)}$ is orthogonal.\footnote{
  Recall that we assume $e_1 + \cdots + e_9 = 1$.
  So exactly one of $e_1, \ldots, e_9$ is non-zero.
  Also, exactly one of $e_1 \lor e_2 \lor e_3$, $e_4 \lor e_5 \lor e_6$ and $e_7 \lor e_8 \lor e_9$ is non-zero.
}
Therefore, by the rule \ref{arule:orth-frame} in \cref{app:sect:logic:inner-prod},
the set $\{ \Prop^{(4)}_{a,b} \}_{a,b\in\{0,1\}}$ is proved orthogonal.

\subsection{Probabilistic Choice and Almost Sure Termination}
\label{app:sect:cases:coin-toss}

Here we give the details of \cref{sect:cases:coin-toss}.

Recall the repeat-until-success program we considered:
\[
  \cointoss'_\prob \ \ \defeq\ \
    \whilex{ \var }{ \paren[\big]{\pk2
      \coin^\var_\prob;\, \varC \store \varC + 1
    \pk2} }
\hspace{3em}
  \cointoss_\prob \ \ \defeq\ \
    \coin^\var_\prob;\
    \cointoss'_\prob.
\]
We can prove the following specification for any probability $\prob \in \opcl{0, 1}$, guaranteeing almost sure termination:
\[
  \ex{\, \Prop \col \frameable, \Prob 1}\
  \hoare{\pk2 \varC \mapsto 0 \pk2}[^\var]
    { \cointoss_\prob }
    {\pk2 \var \mapsto 0 \pk2*\pk2 \Prop \pk2}.
\]

Let $\Prop_n \,\defeq\, \sqrt{(1 - \prob)^n} \pk1\cdot\pk1 (\var \mapsto 0 \oplus_\prob \var \mapsto 1) \pk1*\pk1 \varC \mapsto n$,
$\PropB_n \,\defeq\, \sqrt{(1 - \prob)^n \prob} \pk1\cdot\pk1 \var \mapsto 0 \pk1*\pk1 \varC \mapsto n$,
and $\PropC_n \,\defeq\, \sqrt{(1 - \prob)^{n + 1}} \pk1\cdot\pk1 \var \mapsto 1 \pk1*\pk1 \varC \mapsto n$ for $n \in \NN$.
Also, we set $\Prop \,\defeq\, \bigoplus^\varC_{n \in \NN} \sqrt{(1 - \prob)^n \prob}$.
Then we can conclude the proof by the following derivation:
\begin{derivation}
  \hspace{-3em}
  \infer*{
    \hoare{ \varC \mapsto 0 }[^\var]{ \coin^\var_\prob }{ \Prop_0 }
    \hspace{-.8em}
  \\
    \infer*[Right = \ref{lrule:bigbmix-unframe}]{
      \infer*[Right = \ref{lrule:hoare-while}]{
        \all{n} \hoare{ \Prop_n }{ \var }
          {\pk2 \rtok{0} * \PropB_n \pk1\oplus\pk1
            \rtok{1} * \PropC_n \pk2}
        \hspace{-.8em}
      \\
        \all{n} \hoare{ \PropC_n }
          { \coin^\var_\prob;\, \varC \store \varC + 1 }
          { \Prop_{n + 1} }
      }{
        \hoare{ \Prop_0 }
          { \Prop_0 }
          { \bigoplusn^\varC_{n \in \NN} \sqrt{(1 - \prob)^n \prob}
              \pk1*\pk1 \var \mapsto 0 }
      }
    }{
      \hoare{ \Prop_0 }
        { \cointoss'_\prob }
        { \var \mapsto 0 \pk1*\pk1 \Prop }
    }
  }{
    \hoare{ \varC \mapsto 0 }[^\var]
      { \cointoss_\prob }
      { \var \mapsto 0 \pk1*\pk1 \Prop }
  }
\end{derivation}
Notably, $\Prop \col \frameable, \Prob 1$ holds.
The counter $\varC$ is crucial for the unambiguity $\Prop \col \unambig$.
The probability judgment $\Prop \col \Prob 1$ holds because $\sum_{n \in \NN} (1 - \prob)^n \prob \pk1=\pk1 1$.
 
\end{document}